%% file: hypergraph-arXiv20-06-2016.tex
\documentclass[11pt,a4paper]{article}

\usepackage{amsmath,amssymb,amsthm,hyperref,color}
\usepackage{filecontents}
\usepackage{enumerate}
\hypersetup{colorlinks=true,linkcolor=blue, citecolor=blue}

\usepackage{a4wide}
\usepackage{graphics} 
\usepackage[font={small}]{caption}
\usepackage{subcaption}
\usepackage{tikz}
\usepackage{pgfplots}
\newtheorem{theorem}{Theorem}

\newtheorem{lemma}[theorem]{Lemma}
\newtheorem{claim}[theorem]{Claim}
\newtheorem{observation}[theorem]{Observation}
\newtheorem{corollary}[theorem]{Corollary}
\newtheorem{remark}[theorem]{Remark}

\newtheorem{definition}[theorem]{Definition}
\newtheorem*{lemhardness}{Lemma \ref{hardness}}

\makeatletter
\def\prob#1#2#3{\goodbreak\begin{list}{}{\labelwidth\z@ \itemindent-\leftmargin
                        \itemsep\z@  \topsep6\p@\@plus6\p@
                        \let\makelabel\descriptionlabel}
                \item[\it Name]#1
               \item[\it Instance]                #2
                \item[\it Output]#3
                \end{list}}
\makeatother

\def\Zmu{\mu}
\def\Rplus{{\mathbb R}_+}

\def\A{A}

\def\epsilon{\varepsilon}
\def\condV{\mathrm{cond}(\mathcal{V})}
\def\zerof{f_\mathsf{zero}} 
\def\onef{f_\mathsf{one}} 
\def\allzerof{f_\mathsf{allzero}} 
\def\allonef{f_\mathsf{allone}} 
\def\eqf{f_\mathsf{EQ}} 
\def\evenf{f_\mathsf{even}} 
\def\oddf{f_\mathsf{odd}}

\def\Hyper2Spinf{\#\mathsf{Hyper2Spin}(f,\Delta,c)} 
\newcommand{\nCSP}[1]{\#{\mathsf{CSP}}(#1)}
\newcommand{\CSP}[1]{{\mathsf{CSP}}(#1)}
\newcommand{\nCSPD}[1]{\#{\mathsf{CSP}_{\Delta}}(#1)}
\newcommand{\ACSPD}[1]{\#{\mathsf{CSP}_{\Delta,c}}(#1)}

\def\eq{\mathsf{eq}}
\def\hard{\mathsf{hard}}

\def\NP{\mathsf{NP}}
\def\RP{\mathsf{RP}}

\def\EASYk{\mathsf{EASY}(k)}

\def\Band{\mathsf{AND}}
\def\Bor{\mathsf{OR}}
\def\maj{\mathsf{Maj}}
\def\zeros{\mathbf{0}}
\def\ones{\mathbf{1}}
\def\minority{\mathsf{Minority}}

\begin{document}  
\title{The complexity of approximately counting in 2-spin systems on $k$-uniform bounded-degree hypergraphs\thanks
{ An extended abstract of this paper appears in the proceedings of SODA 2016.
  The research leading to these results has received funding from the European Research Council under
  the European Union's Seventh Framework Programme (FP7/2007-2013) ERC grant agreement no.\ 334828. The paper
  reflects only the authors' views and not the views of the ERC or the European Commission.
  The European Union is not liable for any use that may be made of the information contained therein. }
}
\author{
Andreas Galanis\thanks{Department of Computer Science, University of Oxford, Wolfson Building, Parks Road, Oxford, OX1~3QD, UK. {\tt andreas.galanis@cs.ox.ac.uk}, {\tt leslie.goldberg@cs.ox.ac.uk}}   \and
  Leslie Ann Goldberg$^\dagger$ 
  }

\date{20 June 2016}
\maketitle

\begin{abstract}
One of the most important recent developments in the complexity of approximate counting is the 
classification of the complexity of approximating the partition functions of antiferromagnetic 2-spin systems 
on bounded-degree graphs. This classification is based on a beautiful connection to the so-called 
uniqueness phase transition from statistical physics on the infinite $\Delta$-regular tree. Our objective   is to 
study the impact of this classification on {unweighted} 2-spin models on $k$-uniform {hypergraphs}. As has 
already been indicated by Yin and Zhao,   the connection between the uniqueness phase transition  and the 
complexity of approximate counting  breaks down in the hypergraph setting. Nevertheless, we show that for 
every non-trivial symmetric $k$-ary Boolean function~$f$ there exists a degree bound~$\Delta_0$ so that 
for all $\Delta \geq \Delta_0$ the following problem is $\mathsf{NP}$-hard: given a $k$-uniform hypergraph 
with maximum degree at most~$\Delta$,  approximate the partition function of the hypergraph 2-spin model 
associated with~$f$. It is $\mathsf{NP}$-hard to approximate this partition function even within an 
exponential factor. By contrast, if $f$ is a trivial symmetric Boolean function (e.g., any function $f$ that is 
excluded from our result), then the partition function of the corresponding hypergraph 2-spin model can be 
computed exactly in polynomial time. \end{abstract} 


\section{Introduction}
\allowdisplaybreaks
One of the most important recent developments in the complexity of approximate counting is the classification of the complexity of approximating the partition functions of antiferromagnetic 2-spin systems on bounded-degree graphs \cite{LLY,SlySun}. This classification is based on a beautiful connection to the so-called uniqueness phase transition from statistical physics on the infinite $\Delta$-regular tree, which was first established in the context of the hard-core model in the works of \cite{Weitz,Sly} (see also \cite{DFJ,MWW} for related results) and later developed \cite{SST,GSV:2spin, SlySun,LLY} in the more general framework of antiferromagnetic 2-spin systems.

Our objective   is to study the impact of this classification on \emph{unweighted} 2-spin models on $k$-uniform \emph{hypergraphs}. A $k$-uniform hypergraph $H=(V,\mathcal{F})$ 
consists of a vertex set~$V$ and a set $\mathcal{F}$ of arity-$k$ hyperedges
which are $k$-element subsets of $V$.
Thus, a $2$-uniform hypergraph is the same as a graph. 
The degree of a vertex $v\in V$ is the number of edges
that contain~$v$, namely $|\{e \in \mathcal{F} \mid v \in e\}|$.
The maximum degree of~$H$ is (naturally) the maximum degree of the vertices of~$H$.

A 2-spin model on the class of $k$-uniform hypergraphs is specified by a symmetric function $f:\{0,1\}^k\rightarrow \Rplus$.
The 2-spin model is
 \emph{unweighted} if the function $f$ is  \emph{Boolean}, meaning that
 its range is a subset of the two-element set $ \{0,1\}$. Given a $k$-uniform hypergraph~$H= (V,\mathcal{F})$,
each assignment $\sigma\colon V \to \{0,1\}$ 
induces a weight 
$$w_{f;H}(\sigma):=\prod_{\{v_1,\hdots,v_k\}\in\mathcal{F}} f(\sigma(v_1),\hdots,\sigma(v_k)).$$
The assignment~$\sigma$ is sometimes referred to as a \emph{configuration}.
 The \emph{partition function} $Z_{f;H}$ corresponding to~$f$
and~$H$ is defined as follows.
\[Z_{f;H}:=\sum_{\sigma:V\rightarrow\{0,1\}}w_{f;H}(\sigma)=\sum_{\sigma:V\rightarrow\{0,1\}}\prod_{\{v_1,\hdots,v_k\}\in\mathcal{F}} f(\sigma(v_1),\hdots,\sigma(v_k)).\]
Given a symmetric function $f:\{0,1\}^k\rightarrow \Rplus$
and a hypergraph $H=(V,\mathcal{F})$ 
we will use
$\mu_{f,H}(\cdot)$ to denote the distribution 
on configurations $\sigma\colon V \to \{0,1\}$ 
in which the probability 
of configuration~$\sigma$
is proportional to its weight so $\mu_{f;H}(\sigma) \propto w(\sigma)$.
The distribution $\mu_{f,H}(\cdot)$ is called the \emph{Gibbs distribution}
associated with the partition function $Z_{f;H}$.
 
The computational problem that we study is the problem of approximating $Z_{f;H}$,
given $H$ as input. Formally, this problem has three parameters ---
a symmetric arity-$k$ Boolean function $f$,
a degree bound~$\Delta$, and a value $c>1$ which specifies the desired accuracy of the
approximation.  
The problem is defined as follows. 
\prob{ $\Hyper2Spinf$.}
{ An $n$-vertex $k$-uniform hypergraph $H$ with maximum degree at most $\Delta$.}
{  A number $\hat{Z}$ such that $c^{-n}Z_{f;H}\leq \hat{Z}\leq c^nZ_{f;H}$.}

The most well-known example of an unweighted 2-spin model is the independent set model on graphs.
In this case $k=2$, and $f$ is the function
given by $f(0,0)=f(0,1)=f(1,0)=1$ and $f(1,1)=0$.
Independent sets are in one-to-one correspondence with configurations in the model --- vertices
that are in a given independent set are assigned spin~$1$ by the corresponding configuration~$\sigma$.
The partition function $Z_{f;H}$ is simply the number of independent sets of the graph~$H$.

Let us now consider larger arity.
 A (weak) independent set in a hypergraph is a subset of vertices that does not contain a hyperedge as a subset. 
 Weak independent sets correspond to 
configurations in the 
unweighted $2$-spin model in which
$f$ is the function $f\colon\{0,1\}^k \rightarrow \{0,1\}$ where 
$f(s_1,\ldots,s_k)=1$ iff at least one of $s_1,\ldots,s_k$ is~$0$. 
  A strong independent set in a hypergraph is a subset of vertices that does not contain more than one 
vertex of any given hyperedge. 
Strong independent sets correspond to the unweighted
$2$-spin model in which 
$f(s_1,\ldots,s_k)=1$ iff at most one of $s_1,\ldots,s_k$ is~$1$.
Note that the two notions of hypergraph independent set coincide in the case $k=2$, which  
is the graph case that we have already considered.

The main motivation for our work is
the following striking result about the complexity of approximating the partition function of the independent set model on bounded-degree graphs: (i) There exists an FPRAS for the number of independent sets in graphs of maximum degree 
at most 5 \cite{Weitz}; (ii) There is no FPRAS for the number of independent sets in graphs of maximum degree at most 6  \cite{Sly} (unless NP=RP).  This computational transition was proved using insights from phase transitions and, in fact, 
the transition coincides with the so-called uniqueness threshold of the independent set model on the infinite $\Delta$-regular tree.

The question that we seek to address in this work is whether a similar computational transition occurs 
for the complexity of approximating the partition function of
  (unweighted) 2-spin models on $k$-uniform hypergraphs, in terms of the maximum degree $\Delta$. While the case $k=2$ is completely covered by the results in the previous paragraph, the picture for $k\geq 3$ appears to be much more intricate and the complexity threshold may differ from the uniqueness threshold. 

 This issue has  been discussed in \cite{YZ} in the special case of  approximately counting the  strong independent sets of a hypergraph. While the full picture is still incomplete,  it is useful to see why the complexity threshold may differ from the uniqueness threshold in this particular model for $k=3$.  As is implicit in \cite{LL} and was spelled out explicitly in \cite{YZ},  uniqueness holds for this model on the infinite $\Delta$-regular $3$-uniform hypertree if and only if $\Delta\leq 3$. For $\Delta\leq 3$, the results of \cite{LL,YZ} establish that a (non-trivial) analogue of Weitz's self-avoiding walk computational-tree approach yields an efficient approximation scheme for the partition function by (implicitly) establishing a strong spatial mixing result. Strong spatial mixing does not hold when $\Delta\geq 4$ because the infinite $\Delta$-regular $3$-uniform hypertree is in non-uniqueness. 
While it is known that it is hard to approximate the partition function for $\Delta\geq 6$,
Yin and Zhao~\cite{YZ} show that natural gadgets cannot be used to show hardness for $\Delta=4,5$
and these cases remain open.
 
Generally, as the results of \cite{LL,YZ} suggest, 
one would expect that, for ``natural" functions $f$, an FPRAS should exist 
up to the strong spatial mixing threshold, but this is (in general) below
the  uniqueness  threshold  of the $\Delta$-regular $k$-uniform hypertree.

Above the uniqueness threshold, approximating the partition function may
be hard, but this is not known in general, even for the special case of  strong
independent sets.  Thus, it is not clear 
from the literature that there is a computational threshold
where approximating the partition function on hypergraphs of maximum degree~$\Delta$
becomes intractable and it is not clear whether this
threshold, if it exists, coincides with the uniqueness threshold.

The main contribution of this paper is showing that, for every function~$f$
(apart from seven special ``easy'' functions),  there is indeed a ``barrier'' value
$\Delta_0$ such that for all $\Delta\geq \Delta_0$, it is $\NP$-hard to approximate the
partition function.
  
\begin{definition}
\label{def:easyfunctions}
For $k\geq 2$, let $\EASYk$ be the set
containing the following seven functions.
\begin{align*}
\zerof^{(k)}(x_1,\ldots,x_k)&=0,\quad
\onef^{(k)}(x_1,\ldots,x_k)=1,\quad 
\allzerof^{(k)}(x_1,\ldots,x_k)=\mathbf{1}\{x_1=\hdots=x_k=0\},\\
\allonef^{(k)}(x_1,\ldots,x_k)&=\mathbf{1}\{x_1=\hdots=x_k=1\},\quad 
\eqf^{(k)}(x_1,\ldots,x_k)=\mathbf{1}\{x_1=\hdots=x_k\}, \\
\evenf^{(k)}(x_1,\ldots,x_k)&=\mathbf{1}\{x_1\oplus\cdots\oplus x_k=0\},  \quad
\oddf^{(k)}(x_1,\ldots,x_k)=\mathbf{1}\{x_1\oplus\cdots\oplus x_k=1\}.
\end{align*}
\end{definition}

Considering each of the functions in $\EASYk$, we obtain the following observation.
\begin{observation}
Let $k\geq 2$ and $f\in \EASYk$. Then, 
the problem of (exactly) computing $Z_{f:H}$, given as input a $k$-uniform hypergraph~$H$,
can be solved in   time polynomial in the size of $H$.
\end{observation}
Our main theorem is a contrasting hardness result.
\begin{theorem}\label{thm:main}
Let $k\geq 2$ and  let $f:\{0,1\}^k\rightarrow \{0,1\}$ be a \emph{symmetric Boolean} function such that $f\notin \EASYk$. Then, there exists  $\Delta_0$ such that for all  $\Delta\geq \Delta_0$, there exists $c>1$ such that $\Hyper2Spinf$ is $\NP$-hard.
\end{theorem}

Thus we show that for all $k\geq 2$, for all non-trivial symmetric Boolean functions $f$, for all sufficiently large $\Delta$, it is $\NP$-hard to approximate $Z_{f;H}$, even within an exponential factor,   given a $k$-uniform hypergraph~$H$ of maximum degree at most $\Delta$.  
 We do not pursue the task of obtaining an explicit bound on  $\Delta$, since this would  require heavy numerical work (depending on the function $f$) and we do not expect that it would  yield the exact value of the threshold,
 even if such a threshold exists.

{\bf Note added in final version:\quad} Subsequent to this paper, the authors,
together with Bez\'{a}kov\'{a}, Guo and  \v{S}tefankovi\v{c} \cite{ICALPpaper},
have studied the issue of a computational transition specifically for the problem of counting the weak independent sets of a hypergraph.  They found further evidence, for this problem, that the complexity threshold, if it exists, may differ
from the uniqueness threshold. Particularly, for the weak independent set model, they gave an FPTAS which 
works even beyond the strong spatial mixing threshold and they showed inapproximability even below the
uniqueness threshold.

Note, however that, while Theorem~\ref{thm:main}  
does not guarantee the existence of a complexity threshold, it does at least show
inapproximability if the degree bound is sufficiently large.
That is,  for every non-trivial~$f$, it shows that 
for all sufficiently large $\Delta$, it is $\NP$-hard to approximate 
the partition function~$Z_{f;H}$,  given a $k$-uniform hypergraph~$H$ of maximum degree at most $\Delta$.\footnote{It is an open question whether Theorem~\ref{thm:main} continues to hold if the input $H$ is
further restricted to be a $\Delta$-regular $k$-uniform hypergraph. Our result does not apply to this input restriction because our gadgets are not regular hypergraphs.}

\subsection{Counting Constraint Satisfaction and Related Results}
\label{sec:CSP}

Suppose that $\Gamma$ is a set of Boolean functions of different arities. 
Thus, an arity-$k$ function in~$\Gamma$
is a function from $\{0,1\}^k$ to $\{0,1\}$.
The  counting constraint satisfaction problem
$\nCSP{\Gamma}$ is the  problem
of computing the CSP partition function $Z_{\Gamma,I}$
where $I$ is a CSP instance consisting
of a set $V$ of variables and
a  set $\mathcal{S}$\footnote{The reader who is familiar with weighted counting CSP may  have expected $\mathcal{S}$ to be a multiset rather than a set, but this is not necessary here since the functions in $\Gamma$ have range $\{0,1\}$. Restricting $\mathcal{S}$ to be a set allows some technical simplifications later.
} of constraints,  
where each constraint $C=(v_1,\ldots,v_k,f) \in \mathcal{S}$ constrains
the variables $v_1,\ldots,v_k$   by applying a particular $k$-ary function $f\in {\Gamma}$.
The value of the partition function is given by 
\[Z_{\Gamma;I}:= \sum_{\sigma:V\rightarrow\{0,1\}}\prod_{(v_1,\hdots,v_k,f)\in\mathcal{S}} 
f(\sigma(v_1),\hdots,\sigma(v_k)).\] 
The constraint $C=(v_1,\ldots,v_k,f)$ could use a particular variable more
than once. For example, it is possible that $v_1$ and $v_2$ are both the same variable.
The problem $\nCSPD{\Gamma}$ is the problem of computing $Z_{\Gamma,I}$
given an instance~$I$ in which each variable is used at most~$\Delta$ times.
We can also define a related approximation problem, similar to~$\Hyper2Spinf$.

\prob{ $\ACSPD{\Gamma}$.}
{ An $n$-variable instance $I$ of a CSP in which all constraints apply functions from~$\Gamma$
and each variable is used at most~$\Delta$ times.} 
{  A number $\hat{Z}$ such that $c^{-n}Z_{\Gamma;I}\leq \hat{Z}\leq c^nZ_{\Gamma;I}$.}

It is clear that  our problem $\Hyper2Spinf$ is
closely related to the problem $\ACSPD{\{f\}}$. In particular, 
$\Hyper2Spinf$ is the special case of  $\ACSPD{\{f\}}$
in which constraints are not allowed to re-use variables.
Thus, Theorem~\ref{thm:main} has the following immediate corollary.
\begin{corollary}\label{cor:csp}
Let $k\geq 2$ and  let $f:\{0,1\}^k\rightarrow \{0,1\}$ be a \emph{symmetric Boolean} function such that $f\notin \EASYk$. Then, there exists  $\Delta_0$ such that for all  $\Delta\geq \Delta_0$, there exists $c>1$ such that 
$\ACSPD{\{f\}}$  is $\NP$-hard.
\end{corollary} 
 
The combined results of \cite{CH} and \cite{CLX}  show that
for (exact) counting CSPs, adding a degree bound $\Delta \geq 3$ does not
change the complexity of the problem.  
The situation is less clear 
for decision and approximate counting. Previous work on bounded-degree decision CSP \cite{dalmau} and bounded-degree approximate
counting CSP \cite{DGJR} 
considered only the so-called conservative model where intractability 
arises more easily.
In this model, $\delta_0$ is the unary pinning-to-0 
function which is defined by $\delta_0(0)=1$ and $\delta_0(1)=0$.
Also, $\delta_1$ is the unary pinning-to-1
function which is defined by $\delta_1(0)=0$ and $\delta_1(1)=1$.

Theorem~24 of \cite{DGJR}   allows us to deduce 
(see Observation~\ref{obs:CSP}) 
that
for every $\Delta \geq 6$ and $k\geq 2$ and every 
symmetric $k$-ary Boolean function  $f\not\in \EASYk$,
there is no FPRAS for $\nCSP{\{f,\delta_0,\delta_1\}}$ unless $\NP=\RP$.\footnote{We remark here that \cite{DGJR} also gave a partial classification for $\Delta=3,4,5$, the remaining cases were (partially) resolved in \cite{LL}.}
 This  
hardness result extends from the uniqueness phase transition of the
independent set model (which occurs at $\Delta=6$) because  pinning allows
constructions which realise arbitrarily bad configurations.

The result of \cite{DGJR} 
does not apply to our hypergraph 2-spin context where the pinning functions~$\delta_0$ and~$\delta_1$
are not present.
 To see this, consider the following contrasting positive result of \cite{BDK} which is proved via the MCMC method: an FPRAS exists for approximating the number of (weak) independent sets in a $k$-uniform hypergraph of maximum degree $\Delta$ whenever $k\geq 2\Delta+1$.
 Thus, even though the 
  weak independent 
 function $f$
(given by 
 $f(s_1,\ldots,s_k)=1$ iff at least one of $s_1,\ldots,s_k$ is~$0$)
 is not in $\EASYk$ for any $k\geq 2$,
 the result of Bordewich et al.~\cite{BDK} shows that
for every $\Delta\leq (k-1)/2$,  
there is an FPRAS for the partition function $Z_{f;H}$ 
on the class of $k$-uniform hypergraphs~$H$ with maximum degree at most~$\Delta$.

Thus, it is clear that $\Delta=6$ cannot always be a computational threshold
in the hypergraph $2$-spin framework (where there is no pinning).
However, our Theorem~\ref{thm:main} shows that, 
for every non-trivial symmetric Boolean function~$f$,
there is degree-bound $\Delta_0$ such that approximating the partition function is
intractable beyond this degree bound.

\subsection{Proof Techniques}

In order to prove Theorem~\ref{thm:main},  
we will construct a  $k$-uniform hypergraph $H$
such that the spin-system induced by~$f$ on~$H$
induces  an anti-ferromagentic binary $2$-spin model that
is in the non-uniqueness region of
the corresponding $\Delta$-regular tree.
It follows from a result of Sly and Sun (Theorem~\ref{thm:slysun}) that
approximating the partition function of the binary model is intractable,
and we will use this to show that approximating the partition function of the $k$-ary model
is also intractable.

While this high-level approach is analogous to the one followed in \cite{DGJR}, in our setting where the pinning functions $\delta_0$ and $\delta_1$ are not available, we have to tackle several obstacles. A first indicator of the difficulties that arise is that, in \cite{DGJR}, the target 2-spin model is always the independent set model (largely due to the availability of the pinning functions $\delta_0$ and $\delta_1$). In contrast, our target binary 2-spin model will be weighted and  depend on the function $f$. In fact, we will only know the parameters of the binary 2-spin model only approximately which, as we shall discuss later in detail, poses difficulties in showing that it is intractable. 

To explain   the argument in more detail, let us backtrack and discuss a  natural approach that one might hope would  lead to proof of Theorem~\ref{thm:main}.
First, if one were able to  
construct hypergraphs to
``realise"   the pinning functions~$\delta_0$ and~$\delta_1$ 
then these hypergraphs could be combined with the reduction in~\cite{DGJR} to
prove Theorem~\ref{thm:main}.
The proof  would even be straightforward if 
 perfect reaslisations could be found. 
For example, to realise~$\delta_0$ perfectly we would need a hypergraph $H$  
whose partition function is non-zero which 
has a vertex $v$   such that  every configuration~$\sigma$ with
$w_{f;H}(\sigma)>0$ satisfies $\sigma(v)=0$.
More realistically, one might hope that even  ``approximate" versions of the pinning functions $\delta_0$ or $\delta_1$ would suffice to simulate the reduction in \cite{DGJR}. Unfortunately, this fails rather formidably: first, 
as we shall see below, there are functions $f$ which simply cannot realise (approximate) pinning, and, second, even for those functions $f$ which do support pinning,  the bounded-degree assumption poses strict  limits on  the accuracy of the approximations that can be achieved. 

Despite the failure of the above approach,
it does turn out to be useful to explore the extent to which   the pinning functions $\delta_0$ and $\delta_1$ can be simulated using hypergraphs. We know from the binary case (where the uniqueness phase transition
coincides with the computational transition) that the  achievable ``boundary conditions'' play an important role.
Understanding the pinnings that can be (approximately) achieved gives us the relevant
boundary information for the higher-arity case.
To make the following discussion concrete, consider the following definition (stated more generally for weighted functions $f$). 

\begin{definition}\label{def:pinning}
Let $f:\{0,1\}^k\rightarrow\Rplus$ be symmetric. Suppose that $\epsilon\geq 0$
and $s\in \{0,1\}$. The hypergraph $H$ is an \emph{$\epsilon$-realisation} of \emph{pinning-to-$s$}
  if there exists a vertex $v$ in $H$ such that 
  $\mu_{f;H}(\sigma_v=s)\geq 1-\epsilon$. We will refer to $v$ as the \emph{terminal} of $H$.
\end{definition}

Note that the perfect  realisation discussed earlier corresponds to taking
$\epsilon=0$. 
Suppose that we have an $\epsilon$-realisation of pinning-to-$s$
but we want an $\epsilon'$-realisation 
for some very small positive~$\epsilon'$.
This can be done via standard powering  (see the upcoming Lemma~\ref{lem:gadgets}):
Given a hypergraph $H$ which $\epsilon$-realises  pinning-to-$s$ for some $\epsilon<1/2$,  
 one can construct a hypergraph $H'$ which $\epsilon'$-realises  pinning-to-$s$.  
 Note, however, that the size of~$H'$ may depend on~$\epsilon'$.
For example, in the construction Lemma~\ref{lem:gadgets}, the maximum degree of~$H'$
is proportional to $\log(1/\epsilon')$.
Nevertheless, the possibility of powering motivates the following definition. 
 
\begin{definition}\label{def:boundedsupport1}
Let $s\in\{0,1\}$.
We say that  $f$ supports pinning-to-$s$   if for every $\epsilon>0$,  
there is a (finite) hypergraph $H$ which is an  $\epsilon$-realisation of  pinning-to-$s$.
\end{definition}

We will next consider an example which demonstrates the limits of
what can be achieved. Let   $f\colon\{0,1\}^k \rightarrow \{0,1\}$  
be the  weak independent set
function where $f(s_1,\ldots,s_k)=1$ if and only if at least one of $s_1,\ldots,s_k$ is~$0$. 
First, note that $f$ does not support pinning-to-1 
since for every hypergraph $H$ and every vertex $v$ in $H$ it holds that $\mu_{f;H}(\sigma(v)=1)\leq 1/2$. 
The function~$f$ does support pinning-to-0
but there is still a limit on how small~$\epsilon$ can be.
In particular, 
for every $k$-uniform
hypergraph $H$ with maximum degree $\Delta$,  
and every vertex $v$ of~$H$ we can obtain the crude bound 
$\mu_{f;H}(\sigma(v)=0)\leq 1-1/2^{k\Delta}$.
This shows that we 
 cannot hope to pin the spin of a vertex to~$0$ with arbitrary polynomial precision 
 using bounded-degree hypergraphs. 
 Note that this example already shows that 
it is impossible to prove   Theorem~\ref{thm:main} by approximating  the pinning functions $\delta_0$ and $\delta_1$ and then applying the result of  \cite{DGJR}.

Nevertheless, pinning-to-0 and pinning-to-1 will be important for us since, whenever a function $f$ supports one (or both) of these notions, we  will be able to use them to decrease the arity of the function $f$. This is particularly useful since, recall, our ultimate goal is to obtain an intractable binary 2-spin model. Intriguingly,  there are  functions $f$ which do not support  either pinning-to-0 or pinning-to-1.  For  example, consider the function $f$ which is induced by the ``not-all-equal'' 
constraint. Then,
for every hypergraph $H$ with $Z_{f,H}>0$ and every vertex $v$  of $H$,
it  is easy to see that $\mu_{f;H}(\sigma(v)=1)=\mu_{f;H}(\sigma(v)=0)=1/2$. More generally, the same phenomenon holds for 
any function $f$ whose value is unchanged when the argument is complemented; such functions are called ``self-dual''.

 The first point that we address in this work is a complete characterisation of the functions $f$ which  support either the notion of pinning-to-0 or pinning-to-1. We show 
 (Lemma~\ref{lem:classify} and Lemma~\ref{lem:selfdual})
 that any function $f$ other than those that are self-dual do support either  pinning-to-0 or pinning-to-1 (but perhaps not both). We   show this classification even for weighted functions $f$, see Section~\ref{sec:properties} for more details. The classification allows us to split the proof of Theorem~\ref{thm:main} into three cases: (i) $f$ supports both pinning-to-0 and pinning-to-1, (ii) $f$ is self-dual, and (iii) $f$ supports exactly one of pinning-to-0 and pinning-to-1.

 In cases (i) and (iii) (Sections~\ref{sec:caseI} and~\ref{sec:caseIII}, respectively)
 where pinning is available we show how to use the approximate pinning to simulate binary antiferromagnetic 2-spin models that are intractable. A difficulty that arises in the proof is that not every anti-ferromagnetic binary $2$-spin model is in the non-uniqueness region. In fact, there are relevant values of the parameters for which the corresponding binary $2$-spin model is actually in the uniqueness region for all sufficiently large~$\Delta$. To make matters worse, we will not be able to control the parameters of the resulting binary model with perfect accuracy. In particular,  to analyse the $k$-ary gadgets, we will use $\epsilon$-realisations of pinnings via hypergraphs for some small $\epsilon>0$.  Thus, we are faced with the possibility that the idealised binary $2$-spin model (i.e., the one corresponding to $\epsilon=0$) may be in the non-uniqueness region, but we need to prove that the 
approximate version that we actually achieve is also in the non-uniqueness region. In fact, the idealised binary 2-spin model will sometimes even be on the boundary 
of the region where intractability holds for sufficiently large~$\Delta$, which makes our task harder.

Our approach to this is to revisit (Section~\ref{sec:binary}) antiferromagnetic binary 2-spin models,
showing (Lemma~\ref{lem:strip}) that there is a sufficiently-wide strip outside of
the  natural square where the parameters are at most~$1$ where
the system is in the non-uniqueness region. We will then carefully ensure that
all of the idealised systems are inside this strip, so that even the approximations are still in non-uniqueness.

In case (ii) (Section~\ref{sec:CaseII}), where the function $f$ is self-dual and hence no pinning is possible, we first classify those self-dual functions $f$ where the related decision problem is $\NP$-hard.  In order to do so, we use  techniques (polymorphisms) from constraint satisfaction, which are explained in detail in Section~\ref{sec:qwe}. While this hardness is not for the bounded-degree setting, we show how to lift the results to bounded-degree hypergraphs by showing that one can force the spins of two vertices to be equal (Lemma~\ref{lem:exactequality}). The proof for this class of self-dual functions $f$ is given in Section~\ref{sec:w00}.

  For those self-dual functions $f$ where the associated decision problem is not hard (Section~\ref{sec:Case2w0one}), we show that one can realise approximate equality in the following sense.
\begin{definition}\label{def:equality}
Let $\epsilon\geq 0$ and $t\geq 2$ be an integer. The hypergraph $H$ is an \emph{$\epsilon$-realisation} of \emph{$t$-equality} if there exist distinct vertices $v_1,\hdots,v_t$ such that for each $s\in\{0,1\}$,
$$
\mu_{f;H}(\sigma_{v_1}=\hdots=\sigma_{v_t}=s)\geq (1-\epsilon)/2.
$$
We will refer to $v_1,\hdots,v_t$ as the \emph{terminals} of $H$.
\end{definition}

\begin{definition}\label{def:boundedsupport2a}
A function  $f$ supports $t$-equality if for every $\epsilon>0$,  
there is a (finite) hypergraph $H$   which is an $\epsilon$-realisation of  $t$-equality. 
\end{definition}
Using the upcoming Lemmas~\ref{lem:equalst} and~\ref{lem:classify}, we show that a self-dual function $f$ supports $t$-equality for every integer $t\geq 2$. Roughly, this allows us to decrease the arity of the function by  carefully using (approximate) equality to obtain an anti-ferromagnetic binary 2-spin model which is intractable (note that we again have to deal with the approximation issue that we described for cases (i) and (iii)).

\subsection{Notation}

We conclude this section with a piece of notation. Given a configuration~$\sigma$ and a subset $T\subseteq V$, we will use the notation $\sigma_T$ to denote the restriction of~$\sigma$ to vertices
in~$T$. For a vertex $v\in V$, we will also use $\sigma_v$ to denote the
spin $\sigma(v)$ of vertex~$v$ in~$\sigma$. 
 Given a hyperedge $e\in \mathcal{F}$, we will denote by $H\setminus e$ the hypergraph $(V,\mathcal{F}\setminus e)$.

\section{Properties of non-negative symmetric functions with domain $\{0,1\}^k$}\label{sec:properties}
In this section, we study the concepts of pinning and equality that we will use for the proof of Theorem~\ref{thm:main}. While our primary interest is in symmetric \emph{Boolean} functions $f$, the results of this section extend effortlessly to non-negative symmetric functions $f$ with domain $\{0,1\}^k$ 
and range $\Rplus$.
For the 
remainder of this section, we consider
a symmetric function $f: \{0,1\}^k \to \Rplus$.
Since $f$ is symmetric, there  are values  $w_0,w_1,\hdots,w_k\in \Rplus$ such that $f(x_1,\hdots,x_k)=w_\ell$ whenever $x_1+\hdots+x_k=\ell$. We will refer to~$f$ and to the values $w_i$ in the definitions and proofs in this section.

\subsection{Pinning and equality}

We start with the following remark, which follows from Defintion~\ref{def:equality}
(and makes the definition easier to apply).

\begin{remark}\label{rem:equality}
If $H$ is an $\epsilon$-realisation of $t$-equality and $v_1,\hdots,v_t$ are the \emph{terminals} of $H$,
then  it also holds that $\mu_{f;H}(\sigma_{v_1}=\hdots=\sigma_{v_t}=s)\leq (1+\epsilon)/2$ for  each $s\in\{0,1\}$. Further, we have that $\mu_{f;H}(\exists i,j: \sigma_{v_i}\neq \sigma_{v_j})\leq \epsilon$.
\end{remark}

Next, we give a straightforward extension to the notion of supporting $t$-equality
(see Definition~\ref{def:boundedsupport2a}).

\begin{lemma}\label{lem:equalst}
Let $t\geq 2$ be an integer. The function $f$ supports $t$-equality iff $f$ supports 2-equality.
\end{lemma}
\begin{proof} 
It is immediate   that if $f$ supports $t$-equality 
for some $t\geq 2$
then it supports $2$-equality (terminals $v_3,\ldots,v_t$ can simply be ignored).

We will now suppose that $f$ supports $2$-equality and show that it supports $t$-equality
for a given $t>2$.
Consider $\epsilon>0$.
Choose $\delta>0$ to be sufficiently small (with respect to~$\epsilon$ and~$t$)
so that $\delta \leq \epsilon 2^{-(t+2)}$ and 
$${\left(\frac{1-\delta}{1+\delta}\right)}^{\binom{t}{2}} \geq   \max\left\{\frac{1-\epsilon/2}{1+\epsilon/2},\frac{1}{2}\right\}.$$  
Suppose that $H$ is a $\delta$-realisation of $2$-equality
so it has terminals~$x$ and~$y$
so that for each $s\in \{0,1\}$
$$
\frac{1-\delta}{2} \leq 
\mu_{f;H}(\sigma_x = \sigma_y=s)
\leq \frac{1+\delta}{2} 
 \text{ and }  \mu_{f;H}(\sigma_x = s;\sigma_y=s\oplus 1)    \leq  \delta.$$
Let $H'$ be the hypergraph  constructed as follows.
Let $T = \{v_1,\ldots,v_t\}$ be a set of~$t$ vertices which will be the terminals of~$H'$.
For each  $1\leq i<j \leq t$, let $H_{ij}$ be a new copy of~$H$  
but identify the terminal~$x$ of~$H_{ij}$ with $v_i$ and the terminal~$y$ of~$H_{ij}$ with~$v_j$.
Let $H'$ be the resulting hypergraph.
Now for any
$\tau: \{v_1,\ldots,v_t\} \rightarrow \{0,1\}$
that does \emph{not} satisfy 
$\tau(v_1) = \cdots = \tau(v_t)$,
the contribution to
$Z_{f;H'}$ from configurations $\sigma$ with $\sigma_T=\tau$
is at most 
$$\delta Z_{f;H} {\left(\left(\frac{1+\delta}{2}\right)Z_{f;H}\right)}^{\binom{t}{2}-1}
= \left(\frac{\delta}{1+\delta}\right) {(1+\delta)}^{\binom{t}{2}} \frac{Z_{f;H}^{\binom{t}{2}}}{2^{\binom{t}{2}-1}}$$
On the other hand, by considering the contribution from 
configurations~$\sigma$ with $\sigma(v_1) = \cdots = \sigma(v_t)$, we obtain that $Z_{f;H'}$ is at least
$$ {2\left(\left(\frac{1-\delta}{2}\right)Z_{f;H}\right)}^{\binom{t}{2}}
= {(1-\delta)}^{\binom{t}{2}} \frac{Z_{f;H}^{\binom{t}{2}}}{2^{\binom{t}{2}-1}},$$
Thus
 $$\mu_{f;H'}(\sigma_T=\tau) 
 \leq 
  \frac
  { \left(\frac{\delta}{1+\delta}\right) {(1+\delta)}^{\binom{t}{2}}}
  {{(1-\delta)}^{\binom{t}{2}} }
  \leq \frac{2\delta}{1+\delta}
  \leq 
 2\delta$$ so, 
since $2^t (2\delta) \leq \epsilon/2$,
$\mu_{f;H'}(\exists i,j: \sigma_{v_i}\neq \sigma_{v_j})\leq \epsilon/2$. 
Furthermore, for any $s\in \{0,1\}$,
$$\frac{\mu_{f;H'}(\sigma(v_1)=\cdots = \sigma(v_t)=s)}
{\mu_{f;H'}(\sigma(v_1)=\cdots = \sigma(v_t)=s\oplus 1)}  \geq
 {\left(\frac{1-\delta}{1+\delta}\right)}^{\binom{t}{2}} \geq  \frac{1-\epsilon/2}{1+\epsilon/2}.
$$ 
It follows
that
\begin{align*}
(1+\epsilon/2) \mu_{f;H'}(\sigma(v_1)=\cdots = \sigma(v_t)=s)
&\geq  (1-\epsilon/2) \mu_{f;H'}(\sigma(v_1)=\cdots = \sigma(v_t)=s\oplus 1)\\
&\geq (1-\epsilon/2) (1 - \mu_{f;H'}(\sigma(v_1)=\cdots = \sigma(v_t)=s) - \epsilon/2),
\end{align*}
so 
$\mu_{f;H'}(\sigma(v_1)=\cdots = \sigma(v_t)=s)\geq (1-\epsilon)/2$.
Thus,
$H'$ is an $\epsilon$-realisation of $t$-equality.\end{proof}

Lemma~\ref{lem:equalst} motivates the following definition.

\begin{definition}\label{def:boundedsupport2b}
A function $f$ \emph{supports equality} if, for some $t\geq 2$,
it supports $t$-equality. (In this case, Lemma~\ref{lem:equalst} shows that $f$ supports
$t$-equality for every $t\geq 2$.)
\end{definition}

The following lemma gives sufficient conditions for pinning-to-0, pinning-to-1 and 2-equality. 
\begin{lemma}\label{lem:gadgets}
Let $H$ be a hypergraph and $f:\{0,1\}^k\rightarrow \Rplus$ be symmetric. 
\begin{enumerate}
\item \label{it:pinning0} If there  is a vertex $v$ in $H$ such that $\mu_{f;H}(\sigma_v=0)>\mu_{f;H}(\sigma_v=1)$, then $f$ supports pinning-to-0.
\item \label{it:pinning1} If there is a vertex $v$ in $H$ such that $\mu_{f;H}(\sigma_v=1)>\mu_{f;H}(\sigma_v=0)$, then $f$ supports pinning-to-1.
\item \label{it:equality} If there  are  vertices $x,y$ in $H$ such that 
$\mu_{f;H}(\sigma_{x}=\sigma_{y}=0)=\mu_{f;H}(\sigma_{x}=\sigma_{y}=1)$ and 
$\mu_{f;H}(\sigma_{x}=\sigma_{y})>\mu_{f;H}(\sigma_{x}\neq \sigma_{y})$, then $f$ supports 2-equality.
\end{enumerate}
\end{lemma}
\begin{proof} 
We start with Item~\ref{it:equality}, which is the most difficult.
Given $\epsilon>0$, we will use $H$ to construct a hypergraph $H'$ which is an $\epsilon$-realisation of $2$-equality.
We start by constructing a hypergraph~$H''$ with terminals~$v_1$ and~$v_2$.
We construct~$H''$ 
by taking two copies of~$H$. In the first copy, we identify the vertex~$x$ with the terminal~$v_1$ and
the vertex~$y$ with the terminal~$v_2$.
The second copy is disjoint from the first one, except that we identify the vertex~$x$ of the second copy with
the terminal~$v_2$ and the vertex~$y$ of the second copy with the vertex~$v_1$.

Let $\mu:=\mu_{f;H}$.
Then let $p=\mu(\sigma_{x}=\sigma_{y}=0)^2=\mu(\sigma_{x}=\sigma_{y}=1)^2$ 
and 
$q=\mu(\sigma_{x}=0,\sigma_{y}=1)\mu(\sigma_{x}=1,\sigma_{y}=0)$.

If $q=0$ then we can take~$H'$ to be~$H''$.
Then $\mu_{f;H'}(\sigma_{v_1} \neq \sigma_{v_2})=0$.
However, 
$\mu_{f;H'}(\sigma_{v_1} = \sigma_{v_2}=0)  = \mu_{f;H'}(\sigma_{v_1}=\sigma_{v_2}=1)$
so for $s\in \{0,1\}$, $\mu_{f;H'}(\sigma_{v_1} = \sigma_{v_2}=s)=  1/2$ and $H'$ is a $0$-realisation of
$2$-equality.

So suppose $q>0$. Contruct~$H'$ by taking  $r=1+\left\lceil \ln\epsilon/\ln(q/p)\right\rceil$ disjoint copies of 
$H''$, identifying all terminals~$v_1$ and all terminals~$v_2$.
 Let $\mu':=\mu_{f;H'}$. We have
\[\mu'(\sigma_{v_1}=\sigma_{v_2}=0)\propto p^r, \mu'(\sigma_{v_1}=\sigma_{v_2}=1)\propto p^r, \mu'(\sigma_{v_1}=0,\sigma_{v_2}=1)
= \mu'(\sigma_{v_1}=1,\sigma_{v_2}=0) 
\propto q^r.\]  Our choice of $r$ ensures that $\mu'(\sigma_{v_1}=0,\sigma_{v_2}=1)/\mu'(\sigma_{v_1}=\sigma_{v_2}=0)<\epsilon$, so $H'$ is an $\epsilon$-realisation of $2$-equality.

The proofs for Items~\ref{it:pinning0} and \ref{it:pinning1} 
are similar but simpler. We will do Item~\ref{it:pinning0}. 
Let $p=\mu(\sigma_v=0)$ and $q= \mu(\sigma_v=1)$. Construct $H'$ by taking 
$r=1+\left\lceil \ln\epsilon/\ln(q/p)\right\rceil$ disjoint copies of 
$H$, identifying vertex~$v$ in all copies.
Then $\mu_{f;H'}(\sigma_v=1) = 
q^r/(q^r+p^r) \leq
(q/p)^r \leq \epsilon$ so $H'$ is an $\epsilon$-realisation of pinning-to-0.
\end{proof}

\subsection{Classifying functions with respect to pinning and equality} 
The following lemma will be used in our classification.

\begin{lemma}\label{lem:classify}
Let $k\geq 2$. For all $f:\{0,1\}^k\rightarrow \Rplus$ which are not constant,  it holds that $f$ supports at least one of pinning-to-0, pinning-to-1 and 2-equality.
\end{lemma}
\begin{proof}  
Assume that $f$ does not support pinning-to-0 or pinning-to-1. We will show that $f$ supports 2-equality.

Let $H$ be the hypergraph with vertex set $\{x,y,z_1,\hdots,z_{k-1}\}$ and hyperedge set $\mathcal{F}=\{e_1,e_2\}$, where $e_1=\{x,z_1,\hdots,z_{k-1}\}$ and $e_2=\{y,z_1,\hdots,z_{k-1}\}$ (see Figure~\ref{fig:equality}). Let $\mu:=\mu_{f;H}$. We have that
\[\mu(\sigma_x=s_1,\sigma_y=s_2)\propto Z_{s_1s_2}\mbox{ for } s_1,s_2\in\{0,1\},\]
where $Z_{s_1s_2}=\sum^{k-1}_{\ell=0}\binom{k-1}{\ell}w_{\ell+s_1}w_{\ell+s_2}$. Note that $Z_{01}=Z_{10}$.
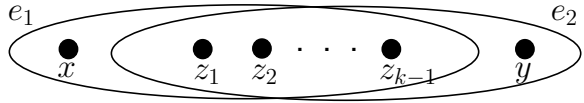
\begin{figure}[t]
\begin{center}
\scalebox{0.6}[0.6]{\input{hypergraph1.tex}}
\end{center}
\caption{The hypergraph $H$ used in the proof of Lemma~\ref{lem:classify}. The hypergraph has two hyperedges $e_1,e_2$ such that $e_1\cap e_2=\{z_1,\hdots,z_{k-1}\}$, $e_1\backslash e_2=\{x\}$ and $e_2\backslash e_1=\{y\}$. We focus on the spins $s_1,s_2$ of the vertices $x,y$, respectively, i.e., for a configuration $\sigma$ on $H$, $s_1=\sigma_x$ and $s_2=\sigma_y$.}\label{fig:equality}
\end{figure}

We first show that $Z_{00}=Z_{11}$. Assume otherwise. Note that 
$\mu(\sigma_x=0)\propto Z_{01}+Z_{00}$ and
$\mu(\sigma_x=1)\propto Z_{10}+Z_{11}$.
Since $Z_{01}=Z_{10}$, $Z_{00}\neq Z_{11}$ would imply $\mu(\sigma_x=0)\neq \mu(\sigma_x=1)$, contradicting  that $f$ does not support pinning-to-0 or pinning-to-1 (by Lemma~\ref{lem:gadgets}).

Further, we have that $Z_{00}Z_{11}\geq Z_{01}^2$, since 
\begin{equation}\label{eq:bcs}
\left[\sum^{k-1}_{\ell=0}\binom{k-1}{\ell}w_{\ell}^2\right]\left[\sum^{k-1}_{\ell=0}\binom{k-1}{\ell}w_{\ell+1}^2\right]\geq \left[\sum^{k-1}_{\ell=0}\binom{k-1}{\ell}w_{\ell}w_{\ell+1}\right]^2
\end{equation}
holds as an immediate consequence of the Cauchy-Schwartz inequality. From $Z_{00}=Z_{11}$, we thus obtain that $Z_{00}\geq Z_{01}$. Equality in \eqref{eq:bcs} holds only if there exists $\alpha\geq 0$ such that $w_{\ell+1}=\alpha w_\ell$ for every $\ell=0,\hdots,k-1$, which yields $w_\ell=\alpha^\ell w_0$ for $\ell=0,\hdots,k$. This gives $Z_{11}=\alpha^2 Z_{00}$, so  $Z_{00}=Z_{11}$ leaves only the possibility $\alpha=1$, which in turn yields that $f$ is a constant function. 

Thus, it holds that $Z_{00}=Z_{11}>Z_{01}=Z_{10}$, so Item~\ref{it:equality} of Lemma~\ref{lem:gadgets} yields that $f$ supports 2-equality.
\end{proof}

\begin{lemma}\label{lem:selfdual}
Let $k\geq 2$. If $f$ supports 2-equality but neither  pinning-to-0 nor pinning-to-1, then it holds that $w_\ell=w_{k-\ell}$ for all $\ell=0,\hdots,k$. 
\end{lemma}
\begin{proof} 

Let $k\geq 2$ and suppose that $f$ supports 2-equality but neither pinning-to-0 nor pinning-to-1.
Let $H$ be the hypergraph with the vertex set $\{v_1,\hdots,v_{k}\}$ and the single hyperedge $e=\{v_1,\hdots,v_{k}\}$.  
We may assume that $f$ is not a constant function since a constant function does not support 2-equality. Thus, at least one of the $w_{\ell}$'s is non-zero and hence $Z_{f;H}>0$.

We will start by establishing the claim for $\ell=0$, 
by showing that $w_0=w_k$.
Assume for contradiction   
that $w_0\neq w_k$. W.l.o.g we may assume that $w_0>w_k$ (if it is the other way around, then we can swap $0$'s and $1$'s in the following argument).
Since $f$ supports $2$-equality, it also supports $k$-equality by Lemma~\ref{lem:equalst}.
Choose $\epsilon>0$ sufficiently small so that 
$$
 w_0  \left(\frac{1-\epsilon}{2}\right)
>  w_k \left(\frac{1+\epsilon}{2}\right) + 2^k \epsilon.
$$
Construct
 $H'$ by taking $H$ and a distinct $\epsilon$-realisation $H''$ of $k$-equality
and identifying  the vertices of~$H$ with the terminals of~$H''$ (see Figure~\ref{fig:caseonea}). Then
the contribution to $Z_{f,H'}$ 
from configurations with $\sigma(v_k)=0$ is
at least the contribution from  all configurations 
 that assign spin~$0$ to all terminals
(giving a contribution of at least 
$w_0  \left(\frac{1-\epsilon}{2}\right) Z_{f,H''}$
since $H''$ is an $\epsilon$-realisation of $k$-equality)
so we get
$$
\mu_{f,H'}(\sigma(v_k)=0) Z_{f,H'} \geq w_0 \left(\tfrac{1-\epsilon}{2}\right) Z_{f,H''}.$$
Now consider the contribution to $Z_{f,H'}$
from configurations with $\sigma(v_k)=1$.
The contribution 
from   configurations which   map all terminals to spin~$1$
is at most   $w_k \left(\frac{1+\epsilon}{2}\right) Z_{f,H''}$ since
$H''$ is an $\epsilon$-realisation of $k$-equality.
In addition, configurations which do not
make the spins at the terminals equal
contribute at most
$ 2^k \epsilon Z_{f,H''}$.
So we get
$$\mu_{f;H'}(\sigma(v_k)=1) Z_{f,H'} \leq  w_k \left(\tfrac{1+\epsilon}{2}\right) Z_{f,H''} 
  +  2^k \epsilon Z_{f;H''}.$$
It follows by the choice of $\epsilon$ that  
 $\mu_{f;H'}(\sigma(v_k)=0) > \mu_{f,H'}(\sigma(v_k)=1)$,
 so Lemma~\ref{lem:gadgets} shows that $f$ supports pinning-to-0, contrary to the statement of the lemma.
Thus, we have  shown that $w_0=w_k$.

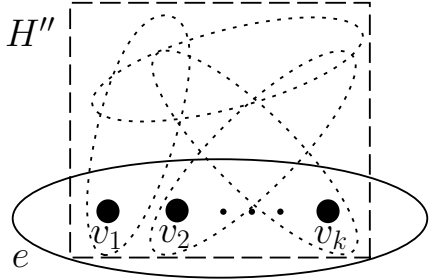
\begin{figure}[t]
  \centering
  \scalebox{0.7}[0.7]{\input{hypergraph2a.tex}}
  \caption{The hypergraph $H'$ used to show that $w_0=w_k$ in Lemma~\ref{lem:selfdual}. It consists of a hypergraph $H''$ (dashed in the figure) which $\epsilon$-realises $k$-equality among its terminals $v_1,\hdots,v_k$  and the hyperedge $e=\{v_1,\hdots,v_k\}$. We show that if $w_0> w_k$, then $H'$ realises pinning-to-0.}
\label{fig:caseonea}
\end{figure}

For ease of notation, we let $Z_{0,k} = w_0$ and $Z_{1,k} = w_k$.
Then, for $t\in\{1, \ldots,k-1\}$ let $V_t=\{v_1,\hdots,v_t\}$.  
For  $s\in \{0,1\}$ 
let $Z'_{s,t} = \sum^{k-1-t}_{\ell=0}\binom{k-1-t}{\ell} w_{\ell+s}$
and $Z''_{s,t} = \sum^{k-1-t}_{\ell=0}\binom{k-1-t}{\ell} w_{\ell+s+t}$.
Let $Z_{s,t} = Z'_{s,t} + Z''_{s,t}$.
We will   establish the following system of equalities.
\begin{equation}
\label{eq:Zbal}
\text{For $t\in \{1,\ldots,k\},$ } 
Z_{0,t} = Z_{1,t}.
\end{equation}

We have already dealt with the case $t=k$.
Next, consider $t=1$.
Note that $\mu_{f;H}(\sigma(v_k)=0)=Z_{0,1}/Z_{f;H}$ and
$\mu_{f;H}(\sigma(v_k)=1) = Z_{1,1}/Z_{f,H}$  (to see these, observe that $Z'_{s,1}$ is the contribution from configurations where, say, vertex $v_{1}$ has spin 0 and $v_k$ has spin $s$; similarly, $Z''_{s,1}$ is the contribution where $v_{1}$ has spin 1 and $v_k$ has spin $s$). 
Thus, if $Z_{0,1} \neq Z_{1,1}$ then 
$\mu_{f;H}(\sigma(v_k)=0)\neq \mu_{f;H}(\sigma(v_k)=1)$ so Lemma~\ref{lem:gadgets} shows
that $f$ supports pinning-to-0 or pinning-to-1, contrary to the statement of the lemma.

We will now show that Equation~\eqref{eq:Zbal} holds for $t \in \{2,\ldots,k-1\}$. 
The proof is similar to the proof that $w_0=w_k$ above.
Consider some $t\in\{2,\ldots,k-1\}$, and
suppose for contradiction that $Z_{0,t} > Z_{1,t}$. 
Let $\rho= Z_{1,t}\geq 0$ and 
$\delta = Z_{0,t} - Z_{1,t}>0$.
Again, since $f$ supports $2$-equality, it also supports $t$-equality by Lemma~\ref{lem:equalst}.
Choose $\epsilon$ sufficiently small so that 
$$
(\rho+\delta)  \left(\frac{1-\epsilon}{2}\right)
> \rho \left(\frac{1+\epsilon}{2}\right) + 2^k \epsilon
$$
Construct
 $H'$ by taking $H$ and a distinct $\epsilon$-realisation $H''$ of $t$-equality
and identifying  the terminals $v_1,\ldots,v_t$ in~$H$ and~$H''$ (see Figure~\ref{fig:caseoneb}).
Note that vertex~$v_k$ is not a vertex of~$H''$.
So
the contribution to $Z_{f,H'}$ 
from configurations with $\sigma(v_k)=0$ is
at least the contribution from such configurations 
which also satisfy $\sigma(v_1)=\cdots=\sigma(v_t)=0$
(giving a contribution of at least $Z'_{0,t} \left(\frac{1-\epsilon}{2}\right) Z_{f,H''}$)
and a similar contribution 
from configurations which also satisfy $\sigma(v_1)=\cdots=\sigma(v_t)=1$
(giving a contribution of at least $Z''_{0,t} \left(\frac{1-\epsilon}{2}\right) Z_{f,H''}$)
so  we get
\begin{align*}
\mu_{f,H'}(\sigma(v_k)=0) Z_{f,H'}& \geq Z'_{0,t} \left(\tfrac{1-\epsilon}{2}\right) Z_{f,H''}+ Z''_{0,t}  \left(\tfrac{1-\epsilon}{2}\right) Z_{f,H''} = Z_{0,t}  \left(\tfrac{1-\epsilon}{2}\right) Z_{f,H''}\\
&\geq  (\rho+\delta) \left(\tfrac{1-\epsilon}{2}\right) Z_{f,H''}.
\end{align*}
Now consider the contribution to $Z_{f,H'}$
from configurations with $\sigma(v_k)=1$.
This is at most the contributions which also satisfy 
$\sigma(v_1)=\cdots=\sigma(v_t)=0$
(giving a contribution at most $Z'_{1,t} \left(\frac{1+\epsilon}{2}\right) Z_{f,H''}$)
and a similar term 
$Z''_{1,t} \left(\frac{1+\epsilon}{2}\right) Z_{f,H''}$ from the contributions which also satisfy 
$\sigma(v_1)=\cdots=\sigma(v_t)=1$.
In addition, configurations which do not
satisfy $\sigma(v_1)=\cdots=\sigma(v_t)$
contribute at most
$ 2^k \epsilon Z_{f,H''}$.
So we get
$$\mu_{f;H'}(\sigma(v_k)=1) Z_{f,H'} \leq  \rho \left(\tfrac{1+\epsilon}{2}\right) Z_{f,H''} 
  +  2^k \epsilon Z_{f;H''}.$$
  Again by the choice of $\epsilon$, we have  
 $\mu_{f;H'}(\sigma(v_k)=0) > \mu_{f,H'}(\sigma(v_k)=1)$,
 so Lemma~\ref{lem:gadgets} shows that $f$ supports pinning-to-0, contrary to the statement of the lemma.
 So we have now established the system of equations~\eqref{eq:Zbal}.

\begin{figure}[t]
  \centering
  \scalebox{0.7}[0.7]{\input{hypergraph2b.tex}}
  \caption{The hypergraph $H'$ used to show that $Z_{0,t}=Z_{1,t}$ for $t\in\{2,\hdots,k-1\}$. It consists of a hypergraph $H''$ (dashed in the figure) which $\epsilon$-realises $t$-equality among its terminals $v_1,\hdots,v_t$  and the hyperedge $e=\{v_1,\hdots,v_t,v_{t+1},\hdots,v_k\}$. We show that if $Z_{0,t}>Z_{1,t}$, then $H'$ realises pinning-to-0.}
\label{fig:caseoneb}
\end{figure}
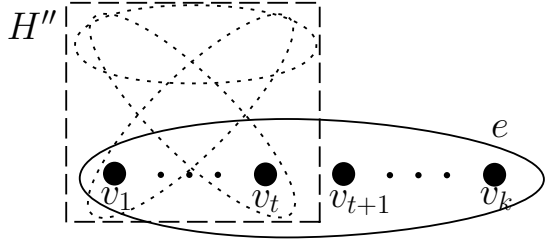

We view~\eqref{eq:Zbal}
as a system of equations over the $k+1$ (real) variables $w_0,\hdots,w_k$. 
We refer to the equation $Z_{0,t} = Z_{1,t}$ as  ``Equation~$t$'' or ``the $t$'th equation''.
In the remainder of the proof we will show
that the solutions to this system of equations
are exactly the assignments of real values to the variables  
satisfying
\begin{equation}
\label{eq:stmt}
w_\ell=w_{k-\ell}  \text{ for } \ell=0,\hdots,k.
\end{equation}

We first show that any solution satisfying~\eqref{eq:stmt} 
satisfies~\eqref{eq:Zbal}.
The case $t=k$ is obvious. For $t\in\{1,\ldots,k-1\}$,
we will show that any solution satisfying~\eqref{eq:stmt}
satisfies
$Z'_{0,t} = Z''_{1,t}$ and $Z'_{1,t} = Z''_{0,t}$. 
To see the first of these, 
consider the coefficient $\binom{k-1-t}{\ell'}$ of $w_{\ell'+1+t}$ in $Z''_{1,t}$.
Now let $\ell = k- (\ell'+1+t)$ so $w_{\ell} = w_{\ell'+1+t}$.
The coefficient of $w_\ell$ in $Z'_{0,t}$ is
$\binom{k-1-t}{\ell} = \binom{k-1-t}{k-1-t-\ell} = \binom{k-1-t}{\ell'}$.
This establishes $Z'_{0,t} = Z''_{1,t}$. The proof  that $Z'_{1,t} = Z''_{0,t}$ is similar.
We will now show that the solution in~\eqref{eq:stmt} has the appropriate
dimension, so there are no other solutions of~\eqref{eq:Zbal}.
It is simplest to split the calculation into two cases, depending on the parity of~$k$.
 
First, suppose $k=2r+1$ is odd.
Consider $1\leq t \leq t' \leq k-1$. We will compute the 
coefficients of the variable $w_t$ in the
quantities $Z'_{0,t'}$, $Z''_{0,t'}$, $Z''_{1,t'}$ and $Z''_{1,t'}$ using the convention
that $\binom{a}{b}=0$ if $b\not\in\{0,\ldots,a\}$.
These are
$\binom{2r-t'}{t}$, $\binom{2r-t'}{t-t'}$, $\binom{2r-t'}{t-1}$ and $\binom{2r-t'}{t-t'-1}$, respectively.
Consider the case where $t\geq r+1$. 
Then $2r-t'<t-1<t$ so these can be simplified to 
$0$, $\binom{2r-t'}{t-t'}$, $0$ and $\binom{2r-t'}{t-t'-1}$, respectively.
If $t'>t$ then all four coefficients are~$0$ so $w_t$ is not in Equation~$t'$.
If $t'=t$ then the final coefficient is~$0$, but the 2nd of these coefficients is~$1$, so
$w_t$  has a non-zero coefficient in Equation~$t$.
We conclude for $t\geq r+1$ that
 $w_t$ is not in Equations~$t+1,\ldots,k-1$,
but it has a non-zero coefficient in Equation~$t$.

Thus, Equations $k,k-1,\ldots,r+1$ 
give us $k-r=r+1$ linearly independent equations.
 Thus the solution space of the system 
over our $k+1$ variables
has dimension at most 
$k+1 - (r+1)=
r+1$. But we have already shown that Equation~\eqref{eq:stmt}
 gives a solution, and the dimension of this solution 
    is $r+1$.

Similarly, suppose that $k=2r$ is even.
Once again consider $1\leq t \leq t' \leq k-1$ and compute the 
coefficients of the variable $w_t$ in the
quantities $Z'_{0,t'}$, $Z''_{0,t'}$, $Z''_{1,t'}$ and $Z''_{1,t'}$.
These are
$\binom{2r-1-t'}{t}$, $\binom{2r-1-t'}{t-t'}$, $\binom{2r-1-t'}{t-1}$ and $\binom{2r-1-t'}{t-t'-1}$, respectively.
Consider the case where $t\geq r+1$.
Then $2r-1-t'<t-1<t$ so these can again be simplified to 
$0$, $\binom{2r-1-t'}{t-t'}$, $0$ and $\binom{2r-1-t'}{t-t'-1}$, respectively.
Once again, if $t'>t$ then all four coefficients are~$0$ so $w_t$ is not in Equation~$t'$.
However, if $t'=t$ then the final coefficient is~$0$, but the 2nd of these coefficients is~$1$, so
$w_t$  has a non-zero coefficient in Equation~$t$.
As before, we conclude  that the equations $k,k-1,\ldots,r+1$ 
give us $r$
linearly independent equations, so
the solution space of the system has dimension at most $(k+1)-(r)=r+1$.
 But $w_\ell=w_{k-\ell}$ is a solution whose dimension is $r+1$. \end{proof}

\subsection{Realising conditional distributions induced by pinning and equality}

We will use pinnings or equality to construct and analyze gadgets; the upcoming Lemma~\ref{lem:pineq} is a first step in doing this effortlessly. It asserts that when $f$ supports one of the properties pinning-to-0, pinning-to-1 or equality we can consider appropriate conditional distributions (depending on the property); these conditional distributions can then be realised using appropriately constructed hypergraphs.

Given a set~$S$ of vertices,
 it will be convenient to write $\sigma_S=\zeros$ to denote 
the event that all vertices in $S$ are assigned the spin 0 under 
the assignment $\sigma$.  We will similarly write
 $\sigma_S=\ones$. 
  We will also use $\sigma^{\eq}_S$ to denote 
  the event that all vertices in $S$ have the same spin under $\sigma$ 
  (the spin could be~$0$ or~$1$).
  
\begin{lemma}\label{lem:pineq}  
Let $f:\{0,1\}^k\rightarrow\Rplus$ be symmetric. Let $H=(V,\mathcal{F})$ be a hypergraph and $S\subseteq V$.  Let $\epsilon>0$.
\begin{enumerate}
\item \label{it:pinningfix} Suppose that $f$ supports pinning-to-$s$ for some $s\in\{0,1\}$. Suppose that $u\in V$ satisfies $\mu_{f;H}(\sigma_{u}=s)>0$. Then
there is a hypergraph~$H'=(V',\mathcal{F}')$  with  $V\subseteq V'$ and $\mathcal{F} \subseteq \mathcal{F}'$ such that  for every 
$\tau:S\rightarrow\{0,1\}$ it holds that 
\[\big|\mu_{f;H'}(\sigma_{S}=\tau)-\mu_{f;H}\big(\sigma_S=\tau\mid \sigma_{u}=s\big)\big|\leq \epsilon.\]
\item \label{it:equalityfix} Suppose that $f$ supports equality and that $R\subseteq V$ satisfies $\mu_{f;H}(\sigma^{\eq}_{R})>0$. Then
there is a hypergraph ~$H'=(V',\mathcal{F}')$  with  $V\subseteq V'$ and $\mathcal{F} \subseteq \mathcal{F}'$ such that 
   for every 
$\tau:S\rightarrow\{0,1\}$ it holds that  
\[\Big|\mu_{f;H'}(\sigma_{S}=\tau)- \mu_{f;H}\big(\sigma_S=\tau\mid \sigma^{\eq}_{R}\big)\Big|\leq \epsilon,\]
\end{enumerate}
\end{lemma}
\begin{proof}
Denote $\mu_{f;H}$ by~$\mu$. 
We begin with the proof of Item~\ref{it:pinningfix}. 
We will take $s=0$. The proof for $s=1$ is similar, swapping $0$'s and $1$'s.
Let $\epsilon'>0$ be
sufficiently small so that    
\begin{equation}
\label{eq:29th}
\left(\frac{\epsilon'}{1-\epsilon'} \right) \left(\frac{\mu(\sigma_u=1)}{\mu(\sigma_u=0)}\right) \leq \frac{\epsilon}{1-\epsilon}.
\end{equation}

 Since $f$ supports pinning-to-0, there exists a hypergraph $H_0$ with vertex set $V_0$ which $\epsilon'$-realises pinning-to-0. Denote by $v_0$ the terminal of $H_0$. For $i\in\{0,1\}$, let $p_{i}=\mu_{f;H_0}(\sigma_{v_0}=i)$, so that $p_0\geq 1-\epsilon'$, $p_1\leq \epsilon'$. To construct the hypergraph $H'$ (see Figure~\ref{fig:pin0}), take a (distinct)  copy of $H_0$ and identify $v_0$ with $u$ (the copy of $H_0$ is otherwise disjoint from the rest of $H$). Let $\mu':=\mu_{f;H'}$. Note that $\mu'(\sigma_u=s)\propto p_s\, \mu(\sigma_u=s)$ for $s\in\{0,1\}$, so that 
\[\frac{\mu'(\sigma_u=1)}{\mu'(\sigma_u=0)}=\left(\frac{p_1}{p_0}\right)\left( \frac{\mu(\sigma_u=1)}{\mu(\sigma_u=0)}\right).\]
Equation~\eqref{eq:29th} yields $\mu'(\sigma_u=1)\leq \epsilon$.

Let $\tau:S\rightarrow \{0,1\}$. Note that
\begin{align*}
\mu'(\sigma_S=\tau)&=\mu'(\sigma_u=0)\mu'(\sigma_S=\tau\mid \sigma_u=0)+\mu'(\sigma_S=\tau, \sigma_u=1)\\
&=\mu'(\sigma_u=0)\,\mu(\sigma_S=\tau\mid \sigma_u=0)+\mu'(\sigma_S=\tau, \sigma_u=1),
\end{align*}
where in the second equality we used that conditioned on the spin of $u$, $\sigma_S$ is independent of the configuration $\sigma_{V_0\setminus \{u\}}$. It follows that
\[|\mu'(\sigma_S=\tau)-\mu(\sigma_S=\tau\mid \sigma_u=0)|\leq \mu'(\sigma_u=1)\leq \epsilon\]
completing the proof of Item~\ref{it:pinningfix}.

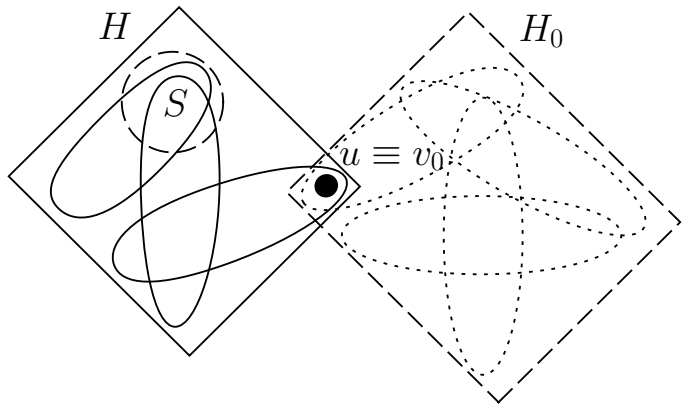
\begin{figure}[t]
  \centering
  \scalebox{0.7}[0.7]{\input{hypergraph3.tex}}
  \caption{The hypergraph $H'$ used in the proof of Item~\ref{it:pinningfix} in Lemma~\ref{lem:pineq}. It consists of the hypergraph $H$ and a hypergraph $H_0$ with terminal $v_0$ which $\epsilon'$-realises pinning-to-0. We identify the vertex $u$ of $H$ with the terminal vertex $v_0$ of $H_0$ keeping otherwise the two hypergraphs $H,H_0$ disjoint. Then, for any subset $S$ of vertices in $H$ (which may in general include $u$ as well), the distribution $\mu_{f;H'}(\sigma_S=\cdot)$ is well approximated by the conditional distribution $\mu_{f,H}(\sigma_S=\cdot\mid \sigma_{u}=0)$.}
\label{fig:pin0}
\end{figure}

The proof of Item~\ref{it:equalityfix} is completely analogous, though slightly more technical. Let $r=|R|$. As before, let $\epsilon'>0$ be sufficiently small, to be picked later.  Since $f$ supports equality, there exists a hypergraph $H_1$ with vertex set $V_1$ which $\epsilon'$-realises $r$-equality. Denote by $T=\{v_1,\hdots,v_r\}$ the set of terminals of $H_1$. For $\eta:T\rightarrow\{0,1\}$, let $p_{\eta}:=\mu_{f;H_1}(\sigma_{R}=\eta)$. Since $H_1$ is an  $\epsilon'$-realisation of  $r$-equality, we have that 
\[(1-\epsilon')/2\leq p_\zeros,p_\ones\leq (1+\epsilon')/2, \quad \sum_{\eta\neq \zeros,\ones}p_\eta\leq \epsilon'.\]

To construct the hypergraph $H'$, take a (distinct)  copy of $H_1$ and identify (in an arbitrary way) vertices in  the set $R$ with terminals in the set $T$ (the copy of $H_1$ is otherwise disjoint from the rest of $H$), we will keep the notation $R$ for the merged vertices in $H'$ (see Figure~\ref{fig:pineq}). Let $\mu':=\mu_{f;H'}$.  For $\eta:R\rightarrow\{0,1\}$, we  $\mu'(\sigma_R=\eta)\propto p_\eta\, \mu(\sigma_R=\eta)$, so that 
\[\mu'(\sigma_R=\eta)=\frac{p_\eta\, \mu(\sigma_R=\eta)}{p_\zeros\,\mu(\sigma_R=\zeros)+p_\ones\,\mu(\sigma_R=\ones)+\sum_{\eta'\neq\zeros,\ones}p_{\eta'}\, \mu(\sigma_R=\eta')}.\]
For $\mathbf{s}\in\{\zeros,\ones\}$, 
we use the upper and lower bounds on $p_{\mathbf{s}}$ to obtain
\begin{align} 
\mu'(\sigma_R=\mathbf{s})&\leq \frac{p_{\mathbf{s}}\, \mu(\sigma_R={\mathbf s})}{p_\zeros\, \mu(\sigma_R=\zeros)+p_\ones\,\mu(\sigma_R=\ones)}\leq 
\left(\frac{1+\epsilon'}{1-\epsilon'}\right)
\left(\frac{\mu(\sigma_R={\mathbf{s}})}{\mu(\sigma_R=\zeros)+\mu(\sigma_R=\ones) }\right),
\label{eq:uppermu}\\
\mu'(\sigma_R={\mathbf{s}})&\geq \frac{p_{\mathbf{s}}\, \mu(\sigma_R={\mathbf{s}})}{p_\zeros\,\mu(\sigma_R=\zeros)+p_\ones\,\mu(\sigma_R=\ones)+\epsilon'}\geq 
\left(\frac{1-\epsilon'}{1+\epsilon'}\right)
\left(\frac{\mu(\sigma_R={\mathbf{s}})}{\mu(\sigma_R=\zeros)+\mu(\sigma_R=\ones)+ \frac{2\epsilon'}{1+\epsilon'}}\right).
\label{eq:lowermu}
\end{align}

From \eqref{eq:uppermu} and \eqref{eq:lowermu}, we obtain that for ${\mathbf{s}}\in\{\zeros,\ones\}$, as $\epsilon'\downarrow 0$, it holds that 
\begin{equation}\label{eq:sigmaRs}
\mu'(\sigma_R={\mathbf{s}})\rightarrow \frac{\mu(\sigma_R=\mathbf{s})}{\mu(\sigma_R=\zeros)+\mu(\sigma_R=\ones)}=\mu(\sigma_R={\mathbf{s}}\mid \sigma_R^{\eq}), \quad \mbox{ and thus } \mu'(\neg \sigma_R^{\eq})\rightarrow 0,
\end{equation} 
where the latter limit follows by observing that  $\mu'(\neg \sigma_R^{\eq})=1-\mu'(\sigma_R=\zeros)-\mu'(\sigma_R=\ones)$.

Let $\tau:S\rightarrow \{0,1\}$. Note that
\begin{align}
\mu'(\sigma_S=\tau)&=\mu'(\sigma_R=\zeros)\mu'(\sigma_S=\tau\mid \sigma_R=\zeros)+\notag\\
&\hskip 2cm \mu'(\sigma_R=\ones)\mu'(\sigma_S=\tau\mid \sigma_R=\ones)+\mu'(\sigma_S=\tau, \neg\sigma_R^\eq)\notag\\
&=\mu'(\sigma_R=\zeros)\mu(\sigma_S=\tau\mid \sigma_R=\zeros)+\notag\\
&\hskip 2cm \mu'(\sigma_R=\ones)\mu(\sigma_S=\tau\mid \sigma_R=\ones)+\mu'(\sigma_S=\tau, \neg\sigma_R^\eq),\label{eq:core}
\end{align}
where again in the second equality we used that conditioned on $\sigma_R$, $\sigma_S$ is independent of the configuration $\sigma_{V_1\setminus R}$. Using the limits in \eqref{eq:sigmaRs} and the equality \eqref{eq:core}, it is not hard to see that as $\epsilon'\downarrow 0$, it holds that
\begin{equation}\label{eq:finallimit}
\mu'(\sigma_S=\tau)\rightarrow \mu(\sigma_R=\zeros\mid \sigma_R^{\eq})\,\mu(\sigma_S=\tau\mid \sigma_R=\zeros)+\mu(\sigma_R=\ones\mid \sigma_R^{\eq})\,\mu(\sigma_S=\tau\mid \sigma_R=\ones).
\end{equation}
The right-hand side in \eqref{eq:finallimit} is equal to $\mu(\sigma_S=\tau\mid \sigma_R^{\eq})$, from where it follows that by choosing small $\epsilon'$, the hypergraph $H'$ satisfies $\big|\mu'(\sigma_{S}=\tau)- \mu\big(\sigma_S=\tau\mid \sigma^{\eq}_{R}\big)\big|\leq \epsilon$, as wanted.

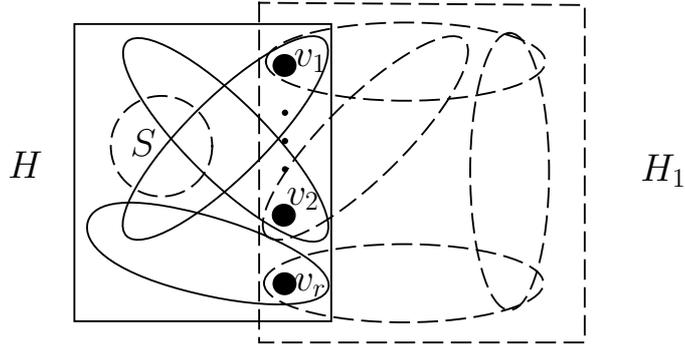
\begin{figure}[t]
  \centering
  \scalebox{0.7}[0.7]{\input{hypergraph4.tex}}
  \caption{The hypergraph $H'$ used in the proof of Item~\ref{it:equalityfix}  in Lemma~\ref{lem:pineq}. It consists of the hypergraph $H$ and a hypergraph $H_1$ which realises $r$-equality among its terminals  $v_1,\hdots,v_r$. We identify the vertices in $R$ with the terminal vertices $v_1,\hdots,v_r$ of $H_1$ keeping otherwise the two hypergraphs $H,H_1$ disjoint. Then, for any subset $S$ of vertices in $H$ (which may in general include any of the vertices $v_1,\hdots,v_r$), the distribution $\mu_{f;H'}(\sigma_S=\cdot)$ is well approximated by the conditional distribution $\mu_{f,H}(\sigma_S=\cdot\mid \sigma_R^{\eq})$.}
\label{fig:pineq}
\end{figure}

This concludes the proof of Lemma~\ref{lem:pineq}.
\end{proof}

Typically, when $f$ satisfies, say, pinning-to-0 we will be interested in ``pinning" more than one vertex to zero, while Item~\ref{it:pinningfix} from Lemma~\ref{lem:pineq} accomodates only one vertex. We will also be interested in cases where $f$ satisfies  multiple properties amongst pinning-to-0, pinning-to-1 and equality. We extend Lemma~\ref{lem:pineq} to address this more general framework. Prior to that, it will be useful for our applications to set up a convenient notation.

\begin{definition}\label{conditioned}
Let $f:\{0,1\}^k\rightarrow\Rplus$ be symmetric. Let $H=(V,\mathcal{F})$ be a hypergraph and assume that  $\mathcal{V}:=(V_0,V_1,V_2,\hdots,V_r)$ is a labelled collection of disjoint subsets of $V$ such that: (i) $V_0=\emptyset$ if $f$ does not support pinning-to-0, (ii) $V_1=\emptyset$ if $f$ does not support pinning-to-1, (iii) $V_2=\hdots=V_r=\emptyset$ if $f$ does not support equality, (iv) it holds that $\mu_{f;H}(\sigma_{V_0}=\mathbf{0},\sigma_{V_1}=\mathbf{1}, \sigma_{V_2}^{\eq},\hdots,\sigma_{V_r}^{\eq})>0$. 
We will then say that $\mathcal{V}$ is \emph{admissible}  for the hypergraph $H$ (with respect to $f$) and denote by $\mu_{f;H}^{\condV}$ the probability distribution $\mu_{f;H}(\cdot \mid \sigma_{V_0}=\mathbf{0},\sigma_{V_1}=\mathbf{1}, \sigma_{V_2}^{\eq},\hdots,\sigma_{V_r}^{\eq})$.
\end{definition}

Definition~\ref{conditioned} provides the framework to prove the following generalization of Lemma~\ref{lem:pineq}. In simple words, we will now be able to combine the conditional distributions that we can realize via hypergraphs using $f$.

\begin{lemma}\label{lem:generalpin}
Let $f:\{0,1\}^k\rightarrow\Rplus$ be a symmetric function. Let $H$ be a hypergraph with vertex set $V$ 
and let $S$ be a subset of~$V$.
Let $\mathcal{V}$ an admissible collection of subsets of $V$ with respect to $H$.  
Then, for every $\epsilon>0$, there is a 
hypergraph~$H'=(V',\mathcal{F}')$  with  $V\subseteq V'$ and $\mathcal{F} \subseteq \mathcal{F}'$ such that, for every $\tau:S\rightarrow \{0,1\}$, it holds that 
\[\big|\mu_{f;H'}(\sigma_S=\tau)-\mu_{f;H}^{\condV}(\sigma_S=\tau)\big|\leq \epsilon,\]
where $\mu_{f;H}^{\condV}(\cdot)$ is as in Definition~\ref{conditioned}.
\end{lemma}
\begin{proof}
To treat the different cases which arise as uniformly as possible, it will be convenient to use the following notation for the purposes of this proof. Let $\sigma:V\rightarrow \{0,1\}$ and $X\subseteq V$.  For $\mathbf{s}\in \{\zeros,\ones\}$, we will write $\sigma_X^{\mathbf{s}}$ as an alternative notation to $\sigma_X=\mathbf{s}$.

We are now set to prove the claim. Suppose that $V=(V_0,V_1,V_2,\hdots,V_r)$. We proceed by induction on $W:=|V_0|+|V_1|+(r-1)$. The base case $W=0$ is trivial --- we can take $H'=H$. 
To carry out the induction step, assume that the claim holds when $W=t$, we show it when $W=t+1$. We have that at least one of the following holds: (i) $|V_0|\geq 1$, (ii) $|V_1|\geq 1$, (iii) $r\geq 2$. For each of these cases, we will have a collection of sets $\overline{\mathcal{V}}$ for which we will invoke the inductive hypothesis and a set $X$ which we wish to add to $\overline{\mathcal{V}}$ to conclude the claim for $\mathcal{V}$. Also, we want to be able to condition on the event $\sigma_X^{\mathbf{s}}$ for some $\mathbf{s}\in\{\zeros,\ones,\eq\}$ (the value of $\mathbf{s}$ will depend on the case that we consider). More precisely, the set $X$, the value of $\mathbf{s}\in\{\zeros,\ones,\eq\}$ and the collection of sets $\overline{\mathcal{V}}$ are specified as follows for the respective cases: 
\begin{itemize}
\item $X=\{v\}$ for some $v\in V_0$, $\mathbf{s}=\zeros$, and $\overline{\mathcal{V}}=(V_0\setminus X,V_1,V_2,\hdots,V_r)$, \item $X=\{v\}$ for some $v\in V_1$, $\mathbf{s}=\ones$, and $\overline{\mathcal{V}}=(V_0,V_1\setminus X,V_2,\hdots,V_r)$, 
\item $X=V_r$, $\mathbf{s}=\eq$, and $\overline{\mathcal{V}}=(V_0,V_1,V_2,\hdots,V_{r-1})$.
\end{itemize}
Note that since $\mathcal{V}$ is admissible for $H$ under $f$, the same is trivially true for $\overline{\mathcal{V}}$ as well.

Let $S'$ be an arbitrary subset of $V$ (which may be our target set $S$). By induction, for the collection of subsets $\overline{\mathcal{V}}$, we have that for every $\epsilon>0$, there exists a hypergraph $H'$ with 
$V\subseteq V'$ and $\mathcal{F} \subseteq \mathcal{F}'$ such that, for every $\tau':S'\rightarrow \{0,1\}$, it holds that 
\begin{equation}\label{dealwithzero}
\big|\mu_{f;H'}(\sigma_{S'}=\tau')-\mu_{f;H}^{\mathrm{cond}(\overline{\mathcal{V}})}(\sigma_{S'}=\tau')\big|\leq \epsilon,
\end{equation}

Since $\mu(\sigma^{0}_{V_0},\sigma^{1}_{V_1},\sigma^{\eq}_{V_2},\hdots,\sigma^{\eq}_{V_r})>0$ (by the assumption that $\mathcal{V}$ is admissible), we also have that $M:=\mu_{f;H}^{\mathrm{cond}(\overline{\mathcal{V}})}(\sigma_X^{\mathbf{s}})>0$, where $X,\mathbf{s},\overline{\mathcal{V}}$ were defined above. Let 
\begin{equation*}
\epsilon_1:=\epsilon/2,\quad \epsilon_2:=\epsilon_1 M^2/4, \quad \epsilon_3:=\epsilon_2/2^{|S|}.
\end{equation*}

Now consider $\tau: S \to \{0,1\}$.
 if $\mu_{f;H}(\sigma_S=\tau,\sigma_X^{\mathbf{s}})=0$ then     $ \mu_{f;H}^{\mathrm{cond}(\overline{\mathcal{V}})}\big(\sigma_{S}=\tau,\sigma_X^{\mathbf{s}}) =0$.
 Also, since $V \subseteq V'$ and $\mathcal{F} \subseteq \mathcal{F}'$.
 $\mu_{f;H'}\big(\sigma_{S}=\tau,\sigma_{X}^{\mathbf{s}}\big)
= 0$. 

Otherwise, 
$\mu_{f;H}(\sigma_S=\tau,\sigma_X^{\mathbf{s}})>0$. In this case,  applying \eqref{dealwithzero} with $S'=S\cup X$
we find that  there is a hypergraph $H'=(V',\mathcal{F}')$ with vertex set $V \subseteq V'$
and $\mathcal{F} \subseteq \mathcal{F}'$
such that for each assignment $\tau': S' \to \{0,1\}$ that
is consistent with $\tau'_S = \tau$ and ${\tau'}_X^{\mathbf{s}}$,
we have 
$$\big|\mu_{f;H'}(\sigma_{S'}=\tau')-\mu_{f;H}^{\mathrm{cond}(\overline{\mathcal{V}})}(\sigma_{S'}=\tau')\big|\leq \epsilon.$$
Summing over all $\tau':S' \to \{0,1\}$ that satisfy both $\tau'_S=\tau$ and ${\tau'}_X^{\mathbf{s}}$, we conclude that 
 \begin{equation}\label{eq:sigmasUu}
\big|\mu_{f;H'}\big(\sigma_{S}=\tau,\sigma_{X}^{\mathbf{s}}\big)-\mu_{f;H}^{\mathrm{cond}(\overline{\mathcal{V}})}\big(\sigma_{S}=\tau,\sigma_{X}^{\mathbf{s}}\big)\big|\leq \epsilon_3.
\end{equation}

So now we have established~\eqref{eq:sigmasUu} for all $\tau: S \to \{0,1\}$
(whether $\mu_{f;H}(\sigma_S=\tau,\sigma_X^{\mathbf{x}})$ is~$0$ or not).
Summing over  $\tau: S \to \{0,1\}$, we obtain
\begin{equation}\label{eq:sigmau}
\big|\mu_{f;H'}\big(\sigma_{X}^{\mathbf{s}}\big)-\mu_{f;H}^{\mathrm{cond}(\overline{\mathcal{V}})}\big(\sigma_{X}^{\mathbf{s}})\big|\leq 2^{|S|}\epsilon_3=\epsilon_2.
\end{equation} 
This implies $\mu_{f;H'}(\sigma_{X}^{\mathbf{s}})\geq M-\epsilon_2>M/2>0$. 

Now  to simplify the notation 
  define
$\A:=\mu_{f;H}^{\mathrm{cond}(\overline{\mathcal{V}})}\big(\sigma_{S}=\tau,\sigma_{X}^{\mathbf{s}}\big)$.
Now note that
$$\mu^{\mathrm{cond}(\mathcal{V})}_{f;H}\big(\sigma_{S}=\tau\big)=\A/M.$$
Let $A' = \mu_{f;H'}(\sigma_S=\tau,\sigma_X^{\mathbf{s}})$
and let $M' = \mu_{f;H'}(\sigma_X^{\mathbf{s}})$. Equation~$\eqref{eq:sigmasUu}$  
(together with $\epsilon_3 \leq \epsilon_2$) shows
that $A' - \epsilon_2 \leq A \leq  A' +\epsilon_2$. 
Also, Equation~\eqref{eq:sigmau} shows $M' - \epsilon_2 \leq M \leq M' + \epsilon_2$.
Then we will use 
the bound 
$$
\max\left(
\frac{A + \epsilon_2}{M - \epsilon_2} - \frac{A}{M} \>,\>
\frac{A}{M} - \frac{A - \epsilon_2}{M + \epsilon_2}
\right)
 \leq \frac{\epsilon_2(M+\A)}{M(M-\epsilon_2)}\leq \frac{2\epsilon_2}{M(M-\epsilon_2)}\leq 4\epsilon_2/M^2=\epsilon_1$$
to conclude that
\begin{equation}\label{eq:asdf}
\left|
\frac{A'}{M'} - \frac{A}{M}
\right| = 
\left|\mu_{f;H'}\big(\sigma_{S}=\tau\mid \sigma_{X}^{\mathbf{s}}\big)-\mu^{\mathrm{cond}({\mathcal{V}})}_{f;H}\big(\sigma_{S}=\tau\big)\right|\leq \epsilon_1.
\end{equation}

We will now apply Lemma~\ref{lem:pineq}
to~$H'$ and~$S$ with error parameter $\epsilon_1$. 
To apply the lemma, we need the fact that $f$ supports pinning-to-$s$ if $\mathbf{s}$ is $\zeros$ or $\ones$.
If $\mathbf{s} = \eq$ then we need the fact that $f$ supports equality.
Both of these follow from the admissibility of $\mathcal{V}$ and the construction of~$X$ and
$\mathbf{s}$.
We aso need the fact that $\mu_{f;H'}(\sigma_X^{\mathbf{s}}) > 0$, which we have established above.
Then
Lemma~\ref{lem:pineq} shows that  there exists a hypergraph $H''=(V'',\mathcal{F}'')$ 
with $V'\subseteq V''$ and $\mathcal{F}' \subseteq \mathcal{F''}$ such that for every 
$\tau:S\rightarrow \{0,1\}$, it holds that 
\begin{equation}\label{eq:sigmaS2}
\big|\mu_{f;H''}(\sigma_{S}=\tau)-\mu_{f;H'}\big(\sigma_{S}=\tau\mid \sigma_{X}^{\mathbf{s}}\big)\big|\leq \epsilon_1.
\end{equation}
It follows by \eqref{eq:asdf} and \eqref{eq:sigmaS2} that 
\begin{equation*}
\big|\mu_{f;H''}(\sigma_{S}=\tau)-\mu^{\condV}_{f;H}\big(\sigma_{S}=\tau\big)\big|\leq 2\epsilon_1=\epsilon.
\end{equation*}
This completes the induction.
\end{proof}

\section{A general inapproximability lemma}
The purpose of this section is to prove the following lemma, which will allow us to exploit our study of pinning-to-0, pinning-to-1 and equality.

\newcommand{\statelemnonumhardness}{
Let $f:\{0,1\}^k\rightarrow\Rplus$ be symmetric. Let $H$ be a hypergraph, let $\mathcal{V}$ be admissible for $H$ and let~$x$ and~$y$  be vertices of $H$. For $s_1,s_2\in\{0,1\}$, 
define $\Zmu_{s_1s_2}$ by
$$\Zmu_{s_1s_2}:=\Zmu_{f;H}^{\condV}(\sigma(x)=s_1,\sigma(y)=s_2).$$ 
Suppose that all of the following hold:
\begin{equation*} 
\begin{gathered}
\quad \Zmu_{00}+\Zmu_{11}>0,\ \min\{\Zmu_{00},\Zmu_{11}\}<\sqrt{\Zmu_{01}\Zmu_{10}},\ \max\{\Zmu_{00},\Zmu_{11}\}\leq \sqrt{\Zmu_{01}\Zmu_{10}}.
\end{gathered}
\end{equation*}
Then, for all sufficiently large $\Delta$, there exists $c>1$ such that $\Hyper2Spinf$ is $\NP$-hard. 
}

\newcommand{\statelemnumhardness}{
Let $f:\{0,1\}^k\rightarrow\Rplus$ be symmetric. Let $H$ be a hypergraph, let $\mathcal{V}$ be admissible for $H$ and let~$x$ and~$y$  be vertices of $H$. For $s_1,s_2\in\{0,1\}$, 
define $\Zmu_{s_1s_2}$ by
$$\Zmu_{s_1s_2}:=\Zmu_{f;H}^{\condV}(\sigma(x)=s_1,\sigma(y)=s_2).$$ 
Suppose that all of the following hold:
\begin{equation}\label{eq:antiferro}
\begin{gathered}
\quad \Zmu_{00}+\Zmu_{11}>0,\ \min\{\Zmu_{00},\Zmu_{11}\}<\sqrt{\Zmu_{01}\Zmu_{10}},\ \max\{\Zmu_{00},\Zmu_{11}\}\leq \sqrt{\Zmu_{01}\Zmu_{10}}.
\end{gathered}
\end{equation}
Then, for all sufficiently large $\Delta$, there exists $c>1$ such that $\Hyper2Spinf$ is $\NP$-hard. 
}

\begin{lemma}\label{hardness}
\statelemnonumhardness
\end{lemma}

 The proof of Lemma~\ref{hardness} uses inapproximability results for antiferromagnetic 2-spin systems on bounded-degree graphs. Thus, before giving its proof, it will be helpful to make a detour to extract the results that will be useful in the proof of Lemma~\ref{hardness}.
\begin{remark}
Note that the inequalities for the $\mu_{ij}$'s are stronger than the standard antiferromagnetic condition $\mu_{00}\mu_{11}< \mu_{01}\mu_{10}$ for 2-spin models on graphs. This is to ensure that the corresponding 2-spin system lies in the non-uniqueness region for all sufficiently large $\Delta$ (and hence is intractable). In fact, if $\max\{\mu_{00},\mu_{11}\}>\sqrt{\mu_{01}\mu_{10}}$, for the corresponding binary 2-spin system (even if it is antiferromagnetic), approximating its partition function may be tractable for all graphs (when the external field is fixed).
\end{remark}

\subsection{Inapproximability for antiferromagnetic 2-spin systems}
\label{sec:binary}

We review inapproximability results for the partition function of   antiferromagnetic 2-spin models on graphs. We start with  a few relevant definitions following \cite{LLY}. A 2-spin model on a graph is specified by three parameters $\beta,\gamma\geq 0$ and $\lambda>0$. For a graph $G=(V,E)$, configurations of the model are all possible assignments $\sigma:V\rightarrow \{0,1\}$ and the partition function is given by
\[Z_{\beta,\gamma,\lambda;G}=\sum_{\sigma:V\rightarrow\{0,1\}}\lambda^{|\sigma^{-1}(0)|}\prod_{(u,v)\in E}\beta^{\mathbf{1}\{\sigma(u)=\sigma(v)=0\}}\gamma^{\mathbf{1}\{\sigma(u)=\sigma(v)=1\}},\]
where we adopt the convention that $0^0\equiv 1$ when one of the parameters $\beta,\gamma$ is equal to zero. The case $\beta=\gamma$ corresponds to the Ising model, while the case $\beta=0$ and $\gamma=1$ corresponds to the hard-core model.

The 2-spin system with parameters $\beta,\gamma, \lambda$ is called \textit{antiferromagnetic} if $\beta \gamma<1$.  In \cite{SlySun}, it was shown that the computational hardness of approximating the partition function in antiferromagnetic 2-spin systems on $\Delta$-regular graphs is captured by the so-called uniqueness threshold on the infinite $\Delta$-regular tree.  More precisely, we have the following. 

\begin{theorem}[{\cite[Theorems 2 \& 3]{SlySun}}]\label{thm:slysun}
Consider the 2-spin systems specified by the following parameters $\beta,\gamma,\lambda$: (i) $0<\beta=\gamma<1$, $\lambda>0$ (antiferromagnetic Ising model), (ii) $\beta=0,\gamma=1$, $\lambda>0$ (hard-core model). If the 2-spin system specified by the parameters $\beta,\gamma,\lambda$ is in the non-uniqueness regime of the infinite $\Delta$-regular tree for $\Delta \geq 3$, then
there is a $c>1$ such that  it is $\NP$-$\hard$ to approximate $Z_{\beta,\gamma,\lambda;G}$ within 
a factor of $c^{n}$ on the class of $\Delta$-regular graphs $G$.
\end{theorem}

Note that it is not important for us that Theorem~\ref{thm:slysun} shows hardness all the
way to the uniqueness threshold. It would suffice to have a weaker bound for the 
antiferrogmanetic Ising model and the hard-core model. Luby and Vigoda~\cite{LV} provide such a result for the hard-core model. Rather than explicitly deriving such a bound for the antiferromagnetic Ising model, we use Theorem~\ref{thm:slysun}. Cai et al \cite[Theorems 1 \& 2]{Cai} give similar results that apply in some of the relevant parameter space, however it is simpler to work with the later paper \cite{SlySun}, especially since the latter shows that the partition function is hard to approximate even within an exponential factor.

For future use, we point out the following characterisation of the uniqueness regime on the infinite  $\Delta$-regular tree. For a 2-spin system with parameters $\beta,\gamma,\lambda$, let $h(x):=\lambda\big(\frac{\beta x+1}{x+\gamma}\big)^{\Delta-1}$ and let $x^*$ be the (unique) positive solution of $x^*=h(x^*)$. Then, uniqueness holds on the infinite  $\Delta$-regular tree iff $|h'(x^*)|\leq 1$, i.e., the absolute value of the derivative of $h$ evaluated at $x^*$ is less than or equal than 1. When, instead, $|h'(x^*)|> 1$,  non-uniqueness holds on the infinite  $\Delta$-regular tree. Equivalently, one can derive the following equivalent criterion: non-uniqueness on the infinite $\Delta$-regular tree holds iff
 the system of equations 
\begin{equation}\label{eq:wed}
x=\lambda\Big(\frac{\beta y+1}{y+\gamma}\Big)^{\Delta-1},\quad y=\lambda\Big(\frac{\beta x+1}{x+\gamma}\Big)^{\Delta-1}
\end{equation}
has multiple (i.e., more than one) positive solutions $(x,y)$.

It is well-known (see, e.g., \cite{SlySun}) that antiferromagnetic 2-spin systems on $\Delta$-regular graphs  
can be expressed in terms of either the Ising model  or the hard-core model. Theorem~\ref{thm:slysun} thus also gives the regime where general antiferromagnetic 2-spin systems are hard, albeit somewhat implicitly. For the sake of completeness we do this explicitly in the following simple corollary of Theorem~\ref{thm:slysun} (which is nevertheless lengthy to prove).

\begin{corollary}\label{thm:isinghardness}
Let $\beta,\gamma\geq 0$ with $\beta \gamma<1$, $\gamma>0$, $\lambda>0$ and $\Delta\geq 3$. If the 2-spin system specified by the parameters $\beta,\gamma,\lambda$ is in the non-uniqueness regime of the infinite $\Delta$-regular tree, then
there is a $c>1$ such that  it is $\NP$-$\hard$ to approximate $Z_{\beta,\gamma,\lambda;G}$ within  
a factor of $c^n$ on the class of
 $\Delta$-regular graphs $G$.
\end{corollary}
\begin{proof}
For $\beta,\gamma,\lambda,\Delta$ as in the statement of the lemma, consider the following map:
\[R(\beta,\gamma,\lambda)=\begin{cases} \big(\sqrt{\beta \gamma},\sqrt{\beta \gamma},\lambda (\beta/\gamma)^{\Delta/2}\big) , & \mbox{ if } \beta>0,\\ \big(0, 1, \lambda/\gamma^{\Delta}\big), & \mbox{ if } \beta=0.\end{cases}\]
To prove the claim, it suffices to show the following two  facts and then to use Theorem~\ref{thm:slysun}.
\begin{enumerate}[{Fact}~1.]
\item \label{equiv:approx} For a $\Delta$-regular graph $G$, a multiplicative approximation of $Z_{\beta,\gamma,\lambda;G}$ within a factor $C$ yields a multiplicative approximation of $Z_{R(\beta,\gamma,\lambda);G}$ within a factor $C$.
\item \label{equiv:uniq}  The 2-spin system with parameters $\beta,\gamma,\lambda$ is in the non-uniqueness regime of the infinite $\Delta$-regular tree iff the 2-spin system with parameters $R(\beta,\gamma,\lambda)$ is in the non-uniqueness regime of the infinite $\Delta$-regular tree.
\end{enumerate}
We consider first the case $\beta>0$. Let $\lambda'$ be defined from $(\lambda')^{1/\Delta}=\lambda^{1/\Delta}\sqrt{\beta}/\sqrt{\gamma}$ and let $\beta'=\sqrt{\beta \gamma}$. Note that $R(\beta,\gamma,\lambda)=(\beta',\beta',\lambda')$. 

We first show  Fact~\ref{equiv:approx}. Let $G=(V,E)$ be a $\Delta$-regular graph.  Observe that
\begin{align*}
Z_{\beta,\gamma,\lambda;G}&=\sum_{\sigma:V\rightarrow\{0,1\}}\lambda^{|\sigma^{-1}(0)|}\prod_{(u,v)\in E}\beta^{\mathbf{1}\{\sigma(u)=\sigma(v)=0\}}\gamma^{\mathbf{1}\{\sigma(u)=\sigma(v)=1\}}\\
&=\sum_{\sigma:V\rightarrow\{0,1\}}\prod_{(u,v)\in E}(\beta\lambda^{2/\Delta})^{\mathbf{1}\{\sigma(u)=\sigma(v)=0\}}(\lambda^{1/\Delta})^{\mathbf{1}\{\sigma(u)\neq \sigma(v)\}}\gamma^{\mathbf{1}\{\sigma(u)=\sigma(v)=1\}}\\
&=\Big(\frac{\sqrt{\gamma}}{\sqrt{\beta}}\Big)^{|E|}\sum_{\sigma:V\rightarrow\{0,1\}}\prod_{(u,v)\in E}\big(\beta'(\lambda')^{2/\Delta}\big)^{\mathbf{1}\{\sigma(u)=\sigma(v)=0\}}\big((\lambda')^{1/\Delta}\big)^{\mathbf{1}\{\sigma(u)\neq \sigma(v)\}}\big(\beta'\big)^{\mathbf{1}\{\sigma(u)=\sigma(v)=1\}}\\
&=\Big(\frac{\sqrt{\gamma}}{\sqrt{\beta}}\Big)^{|E|}\sum_{\sigma:V\rightarrow\{0,1\}}(\lambda')^{|\sigma^{-1}(0)|}\prod_{(u,v)\in E}\big(\beta'\big)^{\mathbf{1}\{\sigma(u)=\sigma(v)\}}\\
&=\Big(\frac{\sqrt{\gamma}}{\sqrt{\beta}}\Big)^{|E|}Z_{\beta',\beta',\lambda';G},
\end{align*}
which clearly yields the desired fact.

We next show Fact~\ref{equiv:uniq}.  Note that for positive $x,y$, \eqref{eq:wed} is equivalent to 
\begin{equation}\label{eq:wedwed}
x\Big(\frac{\beta y+1}{y+\gamma}\Big)=\lambda\Big(\frac{\beta y+1}{y+\gamma}\Big)^{\Delta},\quad y\Big(\frac{\beta x+1}{x+\gamma}\Big)=\lambda\Big(\frac{\beta x+1}{x+\gamma}\Big)^{\Delta}
\end{equation}
Thus, it suffices to show that  the positive solutions $(x,y)$ of \eqref{eq:wedwed} are in one-to-one correspondence with positive solutions $(x',y')$ to 
\begin{equation}\label{eq:wedwedtwo}
x'\Big(\frac{\beta' y'+1}{y'+\beta'}\Big)=\lambda'\Big(\frac{\beta' y'+1}{y'+\beta'}\Big)^{\Delta},\quad y'\Big(\frac{\beta' x'+1}{x'+\beta'}\Big)=\lambda'\Big(\frac{\beta' x'+1}{x'+\beta'}\Big)^{\Delta},
\end{equation}
where $\beta'=\sqrt{\beta \gamma}$ and $\lambda'=\lambda(\sqrt{\beta}/\sqrt{\gamma})^{\Delta}$ are as before. 

Let $x'=(\sqrt{\beta}/\sqrt{\gamma})x$ and $y'=(\sqrt{\beta}/\sqrt{\gamma})y$. Note that for $\beta>0$ we have that $x,y$ are positive iff $x',y'$ are positive and that $x,y$ are in one-to-one correspondence with $x',y'$.  It is also simple to verify that $x,y,x',y'$ satisfy the equations
\begin{equation}\label{eq:qazwsx}
\frac{\beta x+1}{x+\gamma}=\frac{\sqrt{\beta}}{\sqrt{\gamma}}\cdot \frac{\sqrt{\beta \gamma}x'+1}{x'+\sqrt{\beta \gamma}},\quad \frac{\beta y+1}{y+\gamma}=\frac{\sqrt{\beta}}{\sqrt{\gamma}}\cdot \frac{\sqrt{\beta \gamma}y'+1}{y'+\sqrt{\beta \gamma}},
\end{equation}
and 
\begin{equation}\label{eq:qazwsx2}
x\Big(\frac{\beta y+1}{y+\gamma}\Big)=x'\Big( \frac{\sqrt{\beta \gamma}y'+1}{y'+\sqrt{\beta \gamma}}\Big),\quad y\Big(\frac{\beta x+1}{x+\gamma}\Big)=y'\Big( \frac{\sqrt{\beta \gamma}x'+1}{x'+\sqrt{\beta \gamma}}\Big).
\end{equation}
Using \eqref{eq:qazwsx} and \eqref{eq:qazwsx2}, one can easily check that $x,y$ satisfy \eqref{eq:wedwed} iff $x',y'$ satisfy \eqref{eq:wedwedtwo}.

Next consider the case $\beta=0$. The arguments are completely analogous to the case $\beta>0$, up to minor technical details. Let $\lambda'$ be defined from $(\lambda')^{1/\Delta}=\lambda^{1/\Delta}/\gamma$ and note that $R(\beta,\gamma,\lambda)=(0,1,\lambda')$. 

For  Fact~\ref{equiv:approx}, we have that 
\begin{align*}
Z_{\beta,\gamma,\lambda;G}&=\sum_{\sigma:V\rightarrow\{0,1\}}\lambda^{|\sigma^{-1}(0)|}\prod_{(u,v)\in E}\beta^{\mathbf{1}\{\sigma(u)=\sigma(v)=0\}}\gamma^{\mathbf{1}\{\sigma(u)=\sigma(v)=1\}}\\
&=\sum_{\sigma:V\rightarrow\{0,1\}}\prod_{(u,v)\in E}(\beta\lambda^{2/\Delta})^{\mathbf{1}\{\sigma(u)=\sigma(v)=0\}}(\lambda^{1/\Delta})^{\mathbf{1}\{\sigma(u)\neq \sigma(v)\}}\gamma^{\mathbf{1}\{\sigma(u)=\sigma(v)=1\}}\\
&=\gamma^{|E|}\sum_{\sigma:V\rightarrow\{0,1\}}\prod_{(u,v)\in E}\big(\beta\gamma(\lambda')^{2/\Delta}\big)^{\mathbf{1}\{\sigma(u)=\sigma(v)=0\}}\big((\lambda')^{1/\Delta}\big)^{\mathbf{1}\{\sigma(u)\neq \sigma(v)\}}\\
&=\gamma^{|E|}\sum_{\sigma:V\rightarrow\{0,1\}}(\lambda')^{|\sigma^{-1}(0)|}\prod_{(u,v)\in E}\big(\beta \gamma\big)^{\mathbf{1}\{\sigma(u)=\sigma(v)=0\}}\\
&=\gamma^{|E|}Z_{0,1,\lambda';G},
\end{align*}
which again yields the desired  fact.

For  Fact~\ref{equiv:uniq},  set $x=\gamma x'$ and $y=\gamma y'$. The following equivalence is easy to see by inspection: $x,y$ satisfy 
\begin{equation}\label{eq:hardw}
\frac{x}{y+\gamma}=\lambda\Big(\frac{1}{y+\gamma}\Big)^{\Delta},\quad \frac{y}{x+\gamma}=\lambda\Big(\frac{1}{x+\gamma}\Big)^{\Delta}
\end{equation}
iff $x',y'$ satisfy
\begin{equation}\label{eq:hardw2}
\frac{x'}{y'+1}=\lambda'\Big(\frac{1}{y'+1}\Big)^{\Delta},\quad \frac{y'}{x'+1}=\lambda'\Big(\frac{1}{x'+1}\Big)^{\Delta}.
\end{equation}
Note that \eqref{eq:hardw} and \eqref{eq:hardw2} correspond to  equation \eqref{eq:wedwed} for the 2-spin systems with parameters $\beta=0,\gamma,\lambda$ and $\beta=0,\gamma=1,\lambda'$, respectively. 

This concludes the proof of the corollary.
\end{proof}

For us, the case $\lambda=1$ (which is usually referred to as the  case without an external field) 
will be especially important. To motivate what follows, the reader should first bear in mind the following two facts~\cite[Lemma 21]{LLY}
about the uniqueness regime for antiferromagnetic 2-spin systems. 
The two cases correspond to whether or not one of the parameters $\beta,\gamma$ is larger than 1. These parameters cannot both be larger than~1  because of the antiferromagnetic condition $\beta \gamma < 1$. 
\begin{enumerate}
\item \label{it:rew} when $\beta$ and $\gamma$  satisfy $0\leq \beta<1$ and $0<\gamma\leq 1$, non-uniqueness holds on the infinite $\Delta$-regular tree for all sufficiently large $\Delta$. 
\item \label{it:rew2} 
when $\beta$ and $\gamma$ satisfy $0 \leq \beta < 1$ and $\gamma>1$ then uniqueness holds on the
infinite $\Delta$-regular tree for all sufficiently large~$\Delta$.
\end{enumerate}

In order to prove Theorem~\ref{thm:main},  
we will construct a family of $k$-uniform hypergraphs
so that the $2$-spin model that $f$ induces on these hypergraphs
simulates an anti-ferromagnetic binary $2$-spin model.
Thus, the constructed hypergraphs will be viewed as binary gadgets.
It will be important that the induced binary $2$-spin model is in the non-uniqueness region
so that we can prove hardness using Theorem~\ref{thm:slysun}.
 Our constructions will use the   conditional distributions induced by pinning or equality to simplify the analysis of the gadgets.

 The conditional distribution will yield an idealised antiferromagnetic 2-spin system with parameters $\beta_0$ and
 $\gamma_0$, say. The delicate issue that arises is that the
 hypergraphs that we can construct to simulate these conditional distributions (see Lemmas~\ref{lem:pineq}
 and~\ref{lem:generalpin}) are imperfect. There is always a small error~$\epsilon$.
 So even if the ideal antiferromagnetic spin-system given by~$\beta_0$ and~$\gamma_0$
 is in the non-uniqueness region, we will
 have constructed some nearby binary spin-system given by (say) parameters~$\beta$
 and~$\gamma$ and we will have to prove that the spin-system given by~$\beta$ and~$\gamma$
 is also an anti-ferromagnetic spin system in the non-uniqueness region.
 
In general, the error bound that we will get from Lemma~\ref{lem:generalpin}
will tell us that for some small constant $\epsilon$,
 $|\beta-\beta_0|< \epsilon$  and
$|\gamma-\gamma_0|<\epsilon$. The most difficult case will be when $\gamma_0$ is close to~$1$
(including the case where $\gamma_0$ is actually~$1$).
In this case, we might have $\gamma$ slightly larger than~$1$ and we will thus need
to  exclude Item~\ref{it:rew2} above. 

In order to overcome these obstacles, we rely on making the
error $\epsilon$ very small, at the expense, of course, of potentially increasing the
degree bound~$\Delta$. By a continuity-type of argument, 
we will show that for $\beta_0$ strictly less than 1 and $\gamma_0\leq 1$, for all $\beta,\gamma$ which are sufficiently close to $\beta_0,\gamma_0$, there exists a $\Delta$ such that the 2-spin system with parameters $\beta,\gamma$ is in the non-uniqueness regime of the infinite $\Delta$-regular tree  (which can then be used to derive hardness).\footnote{Of course, other approaches may make it possible
to    directly construct ``strictly antiferromagnetic'' gadgets, without relying on Lemma~\ref{lem:strip}.} 
We will prove a slightly stronger statement by giving a bound on the required accuracy $\epsilon$ in terms of the degree $\Delta$, which  will allow us to switch the order of quantifiers. Also, our
result will be monotone   in the degree-bound $\Delta$ (as in Item~\ref{it:rew} above). 
\begin{lemma}\label{lem:strip} 
Suppose $0\leq \beta_0<1$. Then, for all sufficiently large $\Delta$, for $\epsilon=1/\Delta$, for all $\beta,\gamma$ which satisfy
\begin{equation}\label{hardnessregion}
\max\{\beta_0-\epsilon,0\}\leq  \beta< \beta_0+\epsilon\mbox{ and }0<\gamma<1+\epsilon,
\end{equation}
the 2-spin system with parameters $\beta,\gamma$ and $\lambda=1$ (no external field) is antiferromagnetic and in the non-uniqueness regime of the infinite $\Delta$-regular tree.
\end{lemma}
\begin{proof}
By choosing $\Delta$ sufficiently large, for all $\beta,\gamma$ which satisfy \eqref{hardnessregion}, it clearly holds that $\beta \gamma<1$ and thus the corresponding 2-spin system is antiferromagnetic. We next show that for all $\Delta$ sufficiently large the 2-spin system is also in the non-uniqueness regime of the infinite $\Delta$-regular tree.

We first  consider the ``soft-constrained" case $\beta_0>0$, where we will assume throughout that $\Delta>1/\beta_0$, so that for all $\beta,\gamma$ satisfying \eqref{hardnessregion} it holds that $\beta\gamma>0$ and $\beta<1$. 
We first recall basic facts about the uniqueness regime of soft-contrained antiferromagnetic 2-spin systems on the infinite $\Delta$-regular tree for $\Delta\geq 3$. The reader is referred to, e.g., \cite[Lemma 21]{LLY} for more details.

We have already seen  that non-uniqueness holds on the infinite $\Delta$-regular tree   when the system of equations \eqref{eq:wed} has multiple positive solutions. We also saw that
this corresponds to the case where $|h'(x^*)|>1$ where $x^*$ is the unique positive solution of $x^*= h(x^*)$
for a function $h(x)$ defined shortly before Equation~\eqref{eq:wed}.
 In \cite[Lemma 21]{LLY}, it is shown that there exist values $\lambda_1:=\lambda_1(\beta,\gamma,\Delta)$ and $\lambda_2:=\lambda_2(\beta,\gamma,\Delta)$ such that the condition $|h'(x^*)|< 1$ holds either when (i) $\sqrt{\beta \gamma}>(\Delta-2)/\Delta$, or (ii) $\sqrt{\beta \gamma}\leq (\Delta-2)/\Delta$ and $\lambda<\lambda_1$ or $\lambda>\lambda_2$. Adapting the proof \cite[Proof of Lemma 21, Item 7]{LLY} it is not hard to see that the condition $|h'(x^*)|> 1$ which we are interested in holds iff $\sqrt{\beta \gamma}< (\Delta-2)/\Delta$ and $\lambda\in (\lambda_1,\lambda_2)$, where $\lambda_1,\lambda_2$ are as in \cite[Lemma 21, Item 7]{LLY}. Thus, for $\beta_0>0$, our goal is to show that, for all sufficiently large $\Delta$,  for all $\beta,\gamma$ satisfying \eqref{hardnessregion} with $\epsilon=1/\Delta$, it holds that $\lambda_1<1<\lambda_2$.

To show the inequalities for $\lambda_1$ and $\lambda_2$, we next describe explicitly the values of  $\lambda_1$ and $\lambda_2$. For fixed $\beta,\gamma>0$ and $\Delta\geq 3$ with $\sqrt{\beta \gamma}<(\Delta-2)/\Delta$, the values of $\lambda_1,\lambda_2$ can be obtained as follows (see \cite[Lemma 22]{LLY}). To align with the setting in \cite{LLY}, we denote $d:=\Delta-1$. When $\sqrt{\beta \gamma}< (d-1)/(d+1)$, it is not hard to show that the equation 
\begin{equation}\label{eq:nonuniqueness}
\frac{d(1-\beta \gamma)x}{(\beta x+1)(x+\gamma)}=1, \mbox{ which is equivalent to } \beta x^2+\big((d+1)\beta \gamma-(d-1)\big)x+\gamma=0,
\end{equation} 
has two distinct positive solutions $x_1,x_2$. Without loss of generality, we may   assume that   $x_1< x_2$. For future use, we remark that  
\begin{align}
x_1&=\Big((d-1)-(d+1)\beta \gamma-\sqrt{\big((d-1)-(d+1)\beta \gamma\big)^2-4\beta \gamma}\Big)/2\beta\notag\\
&=2\gamma/\Big((d-1)-(d+1)\beta \gamma+\sqrt{\big((d-1)-(d+1)\beta \gamma\big)^2-4\beta \gamma}\Big),\label{eq:asymptotics1}
\end{align}
where the latter expression follows by taking the conjugate expresion. We thus obtain the crude bounds \begin{equation}\label{eq:boundsx1x2}
x_1<\frac{2\gamma}{T}, \quad \frac{T}{2\beta}<x_2,\mbox{ where }T:=(d-1)-(d+1)\beta \gamma.
\end{equation}
(The upper bound for $x_1$ is obtained by ignoring the square root and the lower bound for $x_2$ is obtained from the upper bound for $x_1$ and noticing, from \eqref{eq:nonuniqueness}, that $x_1x_2=\gamma/\beta$.)  

The values of $\lambda_1,\lambda_2$ in terms of $x_1,x_2$ are given by
\begin{equation}\label{eq:lambda1lambda2}
\lambda_1=\lambda_1(\beta,\gamma,d):=x_1\left(\frac{x_1+\gamma}{\beta x_1+ 1}\right)^{d},\quad \lambda_2=\lambda_2(\beta,\gamma,d):=x_2\left(\frac{x_2+\gamma}{\beta x_2+ 1}\right)^{d}.
\end{equation}

For future use, note that $(x_1+\gamma)/(\beta x_1+1)< x_1+\gamma$.
Also, since $\beta<1$ and $\gamma>0$,
 $(x_2+\gamma)/(\beta x_2+1)>x_2/(x_2+1)> (x_2-1)/x_2$. So, using also the bounds for $x_1$ and $x_2$ from \eqref{eq:boundsx1x2}, the expressions in \eqref{eq:lambda1lambda2} yield the bounds 
 \begin{equation}\label{boundslam}
\begin{aligned}
\lambda_1&< x_1(x_1+\gamma)^d< \frac{2\gamma^{d+1}}{T}(1+2/T)^d<\frac{6}{T}(1+2/T)^d, \\ 
\lambda_2&> x_2(1-1/x_2)^d> \frac{T}{2\beta}(1-2\beta/T)^d>\frac{T}{2}(1-2/T)^d,
\end{aligned}
\end{equation}
where in the rightmost inequalities we used the bounds $\gamma^{d+1}\leq \big(1+1/(d+1)\big)^{d+1}<3$ for $d\geq 2$ (since $\gamma\leq 1+\epsilon$ from \eqref{hardnessregion}) and $\beta<1$ (since $\beta_0<1$ and 
we can use our initial assumption that $d$ is sufficiently large and hence $\epsilon$ in \eqref{hardnessregion} small).

For $\beta_0>0$  and $\beta,\gamma$ satisfying \eqref{hardnessregion}, we have the bound $(d+1)\beta \gamma< (d+2)\beta_0+2$, so for all sufficiently large $d$ (depending only on $\beta_0$), it holds that $T>d(1-\beta_0)/2$. Hence, the bounds in \eqref{boundslam} yield that, for all sufficiently large $d$,
\[\lambda_1<\frac{12}{d(1-\beta_0)}\Big(1+\frac{4}{d(1-\beta_0)}\Big)^d, \quad \lambda_2>  \frac{d(1-\beta_0)}{4}\Big(1-\frac{4}{d(1-\beta_0)}\Big)^d. \]
Since $1>\beta_0>0$, we clearly obtain that for all sufficiently large $d$ (depending only on $\beta_0$), it holds that $\lambda_1<1<\lambda_2$, as wanted. This completes the proof of the lemma for the case $\beta_0>0$.

We next consider the case  $\beta_0=0$. For $\gamma>0$, the 2-spin system specified by $\beta=0,\gamma$ and $\lambda=1$ is in the uniqueness regime of the infinite $(d+1)$-regular tree iff $1\leq \lambda_c(\gamma,d):=\frac{\gamma^{d+1}d^d}{(d-1)^{d+1}}$ (see, for example, \cite[Proof of Item 5 in Lemma 21]{LLY}). Since $\lambda_c(\gamma,d)$ is an increasing function of $\gamma$, it suffices to show that $\lambda_c(1+\epsilon,d)<1$ for all sufficiently large $d\geq 2$. Note that $(1+\epsilon)^{d+1}<3$ for all $d\geq 2$ and, for $d\geq 10$, we have $\frac{d^d}{(d-1)^{d+1}}\leq 1/3$. It follows that for all $d\geq 10$, it holds that $\lambda_c(1+\epsilon,d)<1$, as needed.

To complete the proof for the case $\beta_0=0$, we need to argue that for all sufficiently large $d$, for all $(\beta,\gamma)\in (0,\epsilon)\times (0,1+\epsilon)$ with $\epsilon=1/(d+1)$,   the 2-spin system with parameters $\beta,\gamma,\lambda=1$, is in the non-uniqueness regime of the infinite $(d+1)$-regular tree. For $d\geq 10$ and this range of the parameters $\beta,\gamma$ we have that $\beta\gamma>0$ and $\sqrt{\beta \gamma}\leq (d-1)/(d+1)$. Thus it suffices to establish that for all sufficiently large $d$ it holds that $\lambda_1(\beta,\gamma,d)<1<\lambda_2(\beta,\gamma,d)$, where $\lambda_1,\lambda_2$ are as in \eqref{eq:lambda1lambda2}. In fact, we can use the bounds in \eqref{boundslam}, so we only need to provide a lower bound on $T$  since our derivation previously was to account  for the case $\beta_0>0$. For this range of the parameters $\beta,\gamma$, we have that $(d+1)\beta \gamma<2$. It follows that $T>d-3$.  Thus, the bounds in \eqref{boundslam} yield 
\[\lambda_1< \frac{6}{d-3}\Big(1+\frac{2}{d-3}\Big)^d, \quad \lambda_2> \frac{d-3}{2}\Big(1-\frac{2}{d-3}\Big)^d,\] 
so that $\lambda_1<1<\lambda_2$ for all sufficiently large $d$.
This concludes the proof of the lemma. \end{proof}

\subsection{Proof of Lemma~\ref{hardness}}
Using Lemmas~\ref{lem:generalpin} and~\ref{lem:strip}, we now give the proof of Lemma~\ref{hardness}. 
The idea is to use Lemma~\ref{lem:generalpin} to obtain a hypergraph that realises the conditional distribution with sufficient accuracy $\epsilon$. The resulting hypergraph can be used to simulate an antiferromagnetic 2-spin system which, by Lemma~\ref{lem:strip} and  Corollary~\ref{thm:isinghardness},
will be hard to approximate on $\Delta$-regular graphs (for large $\Delta$). The formal proof  is as follows. 

\begin{lemhardness}
\statelemnumhardness
\end{lemhardness}

\begin{proof} 
Let $H=(V,\mathcal{F})$.
We start by applying Lemma~\ref{lem:generalpin}
with $S = \{x,y\}$.
For every $\epsilon'>0$ and every $s_1,s_2\in \{0,1\}$,
the lemma shows that there is a hypergraph $H'=(V',\mathcal{F}')$ with $V \subseteq V'$
and $\mathcal{F} \subseteq \mathcal{F}'$ so that
\begin{equation}\label{eq:limitlimit}
\big|\mu_{f;H'}(\sigma_x=s_1,\sigma_y=s_2)-\mu^{\condV}_{f;H}(\sigma_x=s_1,\sigma_y=s_2)\big|\leq \epsilon'.
\end{equation}
For $s_1,s_2\in\{0,1\}$, let $\Zmu'_{s_1s_2}=\mu_{f;H'}(\sigma_x=s_1,\sigma_y=s_2)$. Thus, \eqref{eq:limitlimit} becomes 
\begin{equation}\label{eq:sec}
|\Zmu'_{s_1s_2}-\Zmu_{s_1s_2}|\leq \epsilon'.
\end{equation}

The conditions in~\eqref{eq:antiferro} guarantee that $\Zmu_{01}$ and $\Zmu_{10}$
are positive, so by choosing $\epsilon'$ sufficiently small, we can also guarantee
that $\Zmu'_{01}$ and $\Zmu'_{10}$ are positive.

Assume without loss of generality that $\Zmu_{00} \leq \Zmu_{11}$.
(Otherwise, we will swap the role of the spins~$0$ and~$1$.)
Let $\beta_0 = \Zmu_{00}^2/\Zmu_{01}\Zmu_{10}$ and let
$\gamma_0 = \Zmu_{11}^2/\Zmu_{01}\Zmu_{10}$.
By~\eqref{eq:antiferro}, $0 \leq \beta_0 <1 $ and $0< \gamma_0 \leq 1$.

Next, we ``symmetrise" the hypergraph $H'$ to obtain a hypergraph $H''$ (an analogous argument was used previously in the proof of Lemma~\ref{lem:gadgets}). To do this, take two disjoint copies of $H'$ which we denote by $H_1',H_2'$. For $i=1,2$, denote by $x_i,y_i$ the images of the vertices $x,y$ in $H_i'$.  Now identify vertices $x_1$ and $y_2$ into a single vertex $x$, and similarly identify vertices $x_2$ and $y_1$ into a single vertex $y$. Let $H''$ be the final hypergraph
and let $\mu''_{s_1 s_2}$
denote $\mu_{f;H''}(\sigma(x)=s_1,\sigma(y)=s_2)$.
Then for all $s_1,s_2\in \{0,1\}$, we have 
$$\mu''_{s_1 s_2} =   \frac{\mu'_{s_1 s_2} \mu'_{s_2 s_1}}
{\sum_{t_1,t_2\in\{0,1\}} \mu'_{t_1 t_2} \mu'_{t_2 t_1}}.$$  
Note that $\Zmu_{01}''=\Zmu_{10}''>0$.

Consider the 2-spin system with parameters  $\beta:=\Zmu''_{00}/\Zmu''_{01}$, $\gamma:=\Zmu''_{11}/\Zmu''_{01}$, and $\lambda=1$. 
Using the definitions of~$\beta_0$ and~$\gamma_0$
and
Equation~\eqref{eq:sec} we find that
for every $\Delta'\geq 3$, and for every sufficiently small $\epsilon'>0$, we have
\begin{equation}\label{eq:wer}
|\beta-\beta_0|< 1/\Delta'\mbox{ and } |\gamma-\gamma_0|<1/\Delta'.
\end{equation}
 
 By Lemma~\ref{lem:strip},
   there is a $\Delta'_0\geq 3$ such that
 if $\Delta'\geq \Delta'_0$   then the spin system with parameters $\beta,\gamma,\lambda=1$ is in the non-uniqueness regime of the infinite $\Delta'$-regular tree.  Thus, by Corollary~\ref{thm:isinghardness}, 
 there is a $c>1$ such that approximating $Z_{\beta,\gamma,1;G}$ 
 within a factor of~$c^n$ is $\NP$-hard on the class of $\Delta'$-regular 
 $n$-vertex graphs~$G$.
  
 Let $\Delta''$ be the maximum degree of the hypergraph $H''$. We will show the lemma  
 for all $\Delta \geq \Delta'\Delta''$.

Let $G=(V,E)$ be a $\Delta'$-regular graph for which we want to compute $Z_{\beta,\gamma,1;G}$. We construct the hypergraph $H'''$ by replacing each edge of $G$ with a copy of the hypergraph $H''$ as follows. For each edge $(u,v)\in E$, take a (distinct) copy $H_{uv}=(V_{uv},\mathcal{F}_{uv})$ of the hypergraph $H''$. Denote by $x_{uv},y_{uv}$ the images of the vertices $x,y$ in $H_{uv}$. Now for each $u\in V$ identify the vertices $x_{uv_1}, \hdots,x_{uv_{\Delta'}}$ into a single vertex $u$, where $v_1,\hdots,v_{\Delta'}$ denote the neighbors of $u$ in $G$. It is clear that the hypergraph $H'''$ has maximum degree $\Delta$.

Let $V''',\mathcal{F}'''$ denote the vertex and hyperedge sets of $H'''$, respectively. For $\sigma:V\rightarrow \{0,1\}$, let $\Sigma(\sigma)=\{\tau:V'''\rightarrow \{0,1\}\mid \tau_V=\sigma\}$. The total contribution to the partition function 
$Z_{f;H'''}$ from configurations in $\Sigma(\sigma)$ is   exactly 
$\prod_{(u,v)\in E}(\Zmu''_{\sigma(u)\sigma(v)} Z_{f;H''})$. Thus, we obtain
\begin{equation}\label{eq:conn1}
\begin{aligned}
Z_{f;H'''}& ={(\Zmu_{01}'' Z_{f;H''})}^{|E|}\sum_{\sigma:V\rightarrow\{0,1\}}\prod_{(u,v)\in E}\frac{\Zmu_{\sigma(u)\sigma(v)}''}{\Zmu_{01}''} =(\Zmu_{01}'' Z_{f;H''})^{|E|}\, Z_{\beta,\gamma,1;G}.
\end{aligned}
\end{equation}
Since $\Zmu_{01}'' Z_{f;H''}$ is an explicitly computable constant,  it follows that an approximation to the partition function $Z_{f;H'''}$ within   a factor of $c^n$ yields an approximation to $Z_{\beta,\gamma,1;G}$ within 
a factor of $c^n$. This completes the proof
since $\Delta'$ and $H''$ are fixed, which guarantees that the number of vertices of~$H'''$
is a constant multiple of the number of vertices of~$G$. \end{proof}
  
\section{Proof of Theorem~\ref{thm:main}}
\label{sec:proof}

In this section, we give the proof of Theorem~\ref{thm:main}.
Let $k\geq 2$ and  $f:\{0,1\}^k\rightarrow \{0,1\}$ be a \emph{symmetric Boolean} function with $f\notin \EASYk$. 
Our goal is to show that there exists $\Delta_0$ such that for all  $\Delta\geq \Delta_0$, there exists $c>1$ such that $\Hyper2Spinf$ is $\NP$-hard.

The case $k=2$ corresponds to approximating the partition function of \emph{unweighted} 2-spin systems in graphs. It is not hard to see that the only symmetric arity-2 Boolean functions with $f\notin \EASYk$ are 
the function which is~$1$ if at least one of $x_1$ and $x_2$ is~$0$
and the function which is~$1$ if at least one of $x_1$ and $x_2$ is~$1$.
  The partition function in both of these cases corresponds to counting the number of the independent sets, or equivalently to the 2-spin system with $\beta=0$ and $\gamma=1$ (and no external field). For this model, it is well known that non-uniqueness holds in the infinite $\Delta$-regular tree when $\Delta\geq 6$, which in conjuction with Theorem~\ref{thm:isinghardness} completes the proof of Theorem~\ref{thm:main} in the special case $k=2$  (alternatively, one may use Lemma~\ref{lem:strip} to argue that non-uniqueness holds for all sufficiently large $\Delta$).

Thus, for the rest of the proof, we will assume that $k\geq 3$. We may further assume that at least one of $w_0,\hdots,w_k$ is 0 (otherwise $f$ is the constant function $\onef^{(k)}$) and at least one is 1 (otherwise $f$ is the constant function $\zerof^{(k)}$).

By Lemma~\ref{lem:classify}, to prove Theorem~\ref{thm:main} we may split the analysis into the following cases

\begin{enumerate}
\item $f$ supports both pinning-to-0 and pinning-to-1.
\item $f$ supports 2-equality (but neither pinning-to-0 nor pinning-to-1).
\item $f$ supports pinning-to-0 or pinning-to-1 (but not both).
(By  swapping~$0$ and~$1$, it would be identical to assume that $f$ supports pinning-to-0 but not pinning-to-1.)
\end{enumerate}
In each case, the goal is 
to show that there exists $\Delta_0$ such that for all  $\Delta\geq \Delta_0$, there exists $c>1$ such that $\Hyper2Spinf$ is $\NP$-hard.

For example, when $f$ is the function corresponding to weak independent sets, 
$f$ supports pinning-to-0 but not pinning-to-1, so $f$ is in  Case~3.
The same is true when $f$ is the function corresponding to strong independent sets.
On the other hand, if $f$ is the ``not-all-equal'' function, then it supports $2$-equality, but neither pinning-to-0 nor
pinning-to-1, so it is in Case 2.

Before delving into the proofs for each of this cases, we  give a piece of terminology which will simplify the exposition. 
The reader may wish to recall Definition~\ref{conditioned}.
We will typically invoke Lemma~\ref{hardness} for a hypergraph $H$ and an admissible collection of sets $\mathcal{V}=(V_0,V_1,V_2,\hdots,V_r)$. Rather than formally defining $\mathcal{V}$ in each such application of Lemma~\ref{hardness}, it will be convenient (and more instructive) to say, e.g.,  pin vertices $x_1,x_2$ to 0 (instead of specifying $V_0$ as $V_0=\{x_1,x_2\}$),  pin vertices $x_3,x_4$ to 1 (instead of specifying $V_1$ as $V_1=\{x_3,x_4\}$), and force equality among $x_5,x_6,x_7$ (instead of specifying $V_2$ as $V_2=\{x_5,x_6,x_7\}$).

\subsection{Case I}\label{sec:caseI} 

In this section, we assume that the function $f$ supports both pinning-to-0 and pinning-to-1.  In this case, Lemma~\ref{hardness} will always be applied to the hypergraph with a single edge $e:=\{x_1,x_2,\hdots,x_k\}$. For explicitness and with a slight abuse of notation we will denote by $e$ this hypergraph. 

\vskip 0.2cm \textbf{Case Ia} Suppose first that exactly one of $w_0,\hdots,w_k$ is equal to 1, say $w_j=1$. We may assume that $j\neq 0$ and $j\neq k$, otherwise $f=\allzerof^{(k)}$ or $f=\allonef^{(k)}$, respectively. We will consider separately the cases $j=1$ and $k>j\geq 2$. 

Suppose first that $k>j\geq 2$. Pin $x_1,\hdots,x_{j-2}$ to 1 (if $j=2$, no vertex is pinned to one), pin $x_{j+2},\hdots,x_k$ to 0 (if $j+2>k$, no vertex is pinned to zero) and set $x:=x_{j-1}$, $y:=x_{j}$ (since $j< k$, note that vertex $x_{j+1}$ is ``free"). We have:
\begin{align*}
\mu_{f;e}^{\condV}(\sigma_x=\sigma_y=0)&\propto w_{j-2}+w_{j-1}=0,\\
\mu_{f;e}^{\condV}(\sigma_x=0,\sigma_y=1)&\propto w_{j-1}+w_j=1, \\
\mu_{f;e}^{\condV}(\sigma_x=\sigma_y=1)&\propto w_j+w_{j+1}=1.
\end{align*}
Also, by symmetry, 
$\mu_{f;e}^{\condV}(\sigma_x=0,\sigma_y=1) =
\mu_{f;e}^{\condV}(\sigma_x=1,\sigma_y=0)$.
(We will use similar symmetry arguments in the rest of this proof without pointing them out explicitly).
So , from Lemma~\ref{hardness}, for all sufficiently large $\Delta$, there exists $c>1$ such that $\Hyper2Spinf$ is $\NP$-hard.

Suppose next that $j=1$. Pin $x_{4},\hdots,x_k$ to 0 and set $x:=x_{1}$, $y:=x_{2}$ (since $k\geq 3$, note that vertex $x_{3}$ is ``free"). We have:
\begin{align*}
\mu_{f;e}^{\condV}(\sigma_x=\sigma_y=0)&\propto w_{0}+w_1=1,\\
\mu_{f;e}^{\condV}(\sigma_x=0,\sigma_y=1)&\propto w_{1}+w_2=1, \\
\mu_{f;e}^{\condV}(\sigma_x=\sigma_y=1)&\propto w_{2}+w_3=0,
\end{align*}
so, from Lemma~\ref{hardness}, for all sufficiently large $\Delta$, there exists $c>1$ such that $\Hyper2Spinf$ is $\NP$-hard.

\vskip 0.2cm \textbf{Case Ib.} Suppose next that at least two of $w_0,\hdots,w_k$ are equal to 1. Let $i,j$ be two indices with $i<j$ such that $w_{i}=w_{j}=1$ and $w_{i+1}=\hdots=w_{j-1}=0$. We will first prove that for any two such indices, it holds that $j=i+2$ (otherwise, we will show that $\Hyper2Spinf$ is $\NP$-hard). Since $f\neq \eqf^{(k)}$, we may assume that either $i>0$ or $j<k$. Without loss of generality, we assume that $j<k$ (otherwise we may swap the spins 0 and 1).

\textbf{1.} If $j>i+2$, we consider cases whether $w_{j+1}=0$ or 1. If $w_{j+1}=0$, pin $x_1,\hdots,x_{j-2}$ to 1, pin $x_{j+2},\hdots,x_k$ to 0 and set $x:=x_{j-1}$, $y:=x_{j}$ (since $j<k$, note that vertex $x_{j+1}$ is ``free"). We have:
\begin{align*}
\mu_{f;e}^{\condV}(\sigma_x=\sigma_y=0)&\propto w_{j-2}+w_{j-1}=0,\\
\mu_{f;e}^{\condV}(\sigma_x=0,\sigma_y=1)&\propto w_{j-1}+w_j=1, \\
\mu_{f;e}^{\condV}(\sigma_x=\sigma_y=1)&\propto w_j+w_{j+1}=1,
\end{align*}
so, from Lemma~\ref{hardness}, for all sufficiently large $\Delta$, there exists $c>1$ such that $\Hyper2Spinf$ is $\NP$-hard.  If $w_{j+1}=1$, pin $x_1,\hdots,x_{j-1}$ to 1, pin $x_{j+2},\hdots,x_k$ to 0 and set $x:=x_{j}$, $y:=x_{j+1}$. We have:
\begin{align*}
\mu_{f;e}^{\condV}(\sigma_x=\sigma_y=0)&\propto w_{j-1}=0,\\
\mu_{f;e}^{\condV}(\sigma_x=0,\sigma_y=1)&\propto w_j=1, \\
\mu_{f;e}^{\condV}(\sigma_x=\sigma_y=1)&\propto w_{j+1}=1,
\end{align*}
so, from Lemma~\ref{hardness}, for all sufficiently large $\Delta$, there exists $c>1$ such that $\Hyper2Spinf$ is $\NP$-hard.

\textbf{2.} Assume now that $j=i+1$. Suppose there exists $j'$  such that $w_{i}=w_{i+1}=\hdots=w_{j'}=1$ and $w_{j'+1}=0$. Pin $x_1,\hdots,x_{j'-1}$ to 1, pin $x_{j'+2},\hdots,x_k$ to 0 and set $x:=x_{j'}$, $y:=x_{j'+1}$. We have:
\begin{align*}
\mu_{f;e}^{\condV}(\sigma_x=\sigma_y=0)&\propto w_{j'-1}=1,\\ 
\mu_{f;e}^{\condV}(\sigma_x=0,\sigma_y=1)&\propto w_{j'}=1,\\ 
\mu_{f;e}^{\condV}(\sigma_x=\sigma_y=1)&\propto w_{j'+1}=0,
\end{align*}
so, from Lemma~\ref{hardness}, for all sufficiently large $\Delta$, there exists $c>1$ such that $\Hyper2Spinf$ is $\NP$-hard. 

If such a $j'$ does not exist, then it holds that $w_{i}=w_{i+1}=\hdots=w_k=1$,  so there exists $i'>0$ such that $w_{i'}=\hdots=w_{k}=1$ and $w_{i'-1}=0$ (otherwise $f=\onef^{(k)}$). Pin $x_{i'+2},\hdots,x_k$ to 0, pin $x_{1},\hdots,x_{i'-1}$ to 1 and set $x:=x_{i'}$, $y:=x_{i'+1}$. We have:
\begin{align*}
\mu_{f;e}^{\condV}(\sigma_x=\sigma_y=0)&\propto w_{i'-1}=0,\\
\mu_{f;e}^{\condV}(\sigma_x=0,\sigma_y=1)&\propto w_{i'}=1,\\
 \mu_{f;e}^{\condV}(\sigma_x=\sigma_y=1)&\propto w_{i'+1}=1,
\end{align*}
so again  from Lemma~\ref{hardness}, for all sufficiently large $\Delta$, there exists $c>1$ such that $\Hyper2Spinf$ is $\NP$-hard.

It follows that for every two indices with $i<j$ such that $w_{i}=w_{j}=1$ and $w_{i+1}=\hdots=w_{j-1}=0$, it holds that $j=i+2$. Let $i'$ be the minimum integer such that $w_{i'}=1$. We have that $w_0=w_1=\hdots=w_{i'-1}=0$. Let $j'$ be the maximum integer such that  $w_{i'}=w_{i'+2}=\hdots=w_{i'+2j'}=1$. By assumption, at least two $w_i$'s are equal to 1, so we have that $j'\geq 1$. We also have that $w_{i'+1}=w_{i'+3}=\hdots=w_{i'+2j'-1}=0$, $w_{i'+2j'+1}=\hdots=w_k=0$. 

We may assume that either $i'\notin\{0,1\}$ or $i'+2j'\notin \{k-1,k\}$ (otherwise either $f= \evenf^{(k)}$ or $f=\oddf^{(k)}$). Let us assume first that $i'+2j'\notin \{k-1,k\}$, i.e., $i'+2j'\leq k-2$, so that $w_{i'+2j'+2}=0$. Pin $x_1,\hdots,x_{i'}$ to 1, set 
$x:=x_{i'+1},y:=x_{i'+2}$ and pin $x_{i'+2j'+3}, \hdots, x_k$ to 0. We have:
\begin{align*}
\mu_{f;e}^{\condV}(\sigma_x=\sigma_y=0)&\propto \sum^{2j'}_{\ell=0}\binom{2j'}{\ell} w_{i'+\ell}=2^{2j'-1},\\
 \mu_{f;e}^{\condV}(\sigma_x=0,\sigma_y=1)&\propto \sum^{2j'}_{\ell=0}\binom{2j'}{\ell} w_{i'+1+\ell}=2^{2j'-1},\\
 \mu_{f;e}^{\condV}(\sigma_x=\sigma_y=1)&\propto \sum^{2j'}_{\ell=0}\binom{2j'}{\ell} w_{i'+2+\ell}=2^{2j'-1}-1,
\end{align*}
so from Lemma~\ref{hardness}, for all sufficiently large $\Delta$, there exists $c>1$ such that $\Hyper2Spinf$ is $\NP$-hard.

The case $i'\notin \{0,1\}$ can be covered by an analogous argument, the only difference being that now we pin $x_1,\hdots,x_{i'-1}$ to 1, set $x:=x_{i'},y:=x_{i'+1}$ and pin $x_{i'+2j'+1}, \hdots, x_k$ to 0. Other than that, the previous calculations may be easily modified to obtain 
\begin{align*}
\mu_{f;e}^{\condV}(\sigma_x=\sigma_y=0)&\propto 2^{2j'-1}-1,\\ 
 \mu_{f;e}^{\condV}(\sigma_x=0,\sigma_y=1)&\propto 2^{2j'-1},\\ 
 \mu_{f;e}^{\condV}(\sigma_x=\sigma_y=1)&\propto 2^{2j'-1},
\end{align*}
yielding, by Lemma~\ref{hardness}, that for all sufficiently large $\Delta$, there exists $c>1$ such that $\Hyper2Spinf$ is $\NP$-hard. This concludes the proof in the case where $f$ supports both pinning-to-0 and pinning-to-1.

\subsection{Case~II}
\label{sec:CaseII}

In Case~II we assume that $k\geq 3$ and that $f$ is a symmetric arity-$k$ Boolean function 
that is not in $\EASYk$. The function~$f$
supports $2$-equality
but does not support pinning-to-0 or pinning-to-1. 
 
By Lemma~\ref{lem:selfdual}
we conclude that $f$ is self-dual, meaning that $w_\ell = w_{k-\ell}$ for all $\ell \in \{0,\ldots,k\}$.
By Lemma~\ref{lem:equalst}, we conclude that $f$ supports $t$-equality for all $t\geq 2$.
Our goal is to show that for all sufficiently large $\Delta$, there exists $c>1$ such that 
the approximation problem $\Hyper2Spinf$ is NP-hard. We 
prove this by considering two cases, depending on whether  $w_0=0$ or $w_0=1$. 
We will use the following lemma in both cases.

\begin{lemma}\label{lem:qwer}
Let $f$ be an arity-$k$ symmetric Boolean formula 
that is self-dual.   Let $H$ be a hypergraph with vertex set $V$ and denote by $\Sigma:=\{\sigma\mid \sigma:V\rightarrow\{0,1\}\}$ the set of all $\{0,1\}$ assignments on $V$.  Let $Q:\Sigma\rightarrow \Sigma$ be the map which maps an assignment $\sigma$ to its complement $\bar{\sigma}$, i.e., $\bar{\sigma}$ is defined by $\bar{\sigma}_v=1-\sigma_v$ for all $v\in V$. Then for every $\Sigma'\subseteq \Sigma$, it holds that $\mu_{f;H}(\Sigma')=\mu_{f;H}(Q(\Sigma'))$.
\end{lemma}
\begin{proof}
For every $\sigma\in \Sigma$, self-duality gives that $w_{f;H}(\sigma)=w_{f;H}(\bar{\sigma})$. Summing this equality over all assignments $\sigma$ in the subset $\Sigma'$ yields the result.
\end{proof}

We will split the analysis into two cases -- the case where $w_0=0$ (Section~\ref{sec:w00}) and
the case where $w_0=1$ (Section~\ref{sec:Case2w0one}).
Before these two sections, we make a digression into Constraint Satistfaction Problems (CSP).
The digression will introduce and prove a lemma that we will need for the $w_0=0$ case.
In addition, it will provide some missing detail  which we used in the Introduction to
explain the context of existing work.

\subsubsection{A digression regarding Constraint Satisfaction Problems}
\label{sec:qwe}

Recall the  CSP definitions from Section~\ref{sec:CSP}.
Let $\CSP{\Gamma}$ be the problem of determining whether the partition function $Z_{\Gamma,I}$
is non-zero, given an instance~$I$ of a CSP in which all constraints are from the set~$\Gamma$.

We will use the following CSP terminology.
Let $f$ be an arity-$k$ Boolean function.
For some positive integer~$m$,
let $g$ be a function~$g: \{0,1\}^m \rightarrow \{0,1\}$.   
Suppose that $\overline{x}_1,\ldots, \overline{x}_m$
are $m$ Boolean $k$-tuples
so for $i \in \{1,\ldots,m\}$ we can write $\overline{x}_i$ as 
a tuple $\overline{x}_i = (x_{i,1},\ldots,x_{i,k})$ in $\{0,1\}^k$. 
We will let
$\overline{y}_g(\overline{x}_1,\ldots,\overline{x}_m) = (y_1,\ldots,y_k)$
be the Boolean $k$-tuple constructed from~$g$
and from $\overline{x}_1,\ldots, \overline{x}_m$ as follows.
For each $j \in \{1,\ldots,k\}$, 
$y_j$ is obtained by applying $g$ to $x_{1,j}, \ldots,x_{m,j}$
so $y_{j} = g(x_{1,j},\ldots,x_{m,j})$. 
The function~$g$ is said  to be a \emph{polymorphism} of~$f$
if, for any choice of $m$ tuples 
$\overline{x}_1,\ldots,\overline{x}_m$
satisfying $f(\overline{x}_1) = \cdots = f(\overline{x}_m)=1$,
we  also have
$f( \overline{y}_g(\overline{x}_1,\ldots,\overline{x}_m)) = 1$.

We will use the following algebraic formulation of Chen~\cite[Theorem 3.21]{Hubie}
of Schaefer's famous dichotomy theorem~\cite{Schaefer}.
\begin{theorem}(Schaefer)\label{lem:Schaefer}
Let $\Gamma$ be a finite Boolean constraint language. The problem $\CSP{\Gamma}$ is polynomial-tractable
if one of the following six functions is a polymorphism of every function $f\in \Gamma$.
\begin{enumerate}
\item $g$ is the unary function $g_0$ with $g_0(0)=g_0(1)=0$. \label{sch:one}
\item $g$ is the unary function $g_1$ with $g_1(0)=g_1(1)=1$. \label{sch:two}
\item $g$ is the arity-2 Boolean function $\Band$. \label{sch:three}
\item $g$ is the arity-2 Boolean function $\Bor$. \label{sch:four}
\item $g$ is the ternary  majority function~$\maj$ defined by 
$\maj(a,b,c) = (a\wedge b) \vee (a \wedge c) \vee (b \wedge c)$.\label{sch:five}
\item $g$ is the  ternary minority function~$\minority$ defined
by $\minority(a,b,c) = a \oplus b \oplus c$. \label{sch:six}
\end{enumerate}
Otherwise, $\CSP{\Gamma}$ is $\NP$-complete.
\end{theorem}

We start by showing 
that the decision problem $\CSP{\{f\}}$ is $\NP$-hard
when $f$ is a non-trivial arity-$k$ symmetric Boolean formula 
that is self-dual and satisfies $w_0=0$.
In Section~\ref{sec:w00} we will use this fact and the fact $f$ supports equality 
to show that $\Hyper2Spinf$ is NP-hard.

\begin{lemma}\label{lem:applyschaefer}
Suppose $k>2$.
Let $f\neq \zerof^{(k)},\oddf^{(k)}$ be an arity-$k$ symmetric Boolean formula 
that is self-dual and satisfies $w_0=0$. Then
$\CSP{\{f\}}$ is $\NP$-hard.
\end{lemma}
\begin{proof}
We 
will show that~$f$ does not satisfy any of the tractable cases in Schaefer's dichotomy theorem
(Lemma~\ref{lem:Schaefer}). 
Since $f$ is not $\zerof^{(k)}$, it is not
identically zero. Thus, there is a $j$ in the range $1\leq j \leq \lceil k/2\rceil $ such
that  $w_j=1$. We will use this value~$j$ in the cases below.
 
\noindent {\bf Cases~\ref{sch:one} and~\ref{sch:two}}: 
We  first show that $g_0$ is not a polymorphism of~$f$.
 To see this, let $\overline{x}_1$ be any $k$-tuple with $j$ ones
so that $f(\overline{x}_1)=1$. 
Applying the function $g_0$  position-wise, we get
$\overline{y}_{g_0}(\overline{x}_1,\overline{x}_2)
= (y_1,\ldots,y_k)=   (g_0(x_{1,1}),\ldots,g_0(x_{1,k})) = (0,\ldots,0)$.
But since $w_0=0$,
$f(0,\dots,0)=0$, contrary to the fact that
$f(y_1,\ldots,y_k)$ would have to be~$1$ if $g_0$ were a polymorphism of~$f$. Similarly, 
 applying $g_1$ componentwise to $\overline{x}_1$ we get
$ \overline{y}_{g_1}(\overline{x_1},\overline{x_2})
=(1,\ldots,1)$. Since $w_k=0$ (by self-duality),
$f(1,\ldots,1)=0$ so $g_1$ is not a polymorphism of~$f$.

\noindent {\bf Cases~~\ref{sch:three} and \ref{sch:four}}:
Instead of doing both cases, we first use self-duality to argue that 
if $\Bor$ is a polymorphism of~$f$ then so is~$\Band$.
For this, suppose that $\Bor$ is a polymorphism of~$f$.
Let $\overline{x}_1$ and $\overline{x}_2$ be two tuples 
with $f(\overline{x}_1)=f(\overline{x}_2)=1$.
For $i\in \{1,2\}$, let $\neg\overline{x}_i$ be the position-wise Boolean complement of 
$\overline{x}_i$.
By self-duality, 
$f(\neg \overline{x}_1)=f(\neg \overline{x}_2)=1$ so
$f(y_{\Bor}(\neg\overline{x}_1,\neg\overline{x}_2))=1$. 
But $y_{\Bor}(\neg\overline{x}_1,\neg\overline{x}_2)$
is the position-wise Boolean complement of
$y_{\Band}(\overline{x}_1,\overline{x}_2)$,
so by self-duality, we also have
$f(y_{\Band}(\overline{x}_1,\overline{x}_2))=1$, establishing that $\Band$ is also a polymorphism of~$f$.
Thus, we can complete both cases by just showing that $\Band$ is actually not a polymorphism of~$f$.

Recall the value~$j \leq \lceil k/2 \rceil$ from above.
We first deal with the simplest case where 
there is a $j\leq k/2$ with $w_j=1$.
 Consider two tuples 
$\overline{x}_1$ and $\overline{x}_2$, each with $j$ ones, 
chosen so
that there is no position $\ell$ with $x_{1,\ell} = x_{2,\ell} = 1$.
This is possible since $j \leq k/2$.
Then $f(\overline{x}_1) = f(\overline{x}_2)=1$ but
$\overline{y}_{\Band}(\overline{x}_1,\overline{x}_2) = (0,\ldots,0)$. Now
  $w_0=0$, so 
$f(\overline{y}_{\Band}(\overline{x}_1,\overline{x}_2))=0$ and we have shown that
$\Band$ is not a polymorphism of~$f$.

We now deal with the remaining case. We have
$j= \lceil k/2 \rceil$ and $w_{j}=1$ and every   $\ell\neq j$ has $w_\ell=0$.
In this case we can consider any distinct tuples $\overline{x}_1$ and $\overline{x}_2$
with exactly  $j$ ones.
Then $\overline{y}_{\Band}(\overline{x}_1,\overline{x}_2)$
  has fewer than $j$ ones
  so 
       $f(\overline{y}_{\Band}(\overline{x}_1,\overline{x}_2))=0$, so 
$\Band$ is not a polymorphism of~$f$.

\noindent {\bf Case~\ref{sch:five}}:   
Suppose that $\maj$ is a polymorphism of~$f$. 
We will derive a contradiction. In order to simplify the tedious 
special cases arising from floors and ceilings we write $k$ as $k= 6 r + 3 a + b$
where $r$ is a non-negative integer, $a \in \{0,1\}$ and $b \in \{0,1,2\}$.

First, the fact that $\maj$ is a polymorphism of~$f$ implies that for every
$\ell \leq 2r+a$, we have $w_\ell=0$.  
To see this, let $\overline{x}_1$, $\overline{x}_2$ and $\overline{x}_3$  be three $k$-tuples, each with $\ell$ ones, 
such that there is no  position 
$p$ with more than  a single one amongst $x_{1,p}$, $x_{2,p}$ and $x_{3,p}$.
This is possible since $\ell \leq k/3$. Then $\overline{y}_{\maj}(\overline{x}_1,\overline{x}_2,\overline{x}_3) = 
  (0,\ldots,0)$.
So if $\maj$ is a polymorphism of~$f$
we must have that one of $f(\overline{x}_1)$,
$f(\overline{x}_2)$ and $f(\overline{x}_3)$ is~$0$
(which means, by symmetry of~$f$, that all of them are~$0$), so $w_\ell=0$.
 
Now consider 
any integer~$\ell$ in the range $2r+a < \ell \leq  3r+a$. 
Specifically, for an integer~$s$ in the range $ 1\leq s \leq   r$, let
$\ell = s+ 2r+a$.
  Consider three $k$-tuples 
$\overline{x}_1$, $\overline{x}_2$ and $\overline{x}_3$, each 
  with $\ell$ ones, such that $s$  positions have ones in all three tuples, 
  $s$  positions have ones in tuples 
$\overline{x}_1$ and $\overline{x}_2$,  
  and the remaining  positions have a one in exactly one tuple.  
  Then $\overline{y}_{\maj}(\overline{x}_1,\overline{x}_2,\overline{x}_3)$ 
   has $2s$ ones.   Since 
$2s\leq  2r$, we have $w_{2s}=0$.
So if $\maj$ is a polymorphism we must have 
$f(\overline{x}_1) = f(\overline{x}_2) = f(\overline{z}_3)=0$ so $w_\ell=0$.

So the only possible values of~$j$
where we could have $w_j=1$  satisfy
$j> 3r+a$.
By self-duality, they also satisfy
$k-j > 3r+a$
so $ 3r+a + 1 \leq j < 3r+2a+b$. 
Since $a+1 < 2a+b$,
the pair $(a,b)$ is  in the set $\{(0,2),(1,1),(1,2)\}$.
 So the three possibilities are
\begin{itemize}
\item $(a,b)=(0,2)$ so $k = 6r+2$ is even and
the only $j$ with $w_j=1$ satisfies 
$3r+1 \leq j < 3r+2$ so
$j=3r+1=k/2$.
\item $(a,b)=(1,1)$ so $k= 6r+3+1$ is even and
the only $j$ with $w_j=1$ satisfies $3r+2 \leq j < 3r+3$
so $j=3r+2=k/2$.
\item $(a,b)=(1,2)$ so $k=6r+3+2$ is odd and
the only $j$ with $w_j=1$ satisfy
$3r+2 \leq j < 3r+4$ so by self-duality,  there exactly two values $w_j$
that are non-zero, and these
are $j=3r+2 = \lfloor k/2 \rfloor$ and $j=3r+3 = \lceil k/2 \rceil$.
\end{itemize}

There must be a $j$ with $w_j\neq 0$ since $f$ is not the constant zero function.
We show that in all three cases $\maj$ is not a polymorphism.
We take the first two cases together, so suppose that $k$ is even,
and that there is exactly one positive $w_j$ which is $w_{k/2}$.
Since $k$ is even and greater than~$2$, it is at least~$4$.
Choose $\overline{x}_1$ with $1$'s in positions $1,\ldots,k/2$
and $\overline{x}_2$ with $1$'s in positions $2,\ldots,k/2+1$.
Choose $\overline{x}_3$ with $1$'s in positions $k/2+1,\ldots,k-1$ and $1$.
Then $\overline{y}_{\maj}(\overline{x}_1,\overline{x}_2,\overline{x}_3) $
has   $(k/2)+1$ ones (in positions $1,\ldots,k/2+1$) so $f(\overline{y}_{\maj}(\overline{x}_1,\overline{x}_2,\overline{x}_3))=0$ and $\maj$ is not a polymorphism.

The final case is similar.
Suppose that $k= 2 t + 1$ and
that there are exactly two positive $w_j$'s which are $w_t$ and $w_{t+1}$.
Choose $\overline{x}_1$ with $1$'s in positions $1,\ldots,t+1$.
Choose $\overline{x}_2$ with $1$'s in positions $1,\ldots,t$ and $t+2$.
Choose $\overline{x}_3$ with $1$'s in positions $t+1,\ldots,2t+1$.
Then $\overline{y}_{\maj}(\overline{x}_1,\overline{x}_2,\overline{x}_3) $
has   $t+2$ ones (in positions $1,\ldots,t+2$) so $f(\overline{y}_{\maj}(\overline{x}_1,\overline{x}_2,\overline{x}_3))=0$ and $\maj$ is not a polymorphism.

\noindent {\bf Case~\ref{sch:six}}:  
Suppose that $\minority$ is a polymorphism of~$f$. We will derive 
some consequences about the $w_\ell$ values which will give us a contradiction.

First,  
there is no index $\ell$ such that $w_\ell=1$ and $w_{\ell-2}=0$.
Clearly this is not the case for $\ell=k$ since $w_k=0$.
Suppose for contradiction that it is true for some $2 \leq \ell<k$.
Construct tuples $\overline{x}_1$, $\overline{x}_2$ and $\overline{x}_3$ each
with $\ell$ ones,  
such that the first 
$\ell-2$ positions have ones in all three tuples,
and in each of the next three positions there is exactly one zero.
Any remaining positions are all zero.
This is possible since $\ell+1 \leq k$.
 Then $\overline{y}_{\minority}(\overline{x}_1,\overline{x}_2,\overline{x}_3)$ has $\ell-2$ ones
 so if $\minority$ is a polymorphism of~$f$ then 
 $f(\overline{y}_{\minority}(\overline{x}_1,\overline{x}_2,\overline{x}_3))=1$ so $w_{\ell-2}=1$,
 contradicting the assumption.
 
Next, 
there is no odd index $3\leq j\leq k/2$
such that $w_j=0$ and $w_{j-2}=1$. 
Suppose for contradiction that this is true for some $j=2r+1$.
Construct tuples $\overline{x}_1$, $\overline{x}_2$ and $\overline{x}_3$ as follows.
\begin{itemize}
\item
The first $r-1$ positions have ones in tuples $\overline{x}_1$ and $\overline{x}_3$.
\item
The next $r-1$ positions have ones in tuples $\overline{x}_2$ and $\overline{x}_3$.
\item The next $r$ positions have ones in tuple $\overline{x}_1$ only.
\item The next $r$ positions have ones in tuple $\overline{x}_2$ only.
\item The next position has a one in tuple $\overline{x}_3$ only.
\item Any remaining positions are all zero.
\end{itemize}
Each tuple has $2r-1=j-2$ ones. The construction is possible since
$2(r-1)+2r +1 = 4r-1 \leq 4r+2 = 2j \leq k$.
Then $\overline{y}_{\minority}(\overline{x}_1,\overline{x}_2,\overline{x}_3)$ has $2r+1=j $ ones
 so and if $\minority$ is a polymorphism of~$f$ then 
 $f(\overline{y}_{\minority}(\overline{x}_1,\overline{x}_2,\overline{x}_3))=1$ so $w_{j}=1$,
 contradicting the assumption.
 
The first fact rules out the possibility that there is an even index $w_\ell$ with $w_\ell=1$.
This is ruled out because $w_0=0$ and we can derive a contradiction by considering the smallest
even~$\ell$ such that $w_\ell=1$.
This also tells us that $k$ is even.
This follows because there has to be some $j$ with $w_j=1$ 
and from the above, $j$ has to be odd.
But if $k$ is odd then  $k-j$ is even, yet self-duality would imply that $w_{k-j}=1$.

So since $f$ is not the trivial all-zero function~$\zerof^{(k)}$,
there is some (odd) $j \leq k/2$ with $w_j=1$.
Take $j$ as large as possible. 
The first fact tells us that for all odd $\ell<j$, $w_\ell=1$.
The second fact tells us that for all odd $\ell$ between $j$ and $k/2$, $w_\ell=1$.
Thus, all odd~$\ell$ have $w_\ell=1$.
 This implies that  $f=\oddf^{(k)}$, contrary to the statement of the lemma.
\end{proof}

Before returning to our main proof, we present one more lemma
that uses the language of polymorphisms, and we use this lemma to prove Observation~\ref{obs:CSP}, which
supports 
our interpretation of existing literature in the introduction to this paper.

\begin{lemma}
\label{lem:notIM2}
Suppose that $k\geq 3$. Let $f$ be a symmetric  $k$-ary Boolean 
function that is not in $\EASYk$.
Then either the arity-2 Boolean function $\Band$ is not a polymorphism of~$f$
or the arity-2 Boolean function $\Bor$ is not a polymorphism of~$f$ (or both).
\end{lemma}
\begin{proof}
We break the analysis into two cases.
\begin{description}
\item {\bf Case 1. There is an index $j$ in the range $1 \leq j \leq k-1$ such that $w_j = 1$ and $w_{j+1}=0$:}\quad
Let $\overline{x}_1$ have  ones in the first $j$ positions (and only in those positions) and
let $\overline{x}_2$ have ones in positions $2,\ldots,j+1$ (only).
Then $f(\overline{x}_1) = f(\overline{x}_2)=1$.
But $\overline{y}_{\Bor}(\overline{x}_1,\overline{x}_2))$ has $j+1$ ones so
$f(\overline{y}_{\Bor}(\overline{x}_1,\overline{x}_2))=0$ and $\Bor$ is not a polymorphism of~$f$.

\item {\bf Case 2. There is an index $j$ in the range $0 \leq j \leq k-2$ such that $w_j=0$ and $w_{j+1}=1$:}\quad
Let $\overline{x}_1$ have  ones in the first $j+1$ positions (and only in those positions) and
let $\overline{x}_2$ have ones in positions $2,\ldots,j+2$ (only).
Then $f(\overline{x}_1) = f(\overline{x}_2)=1$.
But $\overline{y}_{\Band}(\overline{x}_1,\overline{x}_2))$ has $j$ ones so
$f(\overline{y}_{\Band}(\overline{x}_1,\overline{x}_2))=0$ and $\Band$ is not a polymorphism of~$f$.
\end{description}
If neither  Case~1 nor Case~2 applies then $f$ must be one of the four 
functions $\zerof^{(k)}$, $\onef^{(k)}$, $\allzerof^{(k)}$ and $\allonef^{(k)}$  defined by
 $\zerof^{(k)}(x_1,\ldots,x_k)=0$,
$\onef^{(k)}(x_1,\ldots,x_k)=1$, 
$$\allzerof^{(k)}(x_1,\ldots,x_k)=\mathbf{1}\{x_1=\hdots=x_k=0\},$$ and
$$\allonef^{(k)}(x_1,\ldots,x_k)=\mathbf{1}\{x_1=\hdots=x_k=1\}.$$
All four of these functions are in~$\EASYk$.
\end{proof}

\begin{observation}\label{obs:CSP}
Suppose that $\Delta \geq 6$ and $k\geq 3$ and that $f$ is a symmetric $k$-ary Boolean 
function that is not in $\EASYk$. Then there is no FPRAS for $\nCSP{\{f,\delta_0,\delta_1\}}$ unless $\NP=\RP$.\end{observation}
\begin{proof}
Theorem~24 of \cite{DGJR}   gives the result unless
(i) $f$ is affine (given by a linear equation over GF$_2$), or
(ii) $f$ is in a set of functions called IM-conj.
All symmetric affine functions are in $\EASYk$.
IM-conj is the same as the class IM$_2$ studied in~\cite{BoolBlocks, CKZ}. 
As described in these works, the Galois correspondence between Post's lattice and  its
dual 
shows that $f$ is in IM$_2$ if and only if it has both $\Band$ and $\Bor$ as polymorhisms.
Thus, by Lemma~\ref{lem:notIM2} this case does not arise.
 \end{proof}

\subsubsection{The case $w_0=0$}
\label{sec:w00}
 
We now return to our proof.
By assumption,  $k\geq 3$ and  $f$ is a symmetric arity-$k$ Boolean function 
that is not in $\EASYk$. Since we are in Case~II, the function~$f$
supports $2$-equality
but it does not support pinning-to-0 or pinning-to-1.  
By Lemma~\ref{lem:selfdual}, we know that $f$ is self-dual.
We are interested in the case $w_0=0$ and we know from Lemma~\ref{lem:applyschaefer}
that a related decision CSP problem is $\NP$-hard.
We wish to use the hardness of the CSP decision problem to
show hardness of our bounded-degree counting problem.
We will use the fact that $f$ supports equality 
to introduce degree bounds and also to move to the more restricted
hypergraph 2-spin model where repeated variables are not allowed.
   The following technical lemma is inspired by   techniques from \cite{SATbounded}.

\begin{lemma}\label{lem:exactequality}
Suppose $k>2$.  
Let $f\neq \zerof^{(k)},\oddf^{(k)}$ be an arity-$k$ symmetric Boolean formula 
that is self-dual and satisfies $w_0=0$. Then
 there is a hypergraph $H$  
 with $Z_{f;H}>0$  
 which contains vertices $x$ and $y$   such that
for any configuration $\sigma:V(H) \rightarrow \{0,1\}$ with
 $w_{f;H}(\sigma)>0$, we have
  $\sigma(x)=\sigma(y)$.
\end{lemma}
\begin{proof} 
We first prove that there exists a hypergraph $H_0=(V_0,\mathcal{F}_0)$ such that $Z_{f;H_0}=0$. To see this, consider the complete $k$-uniform hypergraph on $2k-1$ vertices, i.e., $V_0=\{1,\hdots,2k-1\}$ and $\mathcal{F}_0$ is the set of all $k$-element subsets of $V_0$. Consider an arbitrary assignment $\sigma:V_0\rightarrow \{0,1\}$. Under $\sigma$, there exist $k$ vertices which have the same spin, w.l.o.g. assume that these vertices are $1,\hdots,k$. Since $w_0=w_k=0$ (note that $w_k=0$ by self-duality), it follows that the hyperedge $\{1,\hdots,k\}\in \mathcal{F}_0$ is not satisfied under $\sigma$. Since $\sigma$ was arbitrary, this proves that $Z_{f;H_0}=0$.
By removing hyperedges of $H_0$ successively, we obtain a hypergraph $H'=(V',\mathcal{F}')$ such that $Z_{f;H'}=0$ and for every $e\in \mathcal{F}'$, it holds that $Z_{f;H'\setminus e}>0$. 
Since $f$ is not the all-zero function $\zerof^{(k)}$, we can conclude that $H'$ has at least one hyperedge.

Choose $e\in \mathcal{F}'$  and let $S\subseteq e$ be the set of vertices of~$e$ that have non-zero degree in $H'\setminus e$. 
By the minimality of $H'$, we have $S\neq \emptyset$. 
Let $i = |S|$. Denote the vertices in~$S$ by $v_1,\ldots,v_i$ and the vertices in $e\setminus S$ by $v_{i+1},\ldots,v_k$.
    Consider $i$ new vertices $u_1,\hdots,u_{i} \notin V'$ 
and for each $t\in\{0,\ldots,i\}$,  let $e_t=\{u_1,\hdots,u_t,v_{t+1},\hdots,v_{i}\}\cup (e\setminus S)$. Further, consider the hypergraphs $H_t=(V_t,\mathcal{F}_t)$ where $V_t=V'\cup \{u_1,\hdots,u_t\}$ and $\mathcal{F}_t=(\mathcal{F}'\setminus \{e\})\cup \{e_t\}$. 
Note that $e_0=e$ and $e_i$ has no vertices that are in other hyperedges of~$H'$.
By the minimality of $H'$,  we can conclude that $Z_{f;H_i}>0$. Let $j$ be the smallest integer such that $Z_{f;H_j}>0$.  Then $1\leq j\leq k$.  Also, for every $\sigma:V_j\rightarrow\{0,1\}$ with $w_{f;H_j}(\sigma)>0$, it must hold that $\sigma(u_j)\neq \sigma(v_j)$ (otherwise, 
we would have $w_{f;H_{j-1}}(\sigma|_{V_{j-1}})>0$, contradicting
the fact that $Z_{f;H_{j-1}}=0$).

Let $H=(V_j\cup \{u_j'\},\mathcal{F}_j\cup \{e_j'\})$,  where $e_j'=\{u_1,\hdots,u_{j-1},u_j',v_{j+1},\hdots,v_{i}\}\cup (e\setminus S)$. As above, we conclude that 
$Z_{f;H}>0$. Also,
for every 
$\sigma:V_j\cup \{u_j'\}\rightarrow\{0,1\}$ with $w_{f;H}(\sigma)>0$ it must hold that $\sigma(u_j)\neq \sigma(v_j)$ and $\sigma(u_j')\neq \sigma(v_j)$, so $\sigma(u_j)=\sigma(u_j')$.
So the vertices $x$ and $y$ in the statement of the lemma can be taken to be $u_j$ and $u'_j$.
\end{proof}

We will now 
 combine Lemmas~\ref{lem:applyschaefer} and~\ref{lem:exactequality} to conclude the following.
\begin{lemma}\label{lem:w0goal}
Suppose $k>2$.  
Let $f\neq \zerof^{(k)},\oddf^{(k)}$ be an arity-$k$ symmetric Boolean formula 
that is self-dual and satisfies $w_0=0$.
Then there is a $\Delta_0$ such that for every $\Delta \geq \Delta_0$, there exists $c>1$ such that 
$\Hyper2Spinf$ is $\NP$-hard. 
\end{lemma}
\begin{proof}
From Lemma~\ref{lem:exactequality} we know there is a hypergraph $H''$ with 
$Z_{f;H''}>0$ which has
vertices 
$x$ and $y$ such that for any configuration $\sigma: V(H'') \rightarrow \{0,1\}$ with $w_{f;H''}(\sigma)>0$,
we have $\sigma(x)= \sigma(y)$.
Let $Z_0$ be the contribution to $Z_{f;H''}$ from configurations with $\sigma(x)=\sigma(y)=0$
and let $Z_1$ be the contribution to $Z_{f;H''}$ from configurations with $\sigma(x)=\sigma(y)=1$.
By  Lemma~\ref{lem:qwer}, we have $Z_0=Z_1$ so, since $Z_0+Z_1 = Z_{f;H''},$ which is positive,
both~$Z_0$ and~$Z_1$ are positive.
Let $\Delta''$ be the maximum degree of $H''$ and let $\Delta_0 =  2\Delta'' + 1$.

Let $\Gamma=\{f\}$.
We know from Lemma~\ref{lem:applyschaefer}
that it is $\NP$-hard to determine whether $Z_{\Gamma,I}=0$
given a CSP instance~$I$.

Consider a CSP instance~$I$ with variable set~$V$
and constraint set $\mathcal{S}$.
We will show how to (efficiently) construct a $k$-uniform hypergraph~$H'$
with degree at most~$\Delta_0$ so that
$Z_{\Gamma,I} = 0$ if and only if $Z_{f;H'}=0$.
This will imply that determining whether $Z_{f;H'}=0$ is $\NP$-hard given a $k$-uniform
hypergraph~$H'$  with degree at most~$\Delta_0$. Hence, 
for every $\Delta \geq \Delta_0$,
 $\Hyper2Spinf$ is also $\NP$-hard. 

The construction is straightforward, apart from the notation.
For each $v\in V$, let $n(v)$ be the number of times that
variable~$v$ is used (taking all of the constraints in $\mathcal{S}$ together).
Let $V' = \{(v,j) \mid v \in V, 1 \leq j \leq n(v) \}$.
Let $\mathcal{S'}$ be a set of constraints that is identical to
$\mathcal{S}$ except that,
for each $v\in V$ and $1\leq j \leq n(v)$,
the $j$'th use of vertex $v$ 
is replaced with $(v,j)$.
Note that the constraints in $\mathcal{S}'$ use each variable in~$V'$ exactly once.
Thus, they can be viewed as hyperedges of a $k$-uniform hypergraph.

To build $H'$ we will take the vertices in $V'$ and
the hyperedges in $\mathcal{S}'$
but we will add some additional
vertices and hyperedges.
In particular, for each $v\in V$ and $j \in \{1,\ldots,n(v)-1\}$
we will take a new copy $H''_{v,j}$ of $H''$.
We will identify the vertex $x$ of $H''_{v,j}$ with
$(v,j)$ and the vertex $y$ of $H''_{v,j}$ with $(v,j+1)$.
This completes the construction of the hypergraph~$H'$. 
Note that the degree of each vertex of $H'$ is at most $2 \Delta'' + 1 = \Delta_0$.
Also, $H'$ can be efficiently constructed given~$I$.
 
Now, by the properties of
$H''$, every configuration  
$\sigma':V(H') \rightarrow \{0,1\}$ with $w_{f;H'}(\sigma)>0$ 
has $\sigma((v,j)) = \sigma((v,j'))$
for every $v\in V$ and every $1\leq j,j' \leq n(v)$.
Also, $\sigma'$ induces a configuration 
$\sigma:V \rightarrow \{0,1\}$ by   taking $\sigma(v) = \sigma'(v,1) = \cdots = \sigma'(v,n(v))$.
Now note that
$$Z_{f;H'} =  
 \sum_{\sigma:V\rightarrow\{0,1\}}
 \left(\prod_{(v_1,\hdots,v_k,f)\in\mathcal{S}} 
f(\sigma(v_1),\hdots,\sigma(v_k))\right)
\left(\prod_{v \in V} Z_{\sigma(v)}^{n(v)-1}
\right).$$ 
This is identical to the partition function $Z_{f;I}$ apart from the factors of~$Z_0$ and~$Z_1$, which 
are both positive. We conclude that
$Z_{f;H'}$ is positive 
if and only if $Z_{f;I}$ is positive. \end{proof}

\subsubsection{The case $w_0=1$}
\label{sec:Case2w0one}

We now consider the case where $w_0=1$.  We will use the following corollary of Lemma~\ref{hardness} (tailored to the case where $f$ is self-dual).
\begin{corollary}\label{hardnessself}
Let $f$ be an arity-$k$ symmetric Boolean formula 
that is self-dual. Let $H$ be a hypergraph, let $\mathcal{V}$ be admissible for $H$ (with $V_0=V_1=\emptyset$) and 
let~$x$ and~$y$ be   vertices  of~$H$. Suppose that 
\begin{equation}\label{eq:selfself}
0<\mu_{f;H}^{\condV}(\sigma_x=\sigma_y=0)<\mu_{f;H}^{\condV}(\sigma_x=0,\sigma_y=1).
\end{equation}
Then there is a $\Delta_0$ such that for every $\Delta \geq \Delta_0$, there exists $c>1$ such that 
$\Hyper2Spinf$ is $\NP$-hard. 
\end{corollary}
\begin{proof}
For $i,j\in\{0,1\}$, let $\Zmu_{ij}=\Zmu_{f;H}^{\condV}(\sigma_x=i,\sigma_y=j)$. By Lemma~\ref{lem:qwer}, we have   $\Zmu_{11}=\Zmu_{00}$ and $\Zmu_{01}=\Zmu_{10}$. It follows that condition \eqref{eq:antiferro} is equivalent to condition \eqref{eq:selfself}, so the corollary follows by applying Lemma~\ref{hardness}.
\end{proof}

We use Corollary~\ref{hardnessself} to show the following lemma.
\begin{lemma}\label{lem:mainselfdual}
Let $f:\{0,1\}^k\rightarrow\{0,1\}$ be a self-dual symmetric function.  Let $t_1,t_2$ be integers such that $t_1\geq 1,t_2\geq 0,2t_1+t_2\leq k$. Suppose that
\begin{equation}\label{eq:condition}
0<w_0+w_{t_2}+w_{2t_1}+w_{2t_1+t_2}<2(w_{t_1}+w_{t_1+t_2}).
\end{equation}
Then there exists $\Delta_0$ such that for every $\Delta \geq \Delta_0$, there exists $c>1$ such that $\Hyper2Spinf$ is $\NP$-hard. 
\end{lemma}
\begin{proof} 
We will apply Corollary~\ref{hardnessself}.  Let $e=\{x_1,\hdots,x_k\}$. Let $\mathcal{V}$ force equality on the sets  of vertices $\{x_{1},\hdots,x_{t_1}\}$, $\{x_{t_1+1},\hdots,x_{2t_1}\}$,  $\{x_{2t_1+1},\hdots,x_{2t_1+t_2}\}$ and, whenever $2t_1+t_2<k$, on $\{x_{2t_1+t_2+1},\hdots,x_{k}\}$ (see Figure~\ref{fig:somehypergraph}).
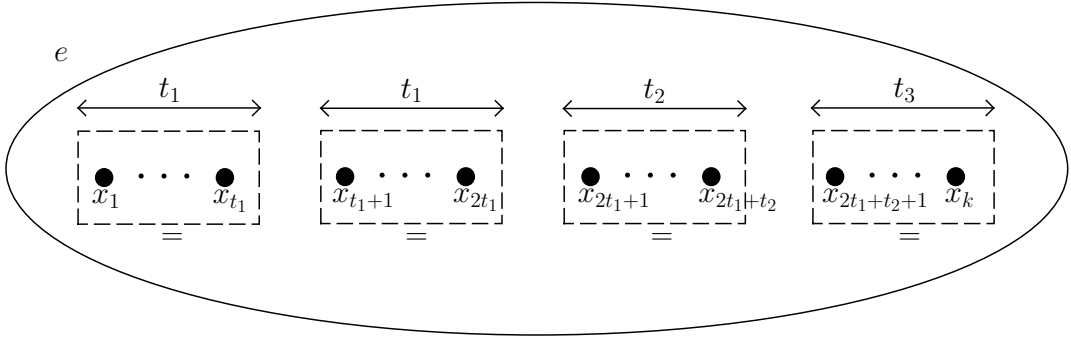
\begin{figure}[t]
\begin{center}
\scalebox{0.56}[0.56]{\input{hypergraph5.tex}}
\end{center}
\caption{The hypergraph used in the proof of Lemma~\ref{lem:mainselfdual}. The hypergraph has just one hyperedge $e=\{x_1,x_2,\hdots,x_k\}$. We partition the hyperedge into sets of sizes $t_1,t_1,t_2,t_3$ and look at the conditional distribution where the spins of the vertices in each set are equal.}\label{fig:somehypergraph}
\end{figure}
Let $t_3:=k-2t_1-t_2$ and set $x:=x_1$ and $y:=x_{t_1+1}$. Note that the condition $2t_1+t_2<k$ is equivalent to $t_3>0$. By the assumptions, we also have that $t_3\geq 0$.
Now define 
 \begin{align*}
 Z^{\mathrm{ex}}_{00}&= w_0+w_{t_2}+w_{2t_1}+w_{2t_1+t_2},  \\
 Z^{\mathrm{ex}}_{01}& =2(w_{t_1}+w_{t_1+t_2}).
\end{align*} 

First, suppose $t_2>0$ and $t_3>0$.
In this case, we have (using self-duality in the first equality in each line)
\begin{align*}
\mu_{f;e}^{\condV}(\sigma_x=\sigma_y=0)&\propto  w_0+w_{t_2}+w_{t_3}+w_{t_2+t_3}
 = w_0 + w_{t_2} + w_{2 t_1+t_2} + w_{2 t_1}
 = Z^{\mathrm{ex}}_{00},\\
\mu_{f;e}^{\condV}(\sigma_x=0,\sigma_y=1)&\propto  w_{t_1}+w_{t_1+t_2}+w_{t_1+t_3}+w_{t_1+t_2+t_3} =
w_{t_1}+w_{t_1+t_2}+w_{t_1+t_2}+w_{t_1}= Z^{\mathrm{ex}}_{01}.
\end{align*}
 It is now immediate  that inequality~\eqref{eq:condition} is equivalent to condition~\eqref{eq:selfself} in Corollary~\ref{hardnessself}, from which the result follows. 
 
The proofs for the remaining cases for the values of $t_2,t_3$ are completely analogous. If $t_2>0$ and $t_3=0$ then $w_0 = w_{2 t_1+t_2}$ and $w_{t_2} = w_{2 t_1}$ so we   have
 \begin{align*}
\mu_{f;e}^{\condV}(\sigma_x=\sigma_y=0)&\propto  w_0+w_{t_2} =\tfrac12 Z^{\mathrm{ex}}_{00},\\
\mu_{f;e}^{\condV}(\sigma_x=0,\sigma_y=1)&\propto  w_{t_1}+w_{t_1+t_2} = \tfrac12 Z^{\mathrm{ex}}_{01}.
\end{align*}
If $t_2=t_3=0$ then also  $w_{t_2}=w_0$ and $w_{t_1+t_2}=w_{t_1}$, so
\begin{align*}
\mu^{\condV}_{f;e}(\sigma_x=\sigma_y=0)&\propto  w_0  =\tfrac14 Z^{\mathrm{ex}}_{00},\\
\mu^{\condV}_{f;e}(\sigma_x=0,\sigma_y=1)&\propto  w_{t_1}  = \tfrac14 Z^{\mathrm{ex}}_{01}.
\end{align*}

Finally, if $t_2=0$ and $t_3>0$ then  
\begin{align*}
\mu^{\condV}_{f;e}(\sigma_x=\sigma_y=0)&\propto   w_0+w_{t_3} = w_0 + w_{2 t_1}=\tfrac12 Z^{\mathrm{ex}}_{00}
 ,\\
\mu^{\condV}_{f;e}(\sigma_x=0,\sigma_y=1)&\propto   w_{t_1}+w_{t_1+t_3} = 2 w_{t_1} = \tfrac12 Z^{\mathrm{ex}}_{01}.
\end{align*}  
Thus, in each of the above cases, the result follows by Corollary~\ref{hardnessself}.
\end{proof}

We will also use the following inequality for binomial coefficients (which is slightly stronger than the well-known log-concavity property $\binom{n}{i-1}\binom{n}{i+1}\leq \binom{n}{i}^2$).   
\begin{lemma}\label{cl:ineq}
For all $n\geq i\geq  2$, it holds that 
\[\left[1+\binom{n}{i}\right]\left[1+\binom{n}{i-2}\right]\leq \binom{n}{i-1}^2.\]
Equality holds iff $n=2,i=2$.
\end{lemma}
\begin{proof} 
For $i=2$ the inequality becomes $n^2-n+2\leq n^2$ which holds for all $n\geq 2$ (equality only if $n=2$). The same argument applies for $i=n$. Thus, we may assume $n-1\geq i\geq 3$. For $i=3$ the inequality becomes 
\[\left(1+\frac{n(n-1)(n-2)}{6}\right)(n+1)\leq \frac{n^2(n-1)^2}{4}\Leftrightarrow n^4-2n^3+5n^2-16n-12\geq 0.\] 
The latter can easily be verified that it  holds strictly for all $n\geq 4$. The same argument applies for $i=n-1$. Thus, we may assume $n-2\geq i\geq 4$.

For all such values of $n,i$, the inequality follows by summing 
\begin{gather}
1+\binom{n}{i}+\binom{n}{i-2}<\frac{1}{i-1}\binom{n}{i}\binom{n}{i-2},\label{eq:simpleone}\\
\frac{i}{i-1}\binom{n}{i}\binom{n}{i-2}< \binom{n}{i-1}^2\label{eq:simpletwo}.
\end{gather}
Inequality \eqref{eq:simpletwo} is equivalent to $n-i+1<n-i+2$ which is trivially true. To see \eqref{eq:simpleone}, we first rewrite it in the equivalent form
\[\left[\binom{n}{i}-(i-1)\right]\left[\binom{n}{i-2}-(i-1)\right]> i(i-1),\]
which follows from the inequalities $\binom{n}{i}\geq 2(i+1)$ and $\binom{n}{i-2}\geq 2(i-1)$ (both are special cases of $\binom{n}{j}\geq 2(j+1)$ which holds for $j\geq 2$ and $n\geq j+2$).
\end{proof}

We can now do the proof for  this case.
Recall   that $k\geq 3$ and that $f$ is a symmetric self-dual arity-$k$ Boolean function 
that is not in $\EASYk$. The function~$f$
supports $2$-equality
but does not support pinning-to-0 or pinning-to-1. 
We are assuming that $w_0=1$. 
Our goal is to show that 
for all sufficiently large $\Delta$, there exists $c>1$ such that 
the approximation problem $\Hyper2Spinf$ is NP-hard.

Let $0<i\leq k/2$ be the smallest 
positive index with $w_i=1$. Clearly, we may assume that such an index $i$ exists (otherwise, by the self-duality of $f$, we have $f=\eqf^{(k)}$).  

\begin{claim}\label{cl:one}
If there is a positive integer $r$ with $r i \leq k$ and $w_{ri}=0$ then,
for all sufficiently large $\Delta$, there exists $c>1$ such that $\Hyper2Spinf$ is $\NP$-hard.
\end{claim}
\begin{proof}
Let $r$ be the smallest positive integer with $ri\leq k$ and $w_{ri}=0$. 
Since $w_i=1$, we have $r\geq 2$. 
 We next check that the conditions of Lemma~\ref{lem:mainselfdual} are satisfied with $t_1=i$ and $t_2=(r-2)i$. First, we clearly have $t_1\geq 1$, $t_2\geq 0$, $2t_1+t_2=ri\leq k$. Moreover, the left-most inequality in \eqref{eq:condition} is true since $w_0=1$. Finally, the right-most inequality in \eqref{eq:condition} also holds, since 
\begin{align*}
&w_{t_1}+w_{t_1+t_2}=w_i+w_{(r-1)i}=2,\\
&w_{0}+w_{t_2}+w_{2t_1}+w_{2t_1+t_2}=w_0+w_{(r-2)i}+w_{2i}+w_{ri}\leq 3,
\end{align*}
where we used that $w_{(r-1)i}=1$ and $w_{ri}=0$ (by the choice of $r$).
Applying Corollary~\ref{hardnessself} yields the claim. 
\end{proof}

\begin{claim}\label{cl:two}
If there is a positive integer $r<k$ 
that is not divisible by~$i$ and has $w_r=1$   then,
for all sufficiently large $\Delta$, there exists $c>1$ such that $\Hyper2Spinf$ is $\NP$-hard. 
\end{claim}
\begin{proof}
Let $r$ be the smallest positive integer  that is 
less than~$k$, is not divisible by~$i$  and has $w_r=1$. 
By the choice of~$i$, we have $r>i$. 
By self-duality, $r$ is not in $\{k-i+1,\ldots,k-1\}$ so $r\leq  k-i$.
We next check that the conditions of Lemma~\ref{lem:mainselfdual} are satisfied with $t_1=i$ and $t_2=r-i$. 
We clearly have $t_1\geq 1$, $t_2> 0$, $2t_1+t_2=r+i \leq k$. 
Moreover, the left-most inequality in \eqref{eq:condition} is true since $w_0=1$. Finally, the right-most inequality in \eqref{eq:condition} also holds, since 
\begin{align*}
&w_{t_1}+w_{t_1+t_2}=w_i+w_{r}=2,\\
&w_{0}+w_{t_2}+w_{2t_1}+w_{2t_1+t_2}=w_0+w_{r-i}+w_{2i}+w_{r+i}\leq 3,
\end{align*}
where we used that $w_{r}=1$ and $w_{r-i}=0$ (by the choice of $r$).
Applying Corollary~\ref{hardnessself} yields the claim.
\end{proof}

The remaining cases that we have to deal with 
are now quite constrained, satisfying the following properties.
\begin{itemize}
\item $w_0=w_k=1$. 
(We know that this is true because this is the case that we are dealing with in 
the current section, Section~\ref{sec:Case2w0one}).
\item The positive integers $\ell \in \{1,\ldots,k-1\}$ 
with $w_{\ell}=1$ are precisely the multiples of~$i$. This follows from Claims~\ref{cl:one} and~\ref{cl:two}.
\item $k$ is a multiple of~$i$. Suppose instead that $k = m i + u$ for some non-negative integer~$m$ and some integer $u\in \{1,\ldots,i-1\}$.
Then $w_{m i}=1$ since $m i $ is  either~$0$ (if $m=0$) or it is a positive multiple of~$i$ which is less than~$k$.
Now $k-m i = u$ so by self-duality $w_u=1$.
But this contradicts the choice of~$i$.
\item $i>2$. If $i=1$ then $f$ is the all-one function $f=\onef^{(k)}$.
 If $i=2$ then $k$ is even since it is a multiple of~$i$.
 Then $f$ is the easy function $f=\evenf^{(k)}$.
 
 \item $2i \leq k$.  We know that $k$ is a multiple of~$i$, but if $k$ is equal to~$i$, then 
 $f$ is the equality function $f=\eqf^{(k)}$.
\end{itemize}
  
To finish the proof, we consider the hypergraph with a single edge $e=\{x_1,\hdots,x_k\}$.
Let $x$ be $x_{2i-1}$ and let $y$ be $x_{2i}$. Let $\mathcal{V}$ force equality among the vertices 
in $S=\{x_{2i+1},\ldots,x_k\}$.  
Suppose first that $k>2i$ (so that $|S|\geq 1$).
We will use $\ell$ to denote the number of spin-one vertices in $x_1,\ldots,x_{2i-1}$. Then, 
since the assignment to vertices in $S$ can be either the $\zeros$ or $\ones$ assignment, we get
\begin{equation}\label{eq:sdfcvb}
\begin{aligned}
\mu_{f;e}^{\condV}(\sigma_x=\sigma_y=0)&\propto  
\sum_{\ell=0}^{2i-2} \binom{2i-2}{\ell} (w_\ell+ w_{\ell+k-2i}),\\
\mu_{f;e}^{\condV}(\sigma_x=0,\sigma_y=1)&\propto  
\sum_{\ell=0}^{2i-2} \binom{2i-2}{\ell} (w_{\ell+1}+ w_{\ell+k-2i+1}). 
\end{aligned}
\end{equation}

But in the range $0\leq \ell \leq 2i-2$, 
$w_\ell$ is only positive if $\ell \in \{0,i\}$.
Similarly, by self-duality $w_{\ell+k-2i} = w_{2i-\ell}$, which is only
positive if $\ell\in\{0,i\}$.
Similarly, 
$w_{\ell+1}$ is only positive if $\ell = i-1$
and $w_{\ell+k-2i+1} = w_{2i-\ell-1}$ which is only positive if $\ell=i-1$.
So, \eqref{eq:sdfcvb} becomes 
\begin{align*}
\mu_{f;e}^{\condV}(\sigma_x=\sigma_y=0)&\propto  
2 +   2\binom{2i-2}{i},\\
\mu_{f;e}^{\condV}(\sigma_x=0,\sigma_y=1)&\propto  
2  \binom{2i-2}{i-1}. 
\end{align*} 
If $k= 2i$ then we get the same equations (apart from a factor of~$2$, which makes no difference).
  
To finish the argument,
we need only show that   $1+\binom{2i-2}{i}<\binom{2i-2}{i-1}$, so that Corollary~\ref{hardnessself} yields that  for all sufficiently large $\Delta$, there exists $c>1$ such that 
$\Hyper2Spinf$ is $\NP$-hard. 
To see that $1+\binom{2i-2}{i}<\binom{2i-2}{i-1}$, 
let $n=2i-2$. Note that $2<i<n$. Then   Lemma~\ref{cl:ineq} gives 
\[\left[1+\binom{2i-2}{i}\right]\left[1+\binom{2i-2}{i-2}\right]<\binom{2i-2}{i-1}^2.\]
The desired inequality follows after observing that $\binom{2i-2}{i-2}=\binom{2i-2}{i}$ and simplifying.

\subsection{Case III}\label{sec:caseIII}
 
Throughout this section, we will assume that $f$ supports pinning-to-0 (the case that $f$ supports pinning-to-1 is analogous by swapping the spins 0 and 1). 
Our goal is to show that there exists $\Delta_0$ such that for all  $\Delta\geq \Delta_0$, there exists $c>1$ such that $\Hyper2Spinf$ is $\NP$-hard.

We have the following analogue of Lemma~\ref{lem:gadgets} (tailored to the case where $f$ supports pinning-to-0).
\begin{lemma}\label{lem:zeropinning}
Assume that $f$ supports pinning-to-0. Let $H=(V,\mathcal{F})$ and  let $V_0\subseteq V$ be admissible for the hypergraph $H$, i.e., $\mu_{f;H}(\sigma_{V_0}=\zeros)>0$. With $\mathcal{V}=(V_0)$, recall that $\mu^{\condV}_{f;H}(\cdot)=\mu_{f;H}(\cdot\mid \sigma_{V_0}=\zeros)$. 

\begin{enumerate}
\item \label{it:bnm} If there exists a vertex $v$ in $H$ such that $\mu^{\condV}_{f;H}(\sigma_v=1)>\mu^{\condV}_{f;H}(\sigma_v=0)$, then $f$ supports pinning-to-1.
\item \label{it:bnm2} If there exists a subset  $S$ of $V$ such that $\mu^{\condV}_{f;H}(\sigma_S=\ones)=\mu^{\condV}_{f;H}(\sigma_S=\zeros)=1/2$, then $f$ supports equality.
\end{enumerate}
\end{lemma}
\begin{proof}
For Item~\ref{it:bnm}, choose $\epsilon$ in the range $0<\epsilon<|\mu^{\condV}_{f;H}(\sigma_v=1)-\mu^{\condV}_{f;H}(\sigma_v=0)|/2$. By Lemma~\ref{lem:generalpin}, there exists a hypergraph $H'$ with vertex set $V'\supseteq V$, such that for $s\in \{0,1\}$, it holds that 
\[|\mu_{f;H'}(\sigma_v=s)-\mu^{\condV}_{f;H}(\sigma_v=s)\leq \epsilon.\]
It follows that $\mu_{f;H'}(\sigma_v=1)>\mu_{f;H'}(\sigma_v=0)$, so by Lemma~\ref{lem:gadgets}, $f$ supports pinning-to-1.

For Item~\ref{it:bnm2}, for $\epsilon>0$, apply Lemma~\ref{lem:generalpin} to conclude that there exists a hypergraph $H'$  which $\epsilon$-realises $|S|$-equality. By Lemma~\ref{lem:equalst}, we obtain that $f$ supports equality.
\end{proof}

We remark here that, as in the analysis for Case I in Section~\ref{sec:caseI}, rather than formally defining $V_0$, in each of the subcases which we consider, we will typically say, e.g., pin vertices $x_1,x_2$ to 0 instead of specifying $V_0$ as $V_0=\{x_1,x_2\}$. Also, unless otherwise stated, we will have $\mathcal{V}=(V_0)$.

Our first application of Lemma~\ref{lem:zeropinning} is to show that if $w_0=0$, then $f$ also supports pinning-to-1. So  assume that $w_0=0$. Let $i$ be the smallest index such that $w_i\neq 0$ (if such an index does not exist, then $f=\zerof^{(k)}$), so that $1\leq i\leq k-1$ (if $i=k$ then $f=\allonef^{(k)}$). Let $e=\{x_1,x_2,\hdots,x_k\}$, pin $x_{i+1},\hdots,x_{k}$ to 0 and set $x:=x_1$. 
\begin{equation*}
\mu_{f;e}^{\condV}(\sigma_x=0)\propto \sum^{i-1}_{\ell=0}\binom{i-1}{\ell}w_\ell=0,\quad
\mu_{f;e}^{\condV}(\sigma_x=1)\propto \sum^{i-1}_{\ell=0}\binom{i-1}{\ell}w_{\ell+1}=w_i=1.
\end{equation*}
Thus, if $w_0\neq 1$, by Lemma~\ref{lem:zeropinning}, we obtain that $f$ also supports pinning-to-1, in which case we fall back in Case I.

Thus for the rest of the proof in this section we will assume that $w_0=1$ (otherwise, as we showed above, we fall back in Case I, where $f$ supports both pinning-to-0 and pinning-to-1). Let $i>0$ be the minimum index $i$ such that $w_i=1$ (we may assume that such an index exists, otherwise $f=\allzerof^{(k)}$). We consider the cases $i=1$ and $i\geq 2$ separately.

\subsubsection{The case $i\geq 2$}
The setting of this section is $w_0=w_i=1$ and $w_1=\hdots=w_{i-1}=0$ (and $f$ supports pinning-to-0). We may assume that $k>i$ (otherwise $f=\eqf^{(k)}$). 
\begin{lemma}\label{lem:igeq2}
If $k>i\geq 2$, then the function $f$ supports equality.
\end{lemma}
\begin{proof}
We show that $f$ supports $i$-equality, so by Lemma~\ref{lem:equalst}, we obtain that $f$ supports equality. 

Let $e=\{x_1,\hdots,x_{k}\}$ and pin vertices $x_{i+1},\hdots,x_k$ to 0. Denote $S:=\{x_1,\hdots,x_i\}$. We have
\[\mu_{f;H}^{\condV}(\sigma_S=\zeros)\propto w_0=1,\quad \mu_{f;H}^{\condV}(\sigma_S=\ones)\propto w_i=1, \quad \mu_{f;H}^{\condV}(\neg \sigma_S^{\eq})\propto 0.\]
It follows that  $\mu_{f;H}^{\condV}(\sigma_x=0)=\mu_{f;H}^{\condV}(\sigma_x=1)=1/2$, so the result follows from Lemma~\ref{lem:zeropinning}.
\end{proof}
 
\vskip 0.2cm \noindent \textbf{Case IIIa} There exists $j>i$ such that $w_j=1$. 

Let $j$ be the minimum such index. We have the following lemma.
\begin{lemma}\label{lem:j2i}
If $j\neq 2i$, the function $f$ supports also pinning-to-1.
\end{lemma} 
\begin{proof}
Suppose first that $j< 2i$. Let $e=\{x_1,\hdots,x_k\}$, pin $x_{j+1},\hdots,x_k$ to zero and set $x:=x_1$. We have:
\begin{align*}
\mu_{f;e}^{\condV}(\sigma_x=0)&\propto \sum^{j-1}_{\ell=0}\binom{j-1}{\ell}w_{\ell}=w_0+\binom{j-1}{i}w_i=1+\binom{j-1}{i},\\
\mu^{\condV}_{f;e}(\sigma_x=1)&\propto \sum^{j-1}_{\ell=0}\binom{j-1}{\ell}w_{\ell+1}=\binom{j-1}{i-1}w_i+w_{j}=1+\binom{j-1}{i-1}.
\end{align*}
From $2i>j$, we have $i-1\geq (j-1)/2$, so that $\mu_{f;e}^{\condV}(\sigma_x=1)>\mu_{f;e}^{\condV}(\sigma_x=0)$, so Lemma~\ref{lem:zeropinning}  yields that $f$ also supports pinning-to-1.

We next consider the more difficult case $j>2i$. Let $H$ be the hypergraph with vertex set 
\[\{x_1,\hdots,x_k\}\cup \{y_{i+1},\hdots,y_k\}\cup \{z_{i+1},\hdots,z_k\}\] 
and hyperedges $\{e_X,e_Y,e_Z\}$ with 
\begin{gather*}
e_X=\{x_1,\hdots,x_j,\hdots,x_k\},\ e_Y=\{x_1,..,x_i, y_{i+1},\hdots,y_j,\hdots,y_k\},\\
 e_Z=\{x_{i+1},\hdots, x_{2i}, z_{i+1},\hdots,z_{j},\hdots,z_k\}.
\end{gather*}
We will  pin to zero the vertices in $V_0:=\{x_{j+1},\hdots,x_{k}\}\cup \{y_{j+1},\hdots,y_k\}\cup \{z_{j+1},\hdots,z_k\}$ (thus, in the conditional distribution, each of the hyperedges $e_X,e_Y,e_Z$ has $j$ vertices). Also, we will force equality among the sets $V_2:=\{x_1,x_2,...,x_{2i}\}$,  $V_3:=\{x_{2i+1},\hdots,x_{j}\}$, $V_4:=\{y_{i+1},\hdots, y_j\}$ and $V_5:=\{z_{i+1},\hdots, z_j\}$. Let $x:=x_1$.

Consider $\mu^{\condV}_{f;H}(\cdot)=\mu_{f;H}(\cdot\mid \sigma_{V_0}=\ones,\sigma_{V_2}^{\eq},\sigma_{V_3}^{\eq},\sigma_{V_4}^{\eq},\sigma_{V_5}^{\eq})$. We will show that 
\begin{equation}\label{eq:edc}
\mu^{\condV}_{f;H}(\sigma_x=1)>\mu^{\condV}_{f;H}(\sigma_x=0),
\end{equation}
so the result will follow  by Lemma~\ref{lem:zeropinning}. To see \eqref{eq:edc}, assume first that $\sigma_{x}=0$. In the conditional distribution $\mu^{\condV}_{f;H}(\cdot)$, it then holds that $\sigma_{V_2}=\zeros$. From $j>2i$, there is only one way to satisfy $e_Y, e_Z$ by setting $\sigma_{V_4}=\sigma_{V_5}=\zeros$. Further, there are at most two ways to satisfy $e_X$ (by setting $\sigma_{V_3}$ to be $\zeros$ and possibly $\ones$, if $w_{j-2i}=1$). Thus,  the total weight of configurations with $\sigma_{x}=0$ is at most two. On the other hand, if $\sigma_x=1$ and hence $\sigma_{V_2}=\ones$,  from $w_j=1$, there is at least one way to satisfy to $e_X$ (by setting $\sigma_{V_3}=\ones$). Further, from $w_i=w_j=1$, there are two ways to satisfy each of $e_Y,e_Z$ (any combination of $\{\sigma_{V_4},\sigma_{V_5}\}\in\{\zeros,\ones\}$ works). Thus, the total weight of configurations with $\sigma_{x}=1$ is at least 4. We thus obtain that $\mu^{\condV}_{f;H}(\sigma_x=1)>\mu^{\condV}_{f;H}(\sigma_x=0)$, as wanted.

This concludes the proof.
\end{proof}

By Lemma~\ref{lem:j2i}, we may thus assume that $j=2i$ (otherwise we fall back into Case I, since $f$ supports both pinning-to-0 and pinning-to-1). That is, we have $w_0=w_i=w_{2i}$ and for $j\leq 2i$ with $j\neq 0,i,2i$, it holds that $w_j=0$.

If $i>2$, by pinning $k-2i$ variables to 0, we (approximately) get a non-trivial self-dual constraint with arity $2i$ (since $w_0=w_i=w_{2i}=1$) and hence the proof may be completed analogously to Case II. In particular, let  $e=\{x_1,\hdots,x_k\}$ and pin $x_{2i+1},\hdots,x_{k}$ to 0. Denote $x:=x_{2i-1}$, $y:=x_{2i}$. Note that if $\sigma$ is such that $\sigma_x=\sigma_y=0$, $\sigma$ has non-zero weight iff exactly $0$ or $i$ vertices from $x_{1},\hdots,x_{2i-2}$ have spin 1 under $\sigma$. Similarly, if $\sigma$ is such that $\sigma_x=\sigma_y=1$, $\sigma$ has non-zero weight iff exactly $i-2$ or $2i-2$ vertices from $x_{1},\hdots,x_{2i-2}$ have spin 1 under $\sigma$. On the other hand, if $\sigma$ is such that  $\sigma_x=0$ and $\sigma_y=1$ (or vice versa), $\sigma$ has non-zero weight iff exactly $i-1$ vertices from $x_{1},\hdots,x_{2i-2}$ have spin 1. It follows that for $\{s_1,s_2\}\in\{0,1\}$, we have $\mu^{\condV}_{f;e}(\sigma_x=s_1, \sigma_y=s_2)\propto Z_{s_1s_2}$ where 
\[Z_{00}=\binom{2i-2}{0}+\binom{2i-2}{i},\quad Z_{11}=\binom{2i-2}{i-2}+\binom{2i-2}{2i-2}, \quad Z_{01}=Z_{10}=\binom{2i-2}{i-1}.\]
Note that $Z_{00}=Z_{11}$ and $Z_{00}Z_{11}<Z_{01}Z_{10}$ by Lemma~\ref{cl:ineq} (since $i>2$), which gives that $Z_{00}=Z_{11}<Z_{01}=Z_{10}$. Lemma~\ref{hardness} thus implies that, for all sufficiently large $\Delta$, there exists $c>1$ such that  $\Hyper2Spinf$ is NP-hard.

Thus, it remains to consider the case $i=2$, i.e., $w_0=w_2=w_4=1$ and $w_1=w_3=0$. Let $j'$ be the largest integer such that $w_{0}=w_{2}=\hdots=w_{2j'}=1$ and $w_1=\hdots=w_{2j'-1}=0$. We may assume that $2j'<k$, otherwise $f=\evenf^{(k)}$. 

Assume first $w_{2j'+1}=1$. Let $e=\{x_1,x_2,\hdots,x_k\}$, pin $x_{2j'+2},\hdots,x_k$ to 0 and set $x:=x_{1}$. We have:
\begin{equation*}
\mu_{f;e}^{\condV}(\sigma_x=0)\propto\sum^{2j'}_{\ell=0}\binom{2j'}{\ell}w_\ell=2^{2j'-1},\quad
\mu_{f;e}^{\condV}(\sigma_x=1)\propto \sum^{2j'}_{\ell=0}\binom{2j'}{\ell}w_{\ell+1}=2^{2j'-1}+1.
\end{equation*}
By Lemma~\ref{lem:zeropinning}, we obtain that $f$ also supports pinning-to-1, in which case we fall back in Case I.

Thus, we may assume that $w_{2j'+1}=0$. We may further assume that $2j'+1<k$, otherwise $f=\evenf^{(k)}$. By the choice of $j'$, it follows that $w_{2j'+2}=0$. As before, let $e=\{x_1,x_2,\hdots,x_k\}$, pin $x_{2j'+3},\hdots,x_k$ to 0,  set $x:=x_{1}$, $y:=x_{2}$. We have:
\begin{align*}
\mu_{f;e}^{\condV}(\sigma_x=\sigma_y=0)&\propto\sum^{2j'}_{\ell=0}\binom{2j'}{\ell}w_\ell=2^{2j'-1},\\
\mu_{f;e}^{\condV}(\sigma_x=0,\sigma_y=1)&\propto\sum^{2j'}_{\ell=0}\binom{2j'}{\ell}w_{\ell+1}=2^{2j'-1}, \\
\mu_{f;e}^{\condV}(\sigma_x=\sigma_y=1)&\propto \sum^{2j'}_{\ell=0}\binom{2j'}{\ell}w_{\ell+2}=2^{2j'-1}-1.
\end{align*}
Once again, we obtain from Lemma~\ref{hardness} that for all sufficiently large $\Delta$, there exists $c>1$ such that $\Hyper2Spinf$ is $\NP$-hard.

\vskip 0.2cm\noindent \textbf{Case IIIb} There does not exist $j>i$ such that $w_j=1$. Then we have $w_0=w_i=1$ and $w_j=0$ for all $j\neq 0,i$ (recall also that $i\geq2$). We may assume that $i<k$, otherwise $f=\eqf^{(k)}$. We consider separately the cases $i=2$ and $i>2$. 

Consider first the case $i>2$. Let $e=\{x_1,x_2,\hdots,x_k\}$, pin $x_{i+2},\hdots,x_k$ to zero, force equality between $x_3,\hdots,x_{i+1}$ and set $x:=x_{1}, y:=x_{2}$. We have:
\begin{align*}
\mu_{f;e}^{\condV}(\sigma_x=0,\sigma_y=0)&\propto w_0+w_{i-1}=1,\\
\mu_{f;e}^{\condV}(\sigma_x=0,\sigma_y=1)&\propto w_1+w_{i}=1,\\
\mu_{f;e}^{\condV}(\sigma_x=1,\sigma_y=1)&\propto w_2+w_{i+1}=0.
\end{align*}
From Lemma~\ref{hardness}, we obtain that for all sufficiently large $\Delta$, there exists $c>1$ such that $\Hyper2Spinf$ is $\NP$-hard.

Finally, we consider the case $i=2$. We may assume that $k\geq 4$, otherwise $f=\evenf^{(3)}$. Let $e=\{x_1,x_2,\hdots,x_k\}$, pin $x_{5},\hdots,x_k$ to zero and set $x:=x_{1}, y:=x_{2}$ (note that $x_3,x_4$ are ``free"). We have:
\begin{align*}
\mu_{f;e}^{\condV}(\sigma_x=0,\sigma_y=0)&\propto \binom{2}{0}w_0+\binom{2}{2}w_2=2,\\
\mu_{f;e}^{\condV}(\sigma_x=0,\sigma_y=1)&\propto \binom{2}{1}w_2=2,\\
\mu_{f;e}^{\condV}(\sigma_x=1,\sigma_y=1)&\propto \binom{2}{0}w_2=1.
\end{align*}
From Lemma~\ref{hardness}, we obtain that for all sufficiently large $\Delta$, there exists $c>1$ such that $\Hyper2Spinf$ is $\NP$-hard.

\subsubsection{The case $i=1$}
In this case, we begin with the assumption that $w_0=w_1=1$ (and $f$ supports pinning-to-0).

Let $j$ be the minimum index $\ell>1$ such that $w_{\ell}=1$. Let $e=\{x_1,x_2,\hdots,x_k\}$. If $j\geq 3$ or such a $j$ does not exist, pin $x_{3},\hdots,x_k$ to 0, and  set $x:=x_{1}$, $y:=x_{2}$. We have:
\begin{align*}
\mu_{f;e}^{\condV}(\sigma_x=\sigma_y=0)&\propto w_0=1,\\
\mu_{f;e}^{\condV}(\sigma_x=0,\sigma_y=1)&\propto w_1=1,\\
\mu_{f;e}^{\condV}(\sigma_x=\sigma_y=1)&\propto w_2=0.
\end{align*}
It follows by Lemma~\ref{hardness} that for all sufficiently large $\Delta$, there exists $c>1$ such that $\Hyper2Spinf$ is $\NP$-hard.

 If $j=2$, let $j'$ be the first index $j'>j$ such that $w_{j'}=0$. We may assume that $j'$ exists otherwise $f=\onef^{(k)}$. We have $j'\geq 3$.  Let $e=\{x_1,x_2,\hdots,x_k\}$. Pin $x_{j'+1},\hdots,x_k$ to 0,  set $x:=x_{1}$, $y:=x_{2}$. We have:
\begin{align*}
\mu_{f;e}^{\condV}(\sigma_x=\sigma_y=0)&\propto\sum^{j'-2}_{\ell=0}\binom{j'-2}{\ell}w_\ell=2^{j'-2},\\
\mu_{f;e}^{\condV}(\sigma_x=0,\sigma_y=1)&\propto\sum^{j'-2}_{\ell=0}\binom{j'-2}{\ell}w_{\ell+1}=2^{j'-2},\\
\mu_{f;e}^{\condV}(\sigma_x=\sigma_y=1)&\propto \sum^{j'-2}_{\ell=0}\binom{j'-2}{\ell}w_{\ell+2}=2^{j'-2}-1.
\end{align*}
From Lemma~\ref{hardness}, for all sufficiently large $\Delta$, there exists $c>1$ such that $\Hyper2Spinf$ is $\NP$-hard.
 
 \bibliographystyle{plain}
 \bibliography{\jobname}

\end{document}

%% file: hypergraph1.tex
\begin{pgfpicture}{0bp}{0bp}{357.065631bp}{62.131238bp}
\begin{pgfscope}
\pgfsetlinewidth{1.0bp}
\pgfsetrectcap 
\pgfsetmiterjoin \pgfsetmiterlimit{10.0}
\pgfpathmoveto{\pgfpoint{126.292912bp}{32.563071bp}}
\pgfpathcurveto{\pgfpoint{126.292912bp}{29.307924bp}}{\pgfpoint{123.761805bp}{26.669105bp}}{\pgfpoint{120.639521bp}{26.669105bp}}
\pgfpathcurveto{\pgfpoint{117.517236bp}{26.669105bp}}{\pgfpoint{114.986123bp}{29.307924bp}}{\pgfpoint{114.986123bp}{32.563071bp}}
\pgfpathcurveto{\pgfpoint{114.986123bp}{35.818218bp}}{\pgfpoint{117.517236bp}{38.457036bp}}{\pgfpoint{120.639521bp}{38.457036bp}}
\pgfpathcurveto{\pgfpoint{123.761805bp}{38.457036bp}}{\pgfpoint{126.292912bp}{35.818218bp}}{\pgfpoint{126.292912bp}{32.563071bp}}
\pgfclosepath
\color[rgb]{0.0,0.0,0.0}\pgfseteorule\pgfusepath{fill}
\pgfpathmoveto{\pgfpoint{126.292912bp}{32.563071bp}}
\pgfpathcurveto{\pgfpoint{126.292912bp}{29.307924bp}}{\pgfpoint{123.761805bp}{26.669105bp}}{\pgfpoint{120.639521bp}{26.669105bp}}
\pgfpathcurveto{\pgfpoint{117.517236bp}{26.669105bp}}{\pgfpoint{114.986123bp}{29.307924bp}}{\pgfpoint{114.986123bp}{32.563071bp}}
\pgfpathcurveto{\pgfpoint{114.986123bp}{35.818218bp}}{\pgfpoint{117.517236bp}{38.457036bp}}{\pgfpoint{120.639521bp}{38.457036bp}}
\pgfpathcurveto{\pgfpoint{123.761805bp}{38.457036bp}}{\pgfpoint{126.292912bp}{35.818218bp}}{\pgfpoint{126.292912bp}{32.563071bp}}
\pgfclosepath
\color[rgb]{0.0,0.0,0.0}
\pgfusepath{stroke}
\end{pgfscope}
\begin{pgfscope}
\pgfsetlinewidth{1.0bp}
\pgfsetrectcap 
\pgfsetmiterjoin \pgfsetmiterlimit{10.0}
\pgfpathmoveto{\pgfpoint{163.092912bp}{32.963071bp}}
\pgfpathcurveto{\pgfpoint{163.092912bp}{29.707924bp}}{\pgfpoint{160.561805bp}{27.069105bp}}{\pgfpoint{157.439521bp}{27.069105bp}}
\pgfpathcurveto{\pgfpoint{154.317236bp}{27.069105bp}}{\pgfpoint{151.786123bp}{29.707924bp}}{\pgfpoint{151.786123bp}{32.963071bp}}
\pgfpathcurveto{\pgfpoint{151.786123bp}{36.218218bp}}{\pgfpoint{154.317236bp}{38.857036bp}}{\pgfpoint{157.439521bp}{38.857036bp}}
\pgfpathcurveto{\pgfpoint{160.561805bp}{38.857036bp}}{\pgfpoint{163.092912bp}{36.218218bp}}{\pgfpoint{163.092912bp}{32.963071bp}}
\pgfclosepath
\color[rgb]{0.0,0.0,0.0}\pgfseteorule\pgfusepath{fill}
\pgfpathmoveto{\pgfpoint{163.092912bp}{32.963071bp}}
\pgfpathcurveto{\pgfpoint{163.092912bp}{29.707924bp}}{\pgfpoint{160.561805bp}{27.069105bp}}{\pgfpoint{157.439521bp}{27.069105bp}}
\pgfpathcurveto{\pgfpoint{154.317236bp}{27.069105bp}}{\pgfpoint{151.786123bp}{29.707924bp}}{\pgfpoint{151.786123bp}{32.963071bp}}
\pgfpathcurveto{\pgfpoint{151.786123bp}{36.218218bp}}{\pgfpoint{154.317236bp}{38.857036bp}}{\pgfpoint{157.439521bp}{38.857036bp}}
\pgfpathcurveto{\pgfpoint{160.561805bp}{38.857036bp}}{\pgfpoint{163.092912bp}{36.218218bp}}{\pgfpoint{163.092912bp}{32.963071bp}}
\pgfclosepath
\color[rgb]{0.0,0.0,0.0}
\pgfusepath{stroke}
\end{pgfscope}
\begin{pgfscope}
\pgfsetlinewidth{1.0bp}
\pgfsetrectcap 
\pgfsetmiterjoin \pgfsetmiterlimit{10.0}
\pgfpathmoveto{\pgfpoint{243.492912bp}{32.563068bp}}
\pgfpathcurveto{\pgfpoint{243.492912bp}{29.307921bp}}{\pgfpoint{240.961805bp}{26.669102bp}}{\pgfpoint{237.839521bp}{26.669102bp}}
\pgfpathcurveto{\pgfpoint{234.717236bp}{26.669102bp}}{\pgfpoint{232.186123bp}{29.307921bp}}{\pgfpoint{232.186123bp}{32.563068bp}}
\pgfpathcurveto{\pgfpoint{232.186123bp}{35.818215bp}}{\pgfpoint{234.717236bp}{38.457033bp}}{\pgfpoint{237.839521bp}{38.457033bp}}
\pgfpathcurveto{\pgfpoint{240.961805bp}{38.457033bp}}{\pgfpoint{243.492912bp}{35.818215bp}}{\pgfpoint{243.492912bp}{32.563068bp}}
\pgfclosepath
\color[rgb]{0.0,0.0,0.0}\pgfseteorule\pgfusepath{fill}
\pgfpathmoveto{\pgfpoint{243.492912bp}{32.563068bp}}
\pgfpathcurveto{\pgfpoint{243.492912bp}{29.307921bp}}{\pgfpoint{240.961805bp}{26.669102bp}}{\pgfpoint{237.839521bp}{26.669102bp}}
\pgfpathcurveto{\pgfpoint{234.717236bp}{26.669102bp}}{\pgfpoint{232.186123bp}{29.307921bp}}{\pgfpoint{232.186123bp}{32.563068bp}}
\pgfpathcurveto{\pgfpoint{232.186123bp}{35.818215bp}}{\pgfpoint{234.717236bp}{38.457033bp}}{\pgfpoint{237.839521bp}{38.457033bp}}
\pgfpathcurveto{\pgfpoint{240.961805bp}{38.457033bp}}{\pgfpoint{243.492912bp}{35.818215bp}}{\pgfpoint{243.492912bp}{32.563068bp}}
\pgfclosepath
\color[rgb]{0.0,0.0,0.0}
\pgfusepath{stroke}
\end{pgfscope}
\begin{pgfscope}
\pgfsetlinewidth{1.0bp}
\pgfsetrectcap 
\pgfsetmiterjoin \pgfsetmiterlimit{10.0}
\pgfpathmoveto{\pgfpoint{326.492912bp}{32.563064bp}}
\pgfpathcurveto{\pgfpoint{326.492912bp}{29.307918bp}}{\pgfpoint{323.961805bp}{26.669099bp}}{\pgfpoint{320.839521bp}{26.669099bp}}
\pgfpathcurveto{\pgfpoint{317.717236bp}{26.669099bp}}{\pgfpoint{315.186123bp}{29.307918bp}}{\pgfpoint{315.186123bp}{32.563064bp}}
\pgfpathcurveto{\pgfpoint{315.186123bp}{35.818212bp}}{\pgfpoint{317.717236bp}{38.45703bp}}{\pgfpoint{320.839521bp}{38.45703bp}}
\pgfpathcurveto{\pgfpoint{323.961805bp}{38.45703bp}}{\pgfpoint{326.492912bp}{35.818212bp}}{\pgfpoint{326.492912bp}{32.563064bp}}
\pgfclosepath
\color[rgb]{0.0,0.0,0.0}\pgfseteorule\pgfusepath{fill}
\pgfpathmoveto{\pgfpoint{326.492912bp}{32.563064bp}}
\pgfpathcurveto{\pgfpoint{326.492912bp}{29.307918bp}}{\pgfpoint{323.961805bp}{26.669099bp}}{\pgfpoint{320.839521bp}{26.669099bp}}
\pgfpathcurveto{\pgfpoint{317.717236bp}{26.669099bp}}{\pgfpoint{315.186123bp}{29.307918bp}}{\pgfpoint{315.186123bp}{32.563064bp}}
\pgfpathcurveto{\pgfpoint{315.186123bp}{35.818212bp}}{\pgfpoint{317.717236bp}{38.45703bp}}{\pgfpoint{320.839521bp}{38.45703bp}}
\pgfpathcurveto{\pgfpoint{323.961805bp}{38.45703bp}}{\pgfpoint{326.492912bp}{35.818212bp}}{\pgfpoint{326.492912bp}{32.563064bp}}
\pgfclosepath
\color[rgb]{0.0,0.0,0.0}
\pgfusepath{stroke}
\end{pgfscope}
\begin{pgfscope}
\pgfsetlinewidth{1.0bp}
\pgfsetrectcap 
\pgfsetmiterjoin \pgfsetmiterlimit{10.0}
\pgfpathmoveto{\pgfpoint{42.892912bp}{32.563061bp}}
\pgfpathcurveto{\pgfpoint{42.892912bp}{29.307915bp}}{\pgfpoint{40.361805bp}{26.669096bp}}{\pgfpoint{37.239521bp}{26.669096bp}}
\pgfpathcurveto{\pgfpoint{34.117236bp}{26.669096bp}}{\pgfpoint{31.586123bp}{29.307915bp}}{\pgfpoint{31.586123bp}{32.563061bp}}
\pgfpathcurveto{\pgfpoint{31.586123bp}{35.818209bp}}{\pgfpoint{34.117236bp}{38.457027bp}}{\pgfpoint{37.239521bp}{38.457027bp}}
\pgfpathcurveto{\pgfpoint{40.361805bp}{38.457027bp}}{\pgfpoint{42.892912bp}{35.818209bp}}{\pgfpoint{42.892912bp}{32.563061bp}}
\pgfclosepath
\color[rgb]{0.0,0.0,0.0}\pgfseteorule\pgfusepath{fill}
\pgfpathmoveto{\pgfpoint{42.892912bp}{32.563061bp}}
\pgfpathcurveto{\pgfpoint{42.892912bp}{29.307915bp}}{\pgfpoint{40.361805bp}{26.669096bp}}{\pgfpoint{37.239521bp}{26.669096bp}}
\pgfpathcurveto{\pgfpoint{34.117236bp}{26.669096bp}}{\pgfpoint{31.586123bp}{29.307915bp}}{\pgfpoint{31.586123bp}{32.563061bp}}
\pgfpathcurveto{\pgfpoint{31.586123bp}{35.818209bp}}{\pgfpoint{34.117236bp}{38.457027bp}}{\pgfpoint{37.239521bp}{38.457027bp}}
\pgfpathcurveto{\pgfpoint{40.361805bp}{38.457027bp}}{\pgfpoint{42.892912bp}{35.818209bp}}{\pgfpoint{42.892912bp}{32.563061bp}}
\pgfclosepath
\color[rgb]{0.0,0.0,0.0}
\pgfusepath{stroke}
\end{pgfscope}
\begin{pgfscope}
\pgftransformcm{1.0}{0.0}{0.0}{1.0}{\pgfpoint{114.900006bp}{13.699988bp}}
\pgftext[left,base]{\sffamily\mdseries\upshape\huge
\color[rgb]{0.0,0.0,0.0}$z_1$}
\end{pgfscope}
\begin{pgfscope}
\pgftransformcm{1.0}{0.0}{0.0}{1.0}{\pgfpoint{177.900006bp}{31.699988bp}}
\pgftext[left,base]{\sffamily\mdseries\upshape\fontsize{36.135}{36.135}\selectfont
\color[rgb]{0.0,0.0,0.0}. . .}
\end{pgfscope}
\begin{pgfscope}
\pgftransformcm{1.0}{0.0}{0.0}{1.0}{\pgfpoint{150.500006bp}{14.099988bp}}
\pgftext[left,base]{\sffamily\mdseries\upshape\huge
\color[rgb]{0.0,0.0,0.0}$z_2$}
\end{pgfscope}
\begin{pgfscope}
\pgftransformcm{1.0}{0.0}{0.0}{1.0}{\pgfpoint{230.300006bp}{14.099988bp}}
\pgftext[left,base]{\sffamily\mdseries\upshape\huge
\color[rgb]{0.0,0.0,0.0}$z_{k-1}$}
\end{pgfscope}
\begin{pgfscope}
\pgfsetlinewidth{1.0bp}
\pgfsetrectcap 
\pgfsetmiterjoin \pgfsetmiterlimit{10.0}
\pgfpathmoveto{\pgfpoint{355.500012bp}{29.899985bp}}
\pgfpathcurveto{\pgfpoint{355.500012bp}{13.662817bp}}{\pgfpoint{290.223126bp}{0.499991bp}}{\pgfpoint{209.699994bp}{0.499991bp}}
\pgfpathcurveto{\pgfpoint{129.176892bp}{0.499991bp}}{\pgfpoint{63.900006bp}{13.662817bp}}{\pgfpoint{63.900006bp}{29.899985bp}}
\pgfpathcurveto{\pgfpoint{63.900006bp}{46.137152bp}}{\pgfpoint{129.176892bp}{59.299994bp}}{\pgfpoint{209.699994bp}{59.299994bp}}
\pgfpathcurveto{\pgfpoint{290.223126bp}{59.299994bp}}{\pgfpoint{355.500012bp}{46.137152bp}}{\pgfpoint{355.500012bp}{29.899985bp}}
\pgfclosepath
\color[rgb]{0.0,0.0,0.0}
\pgfusepath{stroke}
\end{pgfscope}
\begin{pgfscope}
\pgfsetlinewidth{1.0bp}
\pgfsetrectcap 
\pgfsetmiterjoin \pgfsetmiterlimit{10.0}
\pgfpathmoveto{\pgfpoint{292.100012bp}{30.899985bp}}
\pgfpathcurveto{\pgfpoint{292.100012bp}{14.662817bp}}{\pgfpoint{226.823126bp}{1.499991bp}}{\pgfpoint{146.299994bp}{1.499991bp}}
\pgfpathcurveto{\pgfpoint{65.776892bp}{1.499991bp}}{\pgfpoint{0.500006bp}{14.662817bp}}{\pgfpoint{0.500006bp}{30.899985bp}}
\pgfpathcurveto{\pgfpoint{0.500006bp}{47.137152bp}}{\pgfpoint{65.776892bp}{60.299994bp}}{\pgfpoint{146.299994bp}{60.299994bp}}
\pgfpathcurveto{\pgfpoint{226.823126bp}{60.299994bp}}{\pgfpoint{292.100012bp}{47.137152bp}}{\pgfpoint{292.100012bp}{30.899985bp}}
\pgfclosepath
\color[rgb]{0.0,0.0,0.0}
\pgfusepath{stroke}
\end{pgfscope}
\begin{pgfscope}
\pgftransformcm{1.0}{0.0}{0.0}{1.0}{\pgfpoint{30.500006bp}{15.299988bp}}
\pgftext[left,base]{\sffamily\mdseries\upshape\huge
\color[rgb]{0.0,0.0,0.0}$x$}
\end{pgfscope}
\begin{pgfscope}
\pgftransformcm{1.0}{0.0}{0.0}{1.0}{\pgfpoint{314.500006bp}{15.099988bp}}
\pgftext[left,base]{\sffamily\mdseries\upshape\huge
\color[rgb]{0.0,0.0,0.0}$y$}
\end{pgfscope}
\begin{pgfscope}
\pgftransformcm{1.0}{-0.0}{0.0}{1.0}{\pgfpoint{-0.499994bp}{47.665613bp}}
\pgftext[left,bottom]{\sffamily\mdseries\upshape\huge
\color[rgb]{0.0,0.0,0.0}$e_1$}
\end{pgfscope}
\begin{pgfscope}
\pgftransformcm{1.0}{-0.0}{0.0}{1.0}{\pgfpoint{336.900006bp}{47.665613bp}}
\pgftext[left,bottom]{\sffamily\mdseries\upshape\huge
\color[rgb]{0.0,0.0,0.0}$e_2$}
\end{pgfscope}
\end{pgfpicture}

%% file: hypergraph2a.tex
\begin{pgfpicture}{0bp}{0bp}{223.715623bp}{148.599998bp}
\begin{pgfscope}
\pgfsetlinewidth{1.0bp}
\pgfsetrectcap 
\pgfsetmiterjoin \pgfsetmiterlimit{10.0}
\pgfpathmoveto{\pgfpoint{58.808531bp}{36.16308bp}}
\pgfpathcurveto{\pgfpoint{58.808531bp}{32.907934bp}}{\pgfpoint{56.277424bp}{30.269114bp}}{\pgfpoint{53.155139bp}{30.269114bp}}
\pgfpathcurveto{\pgfpoint{50.032855bp}{30.269114bp}}{\pgfpoint{47.501742bp}{32.907934bp}}{\pgfpoint{47.501742bp}{36.16308bp}}
\pgfpathcurveto{\pgfpoint{47.501742bp}{39.418227bp}}{\pgfpoint{50.032855bp}{42.057045bp}}{\pgfpoint{53.155139bp}{42.057045bp}}
\pgfpathcurveto{\pgfpoint{56.277424bp}{42.057045bp}}{\pgfpoint{58.808531bp}{39.418227bp}}{\pgfpoint{58.808531bp}{36.16308bp}}
\pgfclosepath
\color[rgb]{0.0,0.0,0.0}\pgfseteorule\pgfusepath{fill}
\pgfpathmoveto{\pgfpoint{58.808531bp}{36.16308bp}}
\pgfpathcurveto{\pgfpoint{58.808531bp}{32.907934bp}}{\pgfpoint{56.277424bp}{30.269114bp}}{\pgfpoint{53.155139bp}{30.269114bp}}
\pgfpathcurveto{\pgfpoint{50.032855bp}{30.269114bp}}{\pgfpoint{47.501742bp}{32.907934bp}}{\pgfpoint{47.501742bp}{36.16308bp}}
\pgfpathcurveto{\pgfpoint{47.501742bp}{39.418227bp}}{\pgfpoint{50.032855bp}{42.057045bp}}{\pgfpoint{53.155139bp}{42.057045bp}}
\pgfpathcurveto{\pgfpoint{56.277424bp}{42.057045bp}}{\pgfpoint{58.808531bp}{39.418227bp}}{\pgfpoint{58.808531bp}{36.16308bp}}
\pgfclosepath
\color[rgb]{0.0,0.0,0.0}
\pgfusepath{stroke}
\end{pgfscope}
\begin{pgfscope}
\pgfsetlinewidth{1.0bp}
\pgfsetrectcap 
\pgfsetmiterjoin \pgfsetmiterlimit{10.0}
\pgfpathmoveto{\pgfpoint{95.608531bp}{36.56308bp}}
\pgfpathcurveto{\pgfpoint{95.608531bp}{33.307934bp}}{\pgfpoint{93.077424bp}{30.669114bp}}{\pgfpoint{89.955139bp}{30.669114bp}}
\pgfpathcurveto{\pgfpoint{86.832855bp}{30.669114bp}}{\pgfpoint{84.301742bp}{33.307934bp}}{\pgfpoint{84.301742bp}{36.56308bp}}
\pgfpathcurveto{\pgfpoint{84.301742bp}{39.818227bp}}{\pgfpoint{86.832855bp}{42.457045bp}}{\pgfpoint{89.955139bp}{42.457045bp}}
\pgfpathcurveto{\pgfpoint{93.077424bp}{42.457045bp}}{\pgfpoint{95.608531bp}{39.818227bp}}{\pgfpoint{95.608531bp}{36.56308bp}}
\pgfclosepath
\color[rgb]{0.0,0.0,0.0}\pgfseteorule\pgfusepath{fill}
\pgfpathmoveto{\pgfpoint{95.608531bp}{36.56308bp}}
\pgfpathcurveto{\pgfpoint{95.608531bp}{33.307934bp}}{\pgfpoint{93.077424bp}{30.669114bp}}{\pgfpoint{89.955139bp}{30.669114bp}}
\pgfpathcurveto{\pgfpoint{86.832855bp}{30.669114bp}}{\pgfpoint{84.301742bp}{33.307934bp}}{\pgfpoint{84.301742bp}{36.56308bp}}
\pgfpathcurveto{\pgfpoint{84.301742bp}{39.818227bp}}{\pgfpoint{86.832855bp}{42.457045bp}}{\pgfpoint{89.955139bp}{42.457045bp}}
\pgfpathcurveto{\pgfpoint{93.077424bp}{42.457045bp}}{\pgfpoint{95.608531bp}{39.818227bp}}{\pgfpoint{95.608531bp}{36.56308bp}}
\pgfclosepath
\color[rgb]{0.0,0.0,0.0}
\pgfusepath{stroke}
\end{pgfscope}
\begin{pgfscope}
\pgfsetlinewidth{1.0bp}
\pgfsetrectcap 
\pgfsetmiterjoin \pgfsetmiterlimit{10.0}
\pgfpathmoveto{\pgfpoint{176.008531bp}{36.163077bp}}
\pgfpathcurveto{\pgfpoint{176.008531bp}{32.907931bp}}{\pgfpoint{173.477424bp}{30.269111bp}}{\pgfpoint{170.355139bp}{30.269111bp}}
\pgfpathcurveto{\pgfpoint{167.232855bp}{30.269111bp}}{\pgfpoint{164.701742bp}{32.907931bp}}{\pgfpoint{164.701742bp}{36.163077bp}}
\pgfpathcurveto{\pgfpoint{164.701742bp}{39.418224bp}}{\pgfpoint{167.232855bp}{42.057042bp}}{\pgfpoint{170.355139bp}{42.057042bp}}
\pgfpathcurveto{\pgfpoint{173.477424bp}{42.057042bp}}{\pgfpoint{176.008531bp}{39.418224bp}}{\pgfpoint{176.008531bp}{36.163077bp}}
\pgfclosepath
\color[rgb]{0.0,0.0,0.0}\pgfseteorule\pgfusepath{fill}
\pgfpathmoveto{\pgfpoint{176.008531bp}{36.163077bp}}
\pgfpathcurveto{\pgfpoint{176.008531bp}{32.907931bp}}{\pgfpoint{173.477424bp}{30.269111bp}}{\pgfpoint{170.355139bp}{30.269111bp}}
\pgfpathcurveto{\pgfpoint{167.232855bp}{30.269111bp}}{\pgfpoint{164.701742bp}{32.907931bp}}{\pgfpoint{164.701742bp}{36.163077bp}}
\pgfpathcurveto{\pgfpoint{164.701742bp}{39.418224bp}}{\pgfpoint{167.232855bp}{42.057042bp}}{\pgfpoint{170.355139bp}{42.057042bp}}
\pgfpathcurveto{\pgfpoint{173.477424bp}{42.057042bp}}{\pgfpoint{176.008531bp}{39.418224bp}}{\pgfpoint{176.008531bp}{36.163077bp}}
\pgfclosepath
\color[rgb]{0.0,0.0,0.0}
\pgfusepath{stroke}
\end{pgfscope}
\begin{pgfscope}
\pgftransformcm{1.0}{0.0}{0.0}{1.0}{\pgfpoint{43.815625bp}{19.099997bp}}
\pgftext[left,base]{\sffamily\mdseries\upshape\huge
\color[rgb]{0.0,0.0,0.0}$v_1$}
\end{pgfscope}
\begin{pgfscope}
\pgftransformcm{1.0}{0.0}{0.0}{1.0}{\pgfpoint{80.015625bp}{19.499997bp}}
\pgftext[left,base]{\sffamily\mdseries\upshape\huge
\color[rgb]{0.0,0.0,0.0}$v_2$}
\end{pgfscope}
\begin{pgfscope}
\pgftransformcm{1.0}{0.0}{0.0}{1.0}{\pgfpoint{162.815625bp}{18.899997bp}}
\pgftext[left,base]{\sffamily\mdseries\upshape\huge
\color[rgb]{0.0,0.0,0.0}$v_{k}$}
\end{pgfscope}
\begin{pgfscope}
\pgfsetlinewidth{1.0bp}
\pgfsetrectcap 
\pgfsetmiterjoin \pgfsetmiterlimit{10.0}
\pgfpathmoveto{\pgfpoint{223.215625bp}{32.300003bp}}
\pgfpathcurveto{\pgfpoint{223.215625bp}{14.737335bp}}{\pgfpoint{173.78786bp}{0.5bp}}{\pgfpoint{112.815631bp}{0.5bp}}
\pgfpathcurveto{\pgfpoint{51.843402bp}{0.5bp}}{\pgfpoint{2.415637bp}{14.737335bp}}{\pgfpoint{2.415637bp}{32.300003bp}}
\pgfpathcurveto{\pgfpoint{2.415637bp}{49.862656bp}}{\pgfpoint{51.843402bp}{64.099991bp}}{\pgfpoint{112.815631bp}{64.099991bp}}
\pgfpathcurveto{\pgfpoint{173.78786bp}{64.099991bp}}{\pgfpoint{223.215625bp}{49.862656bp}}{\pgfpoint{223.215625bp}{32.300003bp}}
\pgfclosepath
\color[rgb]{0.0,0.0,0.0}
\pgfusepath{stroke}
\end{pgfscope}
\begin{pgfscope}
\pgfsetlinewidth{1.0bp}
\pgfsetrectcap 
\pgfsetmiterjoin \pgfsetmiterlimit{10.0}
\pgfsetdash{{10.0bp}{5.0bp}}{0.0bp}
\pgfpathmoveto{\pgfpoint{192.615625bp}{11.299997bp}}
\pgfpathlineto{\pgfpoint{192.615625bp}{148.099997bp}}
\pgfpathlineto{\pgfpoint{33.015625bp}{148.099997bp}}
\pgfpathlineto{\pgfpoint{33.015625bp}{11.299997bp}}
\pgfpathlineto{\pgfpoint{192.615625bp}{11.299997bp}}
\pgfclosepath
\color[rgb]{0.0,0.0,0.0}
\pgfusepath{stroke}
\end{pgfscope}
\begin{pgfscope}
\pgftransformcm{1.0}{0.0}{0.0}{1.0}{\pgfpoint{-1.584375bp}{125.099997bp}}
\pgftext[left,base]{\sffamily\mdseries\upshape\huge
\color[rgb]{0.0,0.0,0.0}$H''$}
\end{pgfscope}
\begin{pgfscope}
\pgftransformcm{1.0}{0.0}{0.0}{1.0}{\pgfpoint{2.015625bp}{6.299997bp}}
\pgftext[left,base]{\sffamily\mdseries\upshape\huge
\color[rgb]{0.0,0.0,0.0}$e$}
\end{pgfscope}
\begin{pgfscope}
\pgfsetlinewidth{1.0bp}
\pgfsetrectcap 
\pgfsetmiterjoin \pgfsetmiterlimit{10.0}
\pgfsetdash{{1.0bp}{5.0bp}}{0.0bp}
\pgfpathmoveto{\pgfpoint{184.024375bp}{15.091239bp}}
\pgfpathcurveto{\pgfpoint{175.823378bp}{6.890242bp}}{\pgfpoint{145.621386bp}{23.795743bp}}{\pgfpoint{116.566393bp}{52.850737bp}}
\pgfpathcurveto{\pgfpoint{87.511377bp}{81.905752bp}}{\pgfpoint{70.605887bp}{112.107733bp}}{\pgfpoint{78.806884bp}{120.30873bp}}
\pgfpathcurveto{\pgfpoint{87.007903bp}{128.509749bp}}{\pgfpoint{117.209862bp}{111.604237bp}}{\pgfpoint{146.264877bp}{82.549222bp}}
\pgfpathcurveto{\pgfpoint{175.319871bp}{53.494228bp}}{\pgfpoint{192.225394bp}{23.292258bp}}{\pgfpoint{184.024375bp}{15.091239bp}}
\pgfclosepath
\color[rgb]{0.0,0.0,0.0}
\pgfusepath{stroke}
\end{pgfscope}
\begin{pgfscope}
\pgfsetlinewidth{1.0bp}
\pgfsetrectcap 
\pgfsetmiterjoin \pgfsetmiterlimit{10.0}
\pgfsetdash{{1.0bp}{5.0bp}}{0.0bp}
\pgfpathmoveto{\pgfpoint{77.809597bp}{14.891466bp}}
\pgfpathcurveto{\pgfpoint{69.608599bp}{23.092463bp}}{\pgfpoint{86.514101bp}{53.294455bp}}{\pgfpoint{115.569094bp}{82.349448bp}}
\pgfpathcurveto{\pgfpoint{144.624109bp}{111.404463bp}}{\pgfpoint{174.826091bp}{128.309954bp}}{\pgfpoint{183.027088bp}{120.108957bp}}
\pgfpathcurveto{\pgfpoint{191.228107bp}{111.907938bp}}{\pgfpoint{174.322594bp}{81.705979bp}}{\pgfpoint{145.267579bp}{52.650964bp}}
\pgfpathcurveto{\pgfpoint{116.212586bp}{23.59597bp}}{\pgfpoint{86.010615bp}{6.690447bp}}{\pgfpoint{77.809597bp}{14.891466bp}}
\pgfclosepath
\color[rgb]{0.0,0.0,0.0}
\pgfusepath{stroke}
\end{pgfscope}
\begin{pgfscope}
\pgfsetlinewidth{1.0bp}
\pgfsetrectcap 
\pgfsetmiterjoin \pgfsetmiterlimit{10.0}
\pgfsetdash{{1.0bp}{5.0bp}}{0.0bp}
\pgfpathmoveto{\pgfpoint{50.757224bp}{12.842333bp}}
\pgfpathcurveto{\pgfpoint{39.554454bp}{15.844107bp}}{\pgfpoint{39.094051bp}{50.45255bp}}{\pgfpoint{49.728917bp}{90.142409bp}}
\pgfpathcurveto{\pgfpoint{60.363791bp}{129.832298bp}}{\pgfpoint{75.066728bp}{143.973873bp}}{\pgfpoint{86.269499bp}{140.972099bp}}
\pgfpathcurveto{\pgfpoint{97.472299bp}{137.970318bp}}{\pgfpoint{100.932675bp}{118.961898bp}}{\pgfpoint{90.297802bp}{79.272009bp}}
\pgfpathcurveto{\pgfpoint{79.662936bp}{39.58215bp}}{\pgfpoint{61.960024bp}{9.840552bp}}{\pgfpoint{50.757224bp}{12.842333bp}}
\pgfclosepath
\color[rgb]{0.0,0.0,0.0}
\pgfusepath{stroke}
\end{pgfscope}
\begin{pgfscope}
\pgfsetlinewidth{1.0bp}
\pgfsetrectcap 
\pgfsetmiterjoin \pgfsetmiterlimit{10.0}
\pgfsetdash{{1.0bp}{5.0bp}}{0.0bp}
\pgfpathmoveto{\pgfpoint{44.952803bp}{93.043648bp}}
\pgfpathcurveto{\pgfpoint{41.951029bp}{104.246419bp}}{\pgfpoint{71.692619bp}{121.94936bp}}{\pgfpoint{111.382478bp}{132.584226bp}}
\pgfpathcurveto{\pgfpoint{151.072367bp}{143.2191bp}}{\pgfpoint{185.680795bp}{142.758693bp}}{\pgfpoint{188.682568bp}{131.555923bp}}
\pgfpathcurveto{\pgfpoint{191.68435bp}{120.353123bp}}{\pgfpoint{161.942767bp}{102.650215bp}}{\pgfpoint{122.252878bp}{92.015341bp}}
\pgfpathcurveto{\pgfpoint{82.563019bp}{81.380476bp}}{\pgfpoint{47.954584bp}{81.840848bp}}{\pgfpoint{44.952803bp}{93.043648bp}}
\pgfclosepath
\color[rgb]{0.0,0.0,0.0}
\pgfusepath{stroke}
\end{pgfscope}
\begin{pgfscope}
\pgftransformcm{1.40625}{-0.0}{0.0}{1.40625}{\pgfpoint{109.886426bp}{34.237497bp}}
\pgftext[left,base]{\sffamily\mdseries\upshape\fontsize{36.135}{36.135}\selectfont
\color[rgb]{0.0,0.0,0.0}$\hdots$}
\end{pgfscope}
\end{pgfpicture}

%% file: hypergraph2b.tex
\begin{pgfpicture}{0bp}{0bp}{284.915619bp}{127.0bp}
\begin{pgfscope}
\pgfsetlinewidth{1.0bp}
\pgfsetrectcap 
\pgfsetmiterjoin \pgfsetmiterlimit{10.0}
\pgfpathmoveto{\pgfpoint{61.408531bp}{34.763089bp}}
\pgfpathcurveto{\pgfpoint{61.408531bp}{31.507943bp}}{\pgfpoint{58.877424bp}{28.869124bp}}{\pgfpoint{55.755139bp}{28.869124bp}}
\pgfpathcurveto{\pgfpoint{52.632855bp}{28.869124bp}}{\pgfpoint{50.101742bp}{31.507943bp}}{\pgfpoint{50.101742bp}{34.763089bp}}
\pgfpathcurveto{\pgfpoint{50.101742bp}{38.018237bp}}{\pgfpoint{52.632855bp}{40.657054bp}}{\pgfpoint{55.755139bp}{40.657054bp}}
\pgfpathcurveto{\pgfpoint{58.877424bp}{40.657054bp}}{\pgfpoint{61.408531bp}{38.018237bp}}{\pgfpoint{61.408531bp}{34.763089bp}}
\pgfclosepath
\color[rgb]{0.0,0.0,0.0}\pgfseteorule\pgfusepath{fill}
\pgfpathmoveto{\pgfpoint{61.408531bp}{34.763089bp}}
\pgfpathcurveto{\pgfpoint{61.408531bp}{31.507943bp}}{\pgfpoint{58.877424bp}{28.869124bp}}{\pgfpoint{55.755139bp}{28.869124bp}}
\pgfpathcurveto{\pgfpoint{52.632855bp}{28.869124bp}}{\pgfpoint{50.101742bp}{31.507943bp}}{\pgfpoint{50.101742bp}{34.763089bp}}
\pgfpathcurveto{\pgfpoint{50.101742bp}{38.018237bp}}{\pgfpoint{52.632855bp}{40.657054bp}}{\pgfpoint{55.755139bp}{40.657054bp}}
\pgfpathcurveto{\pgfpoint{58.877424bp}{40.657054bp}}{\pgfpoint{61.408531bp}{38.018237bp}}{\pgfpoint{61.408531bp}{34.763089bp}}
\pgfclosepath
\color[rgb]{0.0,0.0,0.0}
\pgfusepath{stroke}
\end{pgfscope}
\begin{pgfscope}
\pgfsetlinewidth{1.0bp}
\pgfsetrectcap 
\pgfsetmiterjoin \pgfsetmiterlimit{10.0}
\pgfpathmoveto{\pgfpoint{141.808531bp}{34.363086bp}}
\pgfpathcurveto{\pgfpoint{141.808531bp}{31.10794bp}}{\pgfpoint{139.277424bp}{28.469121bp}}{\pgfpoint{136.155139bp}{28.469121bp}}
\pgfpathcurveto{\pgfpoint{133.032855bp}{28.469121bp}}{\pgfpoint{130.501742bp}{31.10794bp}}{\pgfpoint{130.501742bp}{34.363086bp}}
\pgfpathcurveto{\pgfpoint{130.501742bp}{37.618233bp}}{\pgfpoint{133.032855bp}{40.257051bp}}{\pgfpoint{136.155139bp}{40.257051bp}}
\pgfpathcurveto{\pgfpoint{139.277424bp}{40.257051bp}}{\pgfpoint{141.808531bp}{37.618233bp}}{\pgfpoint{141.808531bp}{34.363086bp}}
\pgfclosepath
\color[rgb]{0.0,0.0,0.0}\pgfseteorule\pgfusepath{fill}
\pgfpathmoveto{\pgfpoint{141.808531bp}{34.363086bp}}
\pgfpathcurveto{\pgfpoint{141.808531bp}{31.10794bp}}{\pgfpoint{139.277424bp}{28.469121bp}}{\pgfpoint{136.155139bp}{28.469121bp}}
\pgfpathcurveto{\pgfpoint{133.032855bp}{28.469121bp}}{\pgfpoint{130.501742bp}{31.10794bp}}{\pgfpoint{130.501742bp}{34.363086bp}}
\pgfpathcurveto{\pgfpoint{130.501742bp}{37.618233bp}}{\pgfpoint{133.032855bp}{40.257051bp}}{\pgfpoint{136.155139bp}{40.257051bp}}
\pgfpathcurveto{\pgfpoint{139.277424bp}{40.257051bp}}{\pgfpoint{141.808531bp}{37.618233bp}}{\pgfpoint{141.808531bp}{34.363086bp}}
\pgfclosepath
\color[rgb]{0.0,0.0,0.0}
\pgfusepath{stroke}
\end{pgfscope}
\begin{pgfscope}
\pgftransformcm{1.0}{0.0}{0.0}{1.0}{\pgfpoint{48.215625bp}{18.300006bp}}
\pgftext[left,base]{\sffamily\mdseries\upshape\huge
\color[rgb]{0.0,0.0,0.0}$v_1$}
\end{pgfscope}
\begin{pgfscope}
\pgftransformcm{1.0}{0.0}{0.0}{1.0}{\pgfpoint{128.615625bp}{17.100006bp}}
\pgftext[left,base]{\sffamily\mdseries\upshape\huge
\color[rgb]{0.0,0.0,0.0}$v_{t}$}
\end{pgfscope}
\begin{pgfscope}
\pgfsetlinewidth{1.0bp}
\pgfsetrectcap 
\pgfsetmiterjoin \pgfsetmiterlimit{10.0}
\pgfpathmoveto{\pgfpoint{284.415625bp}{31.700012bp}}
\pgfpathcurveto{\pgfpoint{284.415625bp}{14.137344bp}}{\pgfpoint{208.58786bp}{0.500009bp}}{\pgfpoint{147.615631bp}{0.500009bp}}
\pgfpathcurveto{\pgfpoint{86.643402bp}{0.500009bp}}{\pgfpoint{37.215637bp}{14.737344bp}}{\pgfpoint{37.215637bp}{32.300012bp}}
\pgfpathcurveto{\pgfpoint{37.215637bp}{49.862665bp}}{\pgfpoint{86.643402bp}{64.1bp}}{\pgfpoint{147.615631bp}{64.1bp}}
\pgfpathcurveto{\pgfpoint{208.58786bp}{64.1bp}}{\pgfpoint{284.415625bp}{49.262665bp}}{\pgfpoint{284.415625bp}{31.700012bp}}
\pgfclosepath
\color[rgb]{0.0,0.0,0.0}
\pgfusepath{stroke}
\end{pgfscope}
\begin{pgfscope}
\pgftransformcm{1.0}{0.0}{0.0}{1.0}{\pgfpoint{-1.584375bp}{107.100006bp}}
\pgftext[left,base]{\sffamily\mdseries\upshape\huge
\color[rgb]{0.0,0.0,0.0}$H''$}
\end{pgfscope}
\begin{pgfscope}
\pgftransformcm{1.0}{0.0}{0.0}{1.0}{\pgfpoint{256.415625bp}{54.300006bp}}
\pgftext[left,base]{\sffamily\mdseries\upshape\huge
\color[rgb]{0.0,0.0,0.0}$e$}
\end{pgfscope}
\begin{pgfscope}
\pgfsetlinewidth{1.0bp}
\pgfsetrectcap 
\pgfsetmiterjoin \pgfsetmiterlimit{10.0}
\pgfsetdash{{1.0bp}{5.0bp}}{0.0bp}
\pgfpathmoveto{\pgfpoint{149.824375bp}{13.291248bp}}
\pgfpathcurveto{\pgfpoint{141.623378bp}{5.090251bp}}{\pgfpoint{111.421386bp}{21.995753bp}}{\pgfpoint{82.366393bp}{51.050746bp}}
\pgfpathcurveto{\pgfpoint{53.311377bp}{80.105761bp}}{\pgfpoint{36.405887bp}{110.307742bp}}{\pgfpoint{44.606884bp}{118.508739bp}}
\pgfpathcurveto{\pgfpoint{52.807903bp}{126.709758bp}}{\pgfpoint{83.009862bp}{109.804246bp}}{\pgfpoint{112.064877bp}{80.749231bp}}
\pgfpathcurveto{\pgfpoint{141.119871bp}{51.694237bp}}{\pgfpoint{158.025394bp}{21.492267bp}}{\pgfpoint{149.824375bp}{13.291248bp}}
\pgfclosepath
\color[rgb]{0.0,0.0,0.0}
\pgfusepath{stroke}
\end{pgfscope}
\begin{pgfscope}
\pgfsetlinewidth{1.0bp}
\pgfsetrectcap 
\pgfsetmiterjoin \pgfsetmiterlimit{10.0}
\pgfsetdash{{1.0bp}{5.0bp}}{0.0bp}
\pgfpathmoveto{\pgfpoint{43.609597bp}{13.091475bp}}
\pgfpathcurveto{\pgfpoint{35.408599bp}{21.292472bp}}{\pgfpoint{52.314101bp}{51.494464bp}}{\pgfpoint{81.369094bp}{80.549457bp}}
\pgfpathcurveto{\pgfpoint{110.424109bp}{109.604473bp}}{\pgfpoint{140.626091bp}{126.509963bp}}{\pgfpoint{148.827088bp}{118.308966bp}}
\pgfpathcurveto{\pgfpoint{157.028107bp}{110.107947bp}}{\pgfpoint{140.122594bp}{79.905988bp}}{\pgfpoint{111.067579bp}{50.850973bp}}
\pgfpathcurveto{\pgfpoint{82.012586bp}{21.795979bp}}{\pgfpoint{51.810615bp}{4.890456bp}}{\pgfpoint{43.609597bp}{13.091475bp}}
\pgfclosepath
\color[rgb]{0.0,0.0,0.0}
\pgfusepath{stroke}
\end{pgfscope}
\begin{pgfscope}
\pgftransformcm{1.40625}{-0.0}{0.0}{1.40625}{\pgfpoint{75.686426bp}{32.437506bp}}
\pgftext[left,base]{\sffamily\mdseries\upshape\fontsize{36.135}{36.135}\selectfont
\color[rgb]{0.0,0.0,0.0}$\hdots$}
\end{pgfscope}
\begin{pgfscope}
\pgftransformcm{1.0}{0.0}{0.0}{1.0}{\pgfpoint{169.815625bp}{16.700006bp}}
\pgftext[left,base]{\sffamily\mdseries\upshape\huge
\color[rgb]{0.0,0.0,0.0}$v_{t+1}$}
\end{pgfscope}
\begin{pgfscope}
\pgftransformcm{1.0}{0.0}{0.0}{1.0}{\pgfpoint{250.215625bp}{18.500006bp}}
\pgftext[left,base]{\sffamily\mdseries\upshape\huge
\color[rgb]{0.0,0.0,0.0}$v_{k}$}
\end{pgfscope}
\begin{pgfscope}
\pgfsetlinewidth{1.0bp}
\pgfsetrectcap 
\pgfsetmiterjoin \pgfsetmiterlimit{10.0}
\pgfsetdash{{10.0bp}{5.0bp}}{0.0bp}
\pgfpathmoveto{\pgfpoint{165.015625bp}{8.900006bp}}
\pgfpathlineto{\pgfpoint{165.015625bp}{126.500006bp}}
\pgfpathlineto{\pgfpoint{30.015625bp}{126.500006bp}}
\pgfpathlineto{\pgfpoint{30.015625bp}{8.900006bp}}
\pgfpathlineto{\pgfpoint{165.015625bp}{8.900006bp}}
\pgfclosepath
\color[rgb]{0.0,0.0,0.0}
\pgfusepath{stroke}
\end{pgfscope}
\begin{pgfscope}
\pgfsetlinewidth{1.0bp}
\pgfsetrectcap 
\pgfsetmiterjoin \pgfsetmiterlimit{10.0}
\pgfsetdash{{1.0bp}{5.0bp}}{0.0bp}
\pgfpathmoveto{\pgfpoint{34.733536bp}{103.919003bp}}
\pgfpathcurveto{\pgfpoint{34.733536bp}{115.516965bp}}{\pgfpoint{56.509153bp}{125.206289bp}}{\pgfpoint{97.599119bp}{125.206289bp}}
\pgfpathcurveto{\pgfpoint{138.689115bp}{125.206289bp}}{\pgfpoint{163.231923bp}{115.347893bp}}{\pgfpoint{163.231923bp}{103.749931bp}}
\pgfpathcurveto{\pgfpoint{163.231923bp}{92.151939bp}}{\pgfpoint{138.689115bp}{83.206289bp}}{\pgfpoint{97.599119bp}{83.206289bp}}
\pgfpathcurveto{\pgfpoint{56.509153bp}{83.206289bp}}{\pgfpoint{34.733536bp}{92.321011bp}}{\pgfpoint{34.733536bp}{103.919003bp}}
\pgfclosepath
\color[rgb]{0.0,0.0,0.0}
\pgfusepath{stroke}
\end{pgfscope}
\begin{pgfscope}
\pgfsetlinewidth{1.0bp}
\pgfsetrectcap 
\pgfsetmiterjoin \pgfsetmiterlimit{10.0}
\pgfpathmoveto{\pgfpoint{183.408531bp}{34.563089bp}}
\pgfpathcurveto{\pgfpoint{183.408531bp}{31.307943bp}}{\pgfpoint{180.877424bp}{28.669124bp}}{\pgfpoint{177.755139bp}{28.669124bp}}
\pgfpathcurveto{\pgfpoint{174.632855bp}{28.669124bp}}{\pgfpoint{172.101742bp}{31.307943bp}}{\pgfpoint{172.101742bp}{34.563089bp}}
\pgfpathcurveto{\pgfpoint{172.101742bp}{37.818237bp}}{\pgfpoint{174.632855bp}{40.457054bp}}{\pgfpoint{177.755139bp}{40.457054bp}}
\pgfpathcurveto{\pgfpoint{180.877424bp}{40.457054bp}}{\pgfpoint{183.408531bp}{37.818237bp}}{\pgfpoint{183.408531bp}{34.563089bp}}
\pgfclosepath
\color[rgb]{0.0,0.0,0.0}\pgfseteorule\pgfusepath{fill}
\pgfpathmoveto{\pgfpoint{183.408531bp}{34.563089bp}}
\pgfpathcurveto{\pgfpoint{183.408531bp}{31.307943bp}}{\pgfpoint{180.877424bp}{28.669124bp}}{\pgfpoint{177.755139bp}{28.669124bp}}
\pgfpathcurveto{\pgfpoint{174.632855bp}{28.669124bp}}{\pgfpoint{172.101742bp}{31.307943bp}}{\pgfpoint{172.101742bp}{34.563089bp}}
\pgfpathcurveto{\pgfpoint{172.101742bp}{37.818237bp}}{\pgfpoint{174.632855bp}{40.457054bp}}{\pgfpoint{177.755139bp}{40.457054bp}}
\pgfpathcurveto{\pgfpoint{180.877424bp}{40.457054bp}}{\pgfpoint{183.408531bp}{37.818237bp}}{\pgfpoint{183.408531bp}{34.563089bp}}
\pgfclosepath
\color[rgb]{0.0,0.0,0.0}
\pgfusepath{stroke}
\end{pgfscope}
\begin{pgfscope}
\pgfsetlinewidth{1.0bp}
\pgfsetrectcap 
\pgfsetmiterjoin \pgfsetmiterlimit{10.0}
\pgfpathmoveto{\pgfpoint{263.808531bp}{34.163086bp}}
\pgfpathcurveto{\pgfpoint{263.808531bp}{30.90794bp}}{\pgfpoint{261.277424bp}{28.269121bp}}{\pgfpoint{258.155139bp}{28.269121bp}}
\pgfpathcurveto{\pgfpoint{255.032855bp}{28.269121bp}}{\pgfpoint{252.501742bp}{30.90794bp}}{\pgfpoint{252.501742bp}{34.163086bp}}
\pgfpathcurveto{\pgfpoint{252.501742bp}{37.418233bp}}{\pgfpoint{255.032855bp}{40.057051bp}}{\pgfpoint{258.155139bp}{40.057051bp}}
\pgfpathcurveto{\pgfpoint{261.277424bp}{40.057051bp}}{\pgfpoint{263.808531bp}{37.418233bp}}{\pgfpoint{263.808531bp}{34.163086bp}}
\pgfclosepath
\color[rgb]{0.0,0.0,0.0}\pgfseteorule\pgfusepath{fill}
\pgfpathmoveto{\pgfpoint{263.808531bp}{34.163086bp}}
\pgfpathcurveto{\pgfpoint{263.808531bp}{30.90794bp}}{\pgfpoint{261.277424bp}{28.269121bp}}{\pgfpoint{258.155139bp}{28.269121bp}}
\pgfpathcurveto{\pgfpoint{255.032855bp}{28.269121bp}}{\pgfpoint{252.501742bp}{30.90794bp}}{\pgfpoint{252.501742bp}{34.163086bp}}
\pgfpathcurveto{\pgfpoint{252.501742bp}{37.418233bp}}{\pgfpoint{255.032855bp}{40.057051bp}}{\pgfpoint{258.155139bp}{40.057051bp}}
\pgfpathcurveto{\pgfpoint{261.277424bp}{40.057051bp}}{\pgfpoint{263.808531bp}{37.418233bp}}{\pgfpoint{263.808531bp}{34.163086bp}}
\pgfclosepath
\color[rgb]{0.0,0.0,0.0}
\pgfusepath{stroke}
\end{pgfscope}
\begin{pgfscope}
\pgftransformcm{1.40625}{-0.0}{0.0}{1.40625}{\pgfpoint{197.686426bp}{32.237506bp}}
\pgftext[left,base]{\sffamily\mdseries\upshape\fontsize{36.135}{36.135}\selectfont
\color[rgb]{0.0,0.0,0.0}$\hdots$}
\end{pgfscope}
\end{pgfpicture}

%% file: hypergraph3.tex
\begin{pgfpicture}{0bp}{0bp}{359.057678bp}{213.557648bp}
\begin{pgfscope}
\pgfsetlinewidth{1.0bp}
\pgfsetrectcap 
\pgfsetmiterjoin \pgfsetmiterlimit{10.0}
\pgfpathmoveto{\pgfpoint{173.843603bp}{112.964151bp}}
\pgfpathcurveto{\pgfpoint{171.541868bp}{110.662415bp}}{\pgfpoint{167.886178bp}{110.586251bp}}{\pgfpoint{165.678389bp}{112.79404bp}}
\pgfpathcurveto{\pgfpoint{163.470601bp}{115.001828bp}}{\pgfpoint{163.546761bp}{118.657523bp}}{\pgfpoint{165.848496bp}{120.959258bp}}
\pgfpathcurveto{\pgfpoint{168.150233bp}{123.260995bp}}{\pgfpoint{171.805927bp}{123.337154bp}}{\pgfpoint{174.013715bp}{121.129366bp}}
\pgfpathcurveto{\pgfpoint{176.221503bp}{118.921577bp}}{\pgfpoint{176.14534bp}{115.265888bp}}{\pgfpoint{173.843603bp}{112.964151bp}}
\pgfclosepath
\color[rgb]{0.0,0.0,0.0}\pgfseteorule\pgfusepath{fill}
\pgfpathmoveto{\pgfpoint{173.843603bp}{112.964151bp}}
\pgfpathcurveto{\pgfpoint{171.541868bp}{110.662415bp}}{\pgfpoint{167.886178bp}{110.586251bp}}{\pgfpoint{165.678389bp}{112.79404bp}}
\pgfpathcurveto{\pgfpoint{163.470601bp}{115.001828bp}}{\pgfpoint{163.546761bp}{118.657523bp}}{\pgfpoint{165.848496bp}{120.959258bp}}
\pgfpathcurveto{\pgfpoint{168.150233bp}{123.260995bp}}{\pgfpoint{171.805927bp}{123.337154bp}}{\pgfpoint{174.013715bp}{121.129366bp}}
\pgfpathcurveto{\pgfpoint{176.221503bp}{118.921577bp}}{\pgfpoint{176.14534bp}{115.265888bp}}{\pgfpoint{173.843603bp}{112.964151bp}}
\pgfclosepath
\color[rgb]{0.0,0.0,0.0}
\pgfusepath{stroke}
\end{pgfscope}
\begin{pgfscope}
\pgfsetlinewidth{1.0bp}
\pgfsetrectcap 
\pgfsetmiterjoin \pgfsetmiterlimit{10.0}
\pgfsetdash{{10.0bp}{5.0bp}}{0.0bp}
\pgfpathmoveto{\pgfpoint{261.61835bp}{0.70709bp}}
\pgfpathlineto{\pgfpoint{358.350557bp}{97.439298bp}}
\pgfpathlineto{\pgfpoint{245.496315bp}{210.29354bp}}
\pgfpathlineto{\pgfpoint{148.764107bp}{113.561333bp}}
\pgfpathlineto{\pgfpoint{261.61835bp}{0.70709bp}}
\pgfclosepath
\color[rgb]{0.0,0.0,0.0}
\pgfusepath{stroke}
\end{pgfscope}
\begin{pgfscope}
\pgfsetlinewidth{1.0bp}
\pgfsetrectcap 
\pgfsetmiterjoin \pgfsetmiterlimit{10.0}
\pgfsetdash{{1.0bp}{5.0bp}}{0.0bp}
\pgfpathmoveto{\pgfpoint{253.424232bp}{15.462834bp}}
\pgfpathcurveto{\pgfpoint{241.82627bp}{15.462834bp}}{\pgfpoint{232.424232bp}{48.772862bp}}{\pgfpoint{232.424232bp}{89.862828bp}}
\pgfpathcurveto{\pgfpoint{232.424232bp}{130.952825bp}}{\pgfpoint{241.82627bp}{164.262837bp}}{\pgfpoint{253.424232bp}{164.262837bp}}
\pgfpathcurveto{\pgfpoint{265.022224bp}{164.262837bp}}{\pgfpoint{274.424232bp}{130.952825bp}}{\pgfpoint{274.424232bp}{89.862828bp}}
\pgfpathcurveto{\pgfpoint{274.424232bp}{48.772862bp}}{\pgfpoint{265.022224bp}{15.462834bp}}{\pgfpoint{253.424232bp}{15.462834bp}}
\pgfclosepath
\color[rgb]{0.0,0.0,0.0}
\pgfusepath{stroke}
\end{pgfscope}
\begin{pgfscope}
\pgfsetlinewidth{1.0bp}
\pgfsetrectcap 
\pgfsetmiterjoin \pgfsetmiterlimit{10.0}
\pgfsetdash{{1.0bp}{5.0bp}}{0.0bp}
\pgfpathmoveto{\pgfpoint{178.17778bp}{90.426763bp}}
\pgfpathcurveto{\pgfpoint{178.17778bp}{102.024725bp}}{\pgfpoint{211.487808bp}{111.426763bp}}{\pgfpoint{252.577774bp}{111.426763bp}}
\pgfpathcurveto{\pgfpoint{293.667771bp}{111.426763bp}}{\pgfpoint{326.977783bp}{102.024725bp}}{\pgfpoint{326.977783bp}{90.426763bp}}
\pgfpathcurveto{\pgfpoint{326.977783bp}{78.828771bp}}{\pgfpoint{293.667771bp}{69.426763bp}}{\pgfpoint{252.577774bp}{69.426763bp}}
\pgfpathcurveto{\pgfpoint{211.487808bp}{69.426763bp}}{\pgfpoint{178.17778bp}{78.828771bp}}{\pgfpoint{178.17778bp}{90.426763bp}}
\pgfclosepath
\color[rgb]{0.0,0.0,0.0}
\pgfusepath{stroke}
\end{pgfscope}
\begin{pgfscope}
\pgfsetlinewidth{1.0bp}
\pgfsetrectcap 
\pgfsetmiterjoin \pgfsetmiterlimit{10.0}
\pgfsetdash{{1.0bp}{5.0bp}}{0.0bp}
\pgfpathmoveto{\pgfpoint{157.599909bp}{108.106724bp}}
\pgfpathcurveto{\pgfpoint{151.800928bp}{118.150853bp}}{\pgfpoint{175.94724bp}{142.948272bp}}{\pgfpoint{211.532194bp}{163.493254bp}}
\pgfpathcurveto{\pgfpoint{247.117174bp}{184.038253bp}}{\pgfpoint{267.513325bp}{183.641309bp}}{\pgfpoint{273.312306bp}{173.59718bp}}
\pgfpathcurveto{\pgfpoint{279.111302bp}{163.553024bp}}{\pgfpoint{268.117174bp}{147.665186bp}}{\pgfpoint{232.532194bp}{127.120187bp}}
\pgfpathcurveto{\pgfpoint{196.94724bp}{106.575205bp}}{\pgfpoint{163.398905bp}{98.062568bp}}{\pgfpoint{157.599909bp}{108.106724bp}}
\pgfclosepath
\color[rgb]{0.0,0.0,0.0}
\pgfusepath{stroke}
\end{pgfscope}
\begin{pgfscope}
\pgfsetlinewidth{1.0bp}
\pgfsetrectcap 
\pgfsetmiterjoin \pgfsetmiterlimit{10.0}
\pgfsetdash{{1.0bp}{5.0bp}}{0.0bp}
\pgfpathmoveto{\pgfpoint{210.206457bp}{168.921963bp}}
\pgfpathcurveto{\pgfpoint{216.005438bp}{178.966093bp}}{\pgfpoint{249.553787bp}{170.453483bp}}{\pgfpoint{285.138742bp}{149.9085bp}}
\pgfpathcurveto{\pgfpoint{320.723722bp}{129.363502bp}}{\pgfpoint{344.87002bp}{104.566091bp}}{\pgfpoint{339.07104bp}{94.521962bp}}
\pgfpathcurveto{\pgfpoint{333.272044bp}{84.477806bp}}{\pgfpoint{299.723722bp}{92.990435bp}}{\pgfpoint{264.138742bp}{113.535433bp}}
\pgfpathcurveto{\pgfpoint{228.553787bp}{134.080416bp}}{\pgfpoint{204.407461bp}{158.877808bp}}{\pgfpoint{210.206457bp}{168.921963bp}}
\pgfclosepath
\color[rgb]{0.0,0.0,0.0}
\pgfusepath{stroke}
\end{pgfscope}
\begin{pgfscope}
\pgftransformcm{1.0}{0.0}{0.0}{1.0}{\pgfpoint{176.757336bp}{128.30031bp}}
\pgftext[left,base]{\sffamily\mdseries\upshape\huge
\color[rgb]{0.0,0.0,0.0}$u\equiv v_0$}
\end{pgfscope}
\begin{pgfscope}
\pgftransformcm{1.0}{0.0}{0.0}{1.0}{\pgfpoint{272.357336bp}{194.70031bp}}
\pgftext[left,base]{\sffamily\mdseries\upshape\huge
\color[rgb]{0.0,0.0,0.0}$H_{0}$}
\end{pgfscope}
\begin{pgfscope}
\pgfsetlinewidth{1.0bp}
\pgfsetrectcap 
\pgfsetmiterjoin \pgfsetmiterlimit{10.0}
\pgfpathmoveto{\pgfpoint{91.499617bp}{212.850536bp}}
\pgfpathlineto{\pgfpoint{0.707106bp}{122.058026bp}}
\pgfpathlineto{\pgfpoint{97.01505bp}{25.750082bp}}
\pgfpathlineto{\pgfpoint{187.80756bp}{116.542593bp}}
\pgfpathlineto{\pgfpoint{91.499617bp}{212.850536bp}}
\pgfclosepath
\color[rgb]{0.0,0.0,0.0}
\pgfusepath{stroke}
\end{pgfscope}
\begin{pgfscope}
\pgfsetlinewidth{1.0bp}
\pgfsetrectcap 
\pgfsetmiterjoin \pgfsetmiterlimit{10.0}
\pgfsetdash{{10.0bp}{5.0bp}}{0.0bp}
\pgfpathmoveto{\pgfpoint{114.957361bp}{161.900313bp}}
\pgfpathcurveto{\pgfpoint{114.957361bp}{146.988615bp}}{\pgfpoint{102.869043bp}{134.900313bp}}{\pgfpoint{87.957361bp}{134.900313bp}}
\pgfpathcurveto{\pgfpoint{73.045618bp}{134.900313bp}}{\pgfpoint{60.957361bp}{146.988615bp}}{\pgfpoint{60.957361bp}{161.900313bp}}
\pgfpathcurveto{\pgfpoint{60.957361bp}{176.811995bp}}{\pgfpoint{73.045618bp}{188.900313bp}}{\pgfpoint{87.957361bp}{188.900313bp}}
\pgfpathcurveto{\pgfpoint{102.869043bp}{188.900313bp}}{\pgfpoint{114.957361bp}{176.811995bp}}{\pgfpoint{114.957361bp}{161.900313bp}}
\pgfclosepath
\color[rgb]{0.0,0.0,0.0}
\pgfusepath{stroke}
\end{pgfscope}
\begin{pgfscope}
\pgftransformcm{1.0}{0.0}{0.0}{1.0}{\pgfpoint{82.757336bp}{154.50031bp}}
\pgftext[left,base]{\sffamily\mdseries\upshape\huge
\color[rgb]{0.0,0.0,0.0}$S$}
\end{pgfscope}
\begin{pgfscope}
\pgfsetlinewidth{1.0bp}
\pgfsetrectcap 
\pgfsetmiterjoin \pgfsetmiterlimit{10.0}
\pgfpathmoveto{\pgfpoint{56.8647bp}{74.379615bp}}
\pgfpathcurveto{\pgfpoint{52.897964bp}{85.278133bp}}{\pgfpoint{64.783465bp}{100.70586bp}}{\pgfpoint{103.395403bp}{114.759456bp}}
\pgfpathcurveto{\pgfpoint{142.007369bp}{128.813063bp}}{\pgfpoint{176.524229bp}{131.370732bp}}{\pgfpoint{180.490965bp}{120.472213bp}}
\pgfpathcurveto{\pgfpoint{184.457712bp}{109.573666bp}}{\pgfpoint{156.372215bp}{89.345973bp}}{\pgfpoint{117.760249bp}{75.292366bp}}
\pgfpathcurveto{\pgfpoint{79.148311bp}{61.23877bp}}{\pgfpoint{60.831447bp}{63.481067bp}}{\pgfpoint{56.8647bp}{74.379615bp}}
\pgfclosepath
\color[rgb]{0.0,0.0,0.0}
\pgfusepath{stroke}
\end{pgfscope}
\begin{pgfscope}
\pgfsetlinewidth{1.0bp}
\pgfsetrectcap 
\pgfsetmiterjoin \pgfsetmiterlimit{10.0}
\pgfpathmoveto{\pgfpoint{91.201483bp}{175.781875bp}}
\pgfpathcurveto{\pgfpoint{102.799444bp}{175.781875bp}}{\pgfpoint{113.050011bp}{161.157128bp}}{\pgfpoint{113.050011bp}{120.067162bp}}
\pgfpathcurveto{\pgfpoint{113.050011bp}{78.977166bp}}{\pgfpoint{100.999444bp}{41.243603bp}}{\pgfpoint{89.401483bp}{41.243603bp}}
\pgfpathcurveto{\pgfpoint{77.803491bp}{41.243603bp}}{\pgfpoint{71.050011bp}{78.977166bp}}{\pgfpoint{71.050011bp}{120.067162bp}}
\pgfpathcurveto{\pgfpoint{71.050011bp}{161.157128bp}}{\pgfpoint{79.603491bp}{175.781875bp}}{\pgfpoint{91.201483bp}{175.781875bp}}
\pgfclosepath
\color[rgb]{0.0,0.0,0.0}
\pgfusepath{stroke}
\end{pgfscope}
\begin{pgfscope}
\pgfsetlinewidth{1.0bp}
\pgfsetrectcap 
\pgfsetmiterjoin \pgfsetmiterlimit{10.0}
\pgfpathmoveto{\pgfpoint{104.472669bp}{179.63225bp}}
\pgfpathcurveto{\pgfpoint{112.673666bp}{171.431253bp}}{\pgfpoint{109.580654bp}{153.84175bp}}{\pgfpoint{80.52566bp}{124.786757bp}}
\pgfpathcurveto{\pgfpoint{51.470645bp}{95.731742bp}}{\pgfpoint{35.359833bp}{96.663004bp}}{\pgfpoint{27.158836bp}{104.864001bp}}
\pgfpathcurveto{\pgfpoint{18.957817bp}{113.06502bp}}{\pgfpoint{21.77216bp}{125.430226bp}}{\pgfpoint{50.827175bp}{154.485241bp}}
\pgfpathcurveto{\pgfpoint{79.882169bp}{183.540235bp}}{\pgfpoint{96.27165bp}{187.833269bp}}{\pgfpoint{104.472669bp}{179.63225bp}}
\pgfclosepath
\color[rgb]{0.0,0.0,0.0}
\pgfusepath{stroke}
\end{pgfscope}
\begin{pgfscope}
\pgftransformcm{1.0}{0.0}{0.0}{1.0}{\pgfpoint{48.157336bp}{195.70031bp}}
\pgftext[left,base]{\sffamily\mdseries\upshape\huge
\color[rgb]{0.0,0.0,0.0}$H$}
\end{pgfscope}
\end{pgfpicture}

%% file: hypergraph4.tex
\begin{pgfpicture}{0bp}{0bp}{356.99375bp}{183.403076bp}
\begin{pgfscope}
\pgfsetlinewidth{1.0bp}
\pgfsetrectcap 
\pgfsetmiterjoin \pgfsetmiterlimit{10.0}
\pgfpathmoveto{\pgfpoint{145.552544bp}{37.6929bp}}
\pgfpathcurveto{\pgfpoint{148.80769bp}{37.6929bp}}{\pgfpoint{151.446509bp}{35.161793bp}}{\pgfpoint{151.446509bp}{32.039508bp}}
\pgfpathcurveto{\pgfpoint{151.446509bp}{28.917224bp}}{\pgfpoint{148.80769bp}{26.386111bp}}{\pgfpoint{145.552544bp}{26.386111bp}}
\pgfpathcurveto{\pgfpoint{142.297396bp}{26.386111bp}}{\pgfpoint{139.658578bp}{28.917224bp}}{\pgfpoint{139.658578bp}{32.039508bp}}
\pgfpathcurveto{\pgfpoint{139.658578bp}{35.161793bp}}{\pgfpoint{142.297396bp}{37.6929bp}}{\pgfpoint{145.552544bp}{37.6929bp}}
\pgfclosepath
\color[rgb]{0.0,0.0,0.0}\pgfseteorule\pgfusepath{fill}
\pgfpathmoveto{\pgfpoint{145.552544bp}{37.6929bp}}
\pgfpathcurveto{\pgfpoint{148.80769bp}{37.6929bp}}{\pgfpoint{151.446509bp}{35.161793bp}}{\pgfpoint{151.446509bp}{32.039508bp}}
\pgfpathcurveto{\pgfpoint{151.446509bp}{28.917224bp}}{\pgfpoint{148.80769bp}{26.386111bp}}{\pgfpoint{145.552544bp}{26.386111bp}}
\pgfpathcurveto{\pgfpoint{142.297396bp}{26.386111bp}}{\pgfpoint{139.658578bp}{28.917224bp}}{\pgfpoint{139.658578bp}{32.039508bp}}
\pgfpathcurveto{\pgfpoint{139.658578bp}{35.161793bp}}{\pgfpoint{142.297396bp}{37.6929bp}}{\pgfpoint{145.552544bp}{37.6929bp}}
\pgfclosepath
\color[rgb]{0.0,0.0,0.0}
\pgfusepath{stroke}
\end{pgfscope}
\begin{pgfscope}
\pgfsetlinewidth{1.0bp}
\pgfsetrectcap 
\pgfsetmiterjoin \pgfsetmiterlimit{10.0}
\pgfpathmoveto{\pgfpoint{145.152544bp}{74.4929bp}}
\pgfpathcurveto{\pgfpoint{148.40769bp}{74.4929bp}}{\pgfpoint{151.046509bp}{71.961793bp}}{\pgfpoint{151.046509bp}{68.839508bp}}
\pgfpathcurveto{\pgfpoint{151.046509bp}{65.717224bp}}{\pgfpoint{148.40769bp}{63.186111bp}}{\pgfpoint{145.152544bp}{63.186111bp}}
\pgfpathcurveto{\pgfpoint{141.897396bp}{63.186111bp}}{\pgfpoint{139.258578bp}{65.717224bp}}{\pgfpoint{139.258578bp}{68.839508bp}}
\pgfpathcurveto{\pgfpoint{139.258578bp}{71.961793bp}}{\pgfpoint{141.897396bp}{74.4929bp}}{\pgfpoint{145.152544bp}{74.4929bp}}
\pgfclosepath
\color[rgb]{0.0,0.0,0.0}\pgfseteorule\pgfusepath{fill}
\pgfpathmoveto{\pgfpoint{145.152544bp}{74.4929bp}}
\pgfpathcurveto{\pgfpoint{148.40769bp}{74.4929bp}}{\pgfpoint{151.046509bp}{71.961793bp}}{\pgfpoint{151.046509bp}{68.839508bp}}
\pgfpathcurveto{\pgfpoint{151.046509bp}{65.717224bp}}{\pgfpoint{148.40769bp}{63.186111bp}}{\pgfpoint{145.152544bp}{63.186111bp}}
\pgfpathcurveto{\pgfpoint{141.897396bp}{63.186111bp}}{\pgfpoint{139.258578bp}{65.717224bp}}{\pgfpoint{139.258578bp}{68.839508bp}}
\pgfpathcurveto{\pgfpoint{139.258578bp}{71.961793bp}}{\pgfpoint{141.897396bp}{74.4929bp}}{\pgfpoint{145.152544bp}{74.4929bp}}
\pgfclosepath
\color[rgb]{0.0,0.0,0.0}
\pgfusepath{stroke}
\end{pgfscope}
\begin{pgfscope}
\pgfsetlinewidth{1.0bp}
\pgfsetrectcap 
\pgfsetmiterjoin \pgfsetmiterlimit{10.0}
\pgfpathmoveto{\pgfpoint{145.552547bp}{154.8929bp}}
\pgfpathcurveto{\pgfpoint{148.807693bp}{154.8929bp}}{\pgfpoint{151.446512bp}{152.361793bp}}{\pgfpoint{151.446512bp}{149.239508bp}}
\pgfpathcurveto{\pgfpoint{151.446512bp}{146.117224bp}}{\pgfpoint{148.807693bp}{143.586111bp}}{\pgfpoint{145.552547bp}{143.586111bp}}
\pgfpathcurveto{\pgfpoint{142.297399bp}{143.586111bp}}{\pgfpoint{139.658581bp}{146.117224bp}}{\pgfpoint{139.658581bp}{149.239508bp}}
\pgfpathcurveto{\pgfpoint{139.658581bp}{152.361793bp}}{\pgfpoint{142.297399bp}{154.8929bp}}{\pgfpoint{145.552547bp}{154.8929bp}}
\pgfclosepath
\color[rgb]{0.0,0.0,0.0}\pgfseteorule\pgfusepath{fill}
\pgfpathmoveto{\pgfpoint{145.552547bp}{154.8929bp}}
\pgfpathcurveto{\pgfpoint{148.807693bp}{154.8929bp}}{\pgfpoint{151.446512bp}{152.361793bp}}{\pgfpoint{151.446512bp}{149.239508bp}}
\pgfpathcurveto{\pgfpoint{151.446512bp}{146.117224bp}}{\pgfpoint{148.807693bp}{143.586111bp}}{\pgfpoint{145.552547bp}{143.586111bp}}
\pgfpathcurveto{\pgfpoint{142.297399bp}{143.586111bp}}{\pgfpoint{139.658581bp}{146.117224bp}}{\pgfpoint{139.658581bp}{149.239508bp}}
\pgfpathcurveto{\pgfpoint{139.658581bp}{152.361793bp}}{\pgfpoint{142.297399bp}{154.8929bp}}{\pgfpoint{145.552547bp}{154.8929bp}}
\pgfclosepath
\color[rgb]{0.0,0.0,0.0}
\pgfusepath{stroke}
\end{pgfscope}
\begin{pgfscope}
\pgfsetlinewidth{1.0bp}
\pgfsetrectcap 
\pgfsetmiterjoin \pgfsetmiterlimit{10.0}
\pgfpathmoveto{\pgfpoint{170.415627bp}{171.499994bp}}
\pgfpathlineto{\pgfpoint{33.615627bp}{171.499994bp}}
\pgfpathlineto{\pgfpoint{33.615627bp}{11.899994bp}}
\pgfpathlineto{\pgfpoint{170.415627bp}{11.899994bp}}
\pgfpathlineto{\pgfpoint{170.415627bp}{171.499994bp}}
\pgfclosepath
\color[rgb]{0.0,0.0,0.0}
\pgfusepath{stroke}
\end{pgfscope}
\begin{pgfscope}
\pgfsetlinewidth{1.0bp}
\pgfsetrectcap 
\pgfsetmiterjoin \pgfsetmiterlimit{10.0}
\pgfpathmoveto{\pgfpoint{166.624384bp}{162.908744bp}}
\pgfpathcurveto{\pgfpoint{174.825381bp}{154.707747bp}}{\pgfpoint{157.91988bp}{124.505755bp}}{\pgfpoint{128.864887bp}{95.450761bp}}
\pgfpathcurveto{\pgfpoint{99.809871bp}{66.395746bp}}{\pgfpoint{69.60789bp}{49.490256bp}}{\pgfpoint{61.406893bp}{57.691253bp}}
\pgfpathcurveto{\pgfpoint{53.205874bp}{65.892272bp}}{\pgfpoint{70.111387bp}{96.094231bp}}{\pgfpoint{99.166402bp}{125.149246bp}}
\pgfpathcurveto{\pgfpoint{128.221395bp}{154.20424bp}}{\pgfpoint{158.423366bp}{171.109763bp}}{\pgfpoint{166.624384bp}{162.908744bp}}
\pgfclosepath
\color[rgb]{0.0,0.0,0.0}
\pgfusepath{stroke}
\end{pgfscope}
\begin{pgfscope}
\pgfsetlinewidth{1.0bp}
\pgfsetrectcap 
\pgfsetmiterjoin \pgfsetmiterlimit{10.0}
\pgfpathmoveto{\pgfpoint{166.824158bp}{56.693965bp}}
\pgfpathcurveto{\pgfpoint{158.623161bp}{48.492968bp}}{\pgfpoint{128.421169bp}{65.39847bp}}{\pgfpoint{99.366175bp}{94.453463bp}}
\pgfpathcurveto{\pgfpoint{70.31116bp}{123.508478bp}}{\pgfpoint{53.405669bp}{153.710459bp}}{\pgfpoint{61.606667bp}{161.911457bp}}
\pgfpathcurveto{\pgfpoint{69.807685bp}{170.112475bp}}{\pgfpoint{100.009645bp}{153.206963bp}}{\pgfpoint{129.06466bp}{124.151948bp}}
\pgfpathcurveto{\pgfpoint{158.119653bp}{95.096955bp}}{\pgfpoint{175.025176bp}{64.894984bp}}{\pgfpoint{166.824158bp}{56.693965bp}}
\pgfclosepath
\color[rgb]{0.0,0.0,0.0}
\pgfusepath{stroke}
\end{pgfscope}
\begin{pgfscope}
\pgfsetlinewidth{1.0bp}
\pgfsetrectcap 
\pgfsetmiterjoin \pgfsetmiterlimit{10.0}
\pgfpathmoveto{\pgfpoint{168.87329bp}{29.641593bp}}
\pgfpathcurveto{\pgfpoint{165.871517bp}{18.438823bp}}{\pgfpoint{131.263074bp}{17.97842bp}}{\pgfpoint{91.573215bp}{28.613286bp}}
\pgfpathcurveto{\pgfpoint{51.883326bp}{39.24816bp}}{\pgfpoint{37.741751bp}{53.951097bp}}{\pgfpoint{40.743524bp}{65.153868bp}}
\pgfpathcurveto{\pgfpoint{43.745305bp}{76.356668bp}}{\pgfpoint{62.753726bp}{79.817044bp}}{\pgfpoint{102.443614bp}{69.182171bp}}
\pgfpathcurveto{\pgfpoint{142.133474bp}{58.547305bp}}{\pgfpoint{171.875071bp}{40.844393bp}}{\pgfpoint{168.87329bp}{29.641593bp}}
\pgfclosepath
\color[rgb]{0.0,0.0,0.0}
\pgfusepath{stroke}
\end{pgfscope}
\begin{pgfscope}
\pgftransformcm{-0.0}{1.40625}{-1.40625}{-0.0}{\pgfpoint{147.478126bp}{88.770795bp}}
\pgftext[left,base]{\sffamily\mdseries\upshape\fontsize{36.135}{36.135}\selectfont
\color[rgb]{0.0,0.0,0.0}$\hdots$}
\end{pgfscope}
\begin{pgfscope}
\pgftransformcm{1.0}{0.0}{0.0}{1.0}{\pgfpoint{150.615625bp}{148.699998bp}}
\pgftext[left,base]{\sffamily\mdseries\upshape\huge
\color[rgb]{0.0,0.0,0.0}$v_1$}
\end{pgfscope}
\begin{pgfscope}
\pgftransformcm{1.0}{0.0}{0.0}{1.0}{\pgfpoint{146.615625bp}{74.699998bp}}
\pgftext[left,base]{\sffamily\mdseries\upshape\huge
\color[rgb]{0.0,0.0,0.0}$v_2$}
\end{pgfscope}
\begin{pgfscope}
\pgftransformcm{1.0}{0.0}{0.0}{1.0}{\pgfpoint{150.815625bp}{28.499998bp}}
\pgftext[left,base]{\sffamily\mdseries\upshape\huge
\color[rgb]{0.0,0.0,0.0}$v_{r}$}
\end{pgfscope}
\begin{pgfscope}
\pgftransformcm{1.0}{0.0}{0.0}{1.0}{\pgfpoint{-1.584375bp}{89.099998bp}}
\pgftext[left,base]{\sffamily\mdseries\upshape\huge
\color[rgb]{0.0,0.0,0.0}$H$}
\end{pgfscope}
\begin{pgfscope}
\pgfsetlinewidth{1.0bp}
\pgfsetrectcap 
\pgfsetmiterjoin \pgfsetmiterlimit{10.0}
\pgfsetdash{{10.0bp}{5.0bp}}{0.0bp}
\pgfpathmoveto{\pgfpoint{132.015625bp}{182.899998bp}}
\pgfpathlineto{\pgfpoint{132.015625bp}{0.499998bp}}
\pgfpathlineto{\pgfpoint{305.415625bp}{0.499998bp}}
\pgfpathlineto{\pgfpoint{305.415625bp}{181.699998bp}}
\pgfpathlineto{\pgfpoint{132.015625bp}{182.899998bp}}
\pgfclosepath
\color[rgb]{0.0,0.0,0.0}
\pgfusepath{stroke}
\end{pgfscope}
\begin{pgfscope}
\pgfsetlinewidth{1.0bp}
\pgfsetrectcap 
\pgfsetmiterjoin \pgfsetmiterlimit{10.0}
\pgfsetdash{{10.0bp}{5.0bp}}{0.0bp}
\pgfpathmoveto{\pgfpoint{135.615827bp}{151.302052bp}}
\pgfpathcurveto{\pgfpoint{135.615827bp}{162.900014bp}}{\pgfpoint{168.925855bp}{172.302052bp}}{\pgfpoint{210.015821bp}{172.302052bp}}
\pgfpathcurveto{\pgfpoint{251.105817bp}{172.302052bp}}{\pgfpoint{284.41583bp}{162.900014bp}}{\pgfpoint{284.41583bp}{151.302052bp}}
\pgfpathcurveto{\pgfpoint{284.41583bp}{139.70406bp}}{\pgfpoint{251.105817bp}{130.302052bp}}{\pgfpoint{210.015821bp}{130.302052bp}}
\pgfpathcurveto{\pgfpoint{168.925855bp}{130.302052bp}}{\pgfpoint{135.615827bp}{139.70406bp}}{\pgfpoint{135.615827bp}{151.302052bp}}
\pgfclosepath
\color[rgb]{0.0,0.0,0.0}
\pgfusepath{stroke}
\end{pgfscope}
\begin{pgfscope}
\pgfsetlinewidth{1.0bp}
\pgfsetrectcap 
\pgfsetmiterjoin \pgfsetmiterlimit{10.0}
\pgfsetdash{{10.0bp}{5.0bp}}{0.0bp}
\pgfpathmoveto{\pgfpoint{136.007092bp}{57.693292bp}}
\pgfpathcurveto{\pgfpoint{127.806095bp}{65.894289bp}}{\pgfpoint{144.711596bp}{96.096281bp}}{\pgfpoint{173.76659bp}{125.151275bp}}
\pgfpathcurveto{\pgfpoint{202.821605bp}{154.20629bp}}{\pgfpoint{233.023586bp}{171.11178bp}}{\pgfpoint{241.224583bp}{162.910783bp}}
\pgfpathcurveto{\pgfpoint{249.425602bp}{154.709765bp}}{\pgfpoint{232.520089bp}{124.507805bp}}{\pgfpoint{203.465074bp}{95.45279bp}}
\pgfpathcurveto{\pgfpoint{174.410081bp}{66.397796bp}}{\pgfpoint{144.20811bp}{49.492273bp}}{\pgfpoint{136.007092bp}{57.693292bp}}
\pgfclosepath
\color[rgb]{0.0,0.0,0.0}
\pgfusepath{stroke}
\end{pgfscope}
\begin{pgfscope}
\pgfsetlinewidth{1.0bp}
\pgfsetrectcap 
\pgfsetmiterjoin \pgfsetmiterlimit{10.0}
\pgfsetdash{{10.0bp}{5.0bp}}{0.0bp}
\pgfpathmoveto{\pgfpoint{134.615827bp}{32.302052bp}}
\pgfpathcurveto{\pgfpoint{134.615827bp}{43.900014bp}}{\pgfpoint{167.925855bp}{53.302052bp}}{\pgfpoint{209.015821bp}{53.302052bp}}
\pgfpathcurveto{\pgfpoint{250.105817bp}{53.302052bp}}{\pgfpoint{283.41583bp}{43.900014bp}}{\pgfpoint{283.41583bp}{32.302052bp}}
\pgfpathcurveto{\pgfpoint{283.41583bp}{20.70406bp}}{\pgfpoint{250.105817bp}{11.302052bp}}{\pgfpoint{209.015821bp}{11.302052bp}}
\pgfpathcurveto{\pgfpoint{167.925855bp}{11.302052bp}}{\pgfpoint{134.615827bp}{20.70406bp}}{\pgfpoint{134.615827bp}{32.302052bp}}
\pgfclosepath
\color[rgb]{0.0,0.0,0.0}
\pgfusepath{stroke}
\end{pgfscope}
\begin{pgfscope}
\pgfsetlinewidth{1.0bp}
\pgfsetrectcap 
\pgfsetmiterjoin \pgfsetmiterlimit{10.0}
\pgfsetdash{{10.0bp}{5.0bp}}{0.0bp}
\pgfpathmoveto{\pgfpoint{265.41583bp}{166.702048bp}}
\pgfpathcurveto{\pgfpoint{277.013791bp}{166.702048bp}}{\pgfpoint{286.41583bp}{133.39202bp}}{\pgfpoint{286.41583bp}{92.302055bp}}
\pgfpathcurveto{\pgfpoint{286.41583bp}{51.212058bp}}{\pgfpoint{277.013791bp}{17.902045bp}}{\pgfpoint{265.41583bp}{17.902045bp}}
\pgfpathcurveto{\pgfpoint{253.817838bp}{17.902045bp}}{\pgfpoint{244.41583bp}{51.212058bp}}{\pgfpoint{244.41583bp}{92.302055bp}}
\pgfpathcurveto{\pgfpoint{244.41583bp}{133.39202bp}}{\pgfpoint{253.817838bp}{166.702048bp}}{\pgfpoint{265.41583bp}{166.702048bp}}
\pgfclosepath
\color[rgb]{0.0,0.0,0.0}
\pgfusepath{stroke}
\end{pgfscope}
\begin{pgfscope}
\pgftransformcm{1.0}{0.0}{0.0}{1.0}{\pgfpoint{335.215625bp}{87.099998bp}}
\pgftext[left,base]{\sffamily\mdseries\upshape\huge
\color[rgb]{0.0,0.0,0.0}$H_{1}$}
\end{pgfscope}
\begin{pgfscope}
\pgfsetlinewidth{1.0bp}
\pgfsetrectcap 
\pgfsetmiterjoin \pgfsetmiterlimit{10.0}
\pgfsetdash{{10.0bp}{5.0bp}}{0.0bp}
\pgfpathmoveto{\pgfpoint{107.015649bp}{105.900002bp}}
\pgfpathcurveto{\pgfpoint{107.015649bp}{90.988304bp}}{\pgfpoint{94.927332bp}{78.900002bp}}{\pgfpoint{80.015649bp}{78.900002bp}}
\pgfpathcurveto{\pgfpoint{65.103906bp}{78.900002bp}}{\pgfpoint{53.015649bp}{90.988304bp}}{\pgfpoint{53.015649bp}{105.900002bp}}
\pgfpathcurveto{\pgfpoint{53.015649bp}{120.811684bp}}{\pgfpoint{65.103906bp}{132.900002bp}}{\pgfpoint{80.015649bp}{132.900002bp}}
\pgfpathcurveto{\pgfpoint{94.927332bp}{132.900002bp}}{\pgfpoint{107.015649bp}{120.811684bp}}{\pgfpoint{107.015649bp}{105.900002bp}}
\pgfclosepath
\color[rgb]{0.0,0.0,0.0}
\pgfusepath{stroke}
\end{pgfscope}
\begin{pgfscope}
\pgftransformcm{1.0}{0.0}{0.0}{1.0}{\pgfpoint{63.415625bp}{100.899998bp}}
\pgftext[left,base]{\sffamily\mdseries\upshape\huge
\color[rgb]{0.0,0.0,0.0}$S$}
\end{pgfscope}
\end{pgfpicture}

%% file: hypergraph5.tex
\begin{pgfpicture}{0bp}{0bp}{707.800049bp}{224.199997bp}
\begin{pgfscope}
\pgfsetlinewidth{1.0bp}
\pgfsetrectcap 
\pgfsetmiterjoin \pgfsetmiterlimit{10.0}
\pgfpathmoveto{\pgfpoint{71.492912bp}{106.163064bp}}
\pgfpathcurveto{\pgfpoint{71.492912bp}{102.907918bp}}{\pgfpoint{68.961805bp}{100.269099bp}}{\pgfpoint{65.839521bp}{100.269099bp}}
\pgfpathcurveto{\pgfpoint{62.717236bp}{100.269099bp}}{\pgfpoint{60.186123bp}{102.907918bp}}{\pgfpoint{60.186123bp}{106.163064bp}}
\pgfpathcurveto{\pgfpoint{60.186123bp}{109.418212bp}}{\pgfpoint{62.717236bp}{112.05703bp}}{\pgfpoint{65.839521bp}{112.05703bp}}
\pgfpathcurveto{\pgfpoint{68.961805bp}{112.05703bp}}{\pgfpoint{71.492912bp}{109.418212bp}}{\pgfpoint{71.492912bp}{106.163064bp}}
\pgfclosepath
\color[rgb]{0.0,0.0,0.0}\pgfseteorule\pgfusepath{fill}
\pgfpathmoveto{\pgfpoint{71.492912bp}{106.163064bp}}
\pgfpathcurveto{\pgfpoint{71.492912bp}{102.907918bp}}{\pgfpoint{68.961805bp}{100.269099bp}}{\pgfpoint{65.839521bp}{100.269099bp}}
\pgfpathcurveto{\pgfpoint{62.717236bp}{100.269099bp}}{\pgfpoint{60.186123bp}{102.907918bp}}{\pgfpoint{60.186123bp}{106.163064bp}}
\pgfpathcurveto{\pgfpoint{60.186123bp}{109.418212bp}}{\pgfpoint{62.717236bp}{112.05703bp}}{\pgfpoint{65.839521bp}{112.05703bp}}
\pgfpathcurveto{\pgfpoint{68.961805bp}{112.05703bp}}{\pgfpoint{71.492912bp}{109.418212bp}}{\pgfpoint{71.492912bp}{106.163064bp}}
\pgfclosepath
\color[rgb]{0.0,0.0,0.0}
\pgfusepath{stroke}
\end{pgfscope}
\begin{pgfscope}
\pgfsetlinewidth{1.0bp}
\pgfsetrectcap 
\pgfsetmiterjoin \pgfsetmiterlimit{10.0}
\pgfpathmoveto{\pgfpoint{151.892912bp}{106.163086bp}}
\pgfpathcurveto{\pgfpoint{151.892912bp}{102.90794bp}}{\pgfpoint{149.361805bp}{100.269121bp}}{\pgfpoint{146.239521bp}{100.269121bp}}
\pgfpathcurveto{\pgfpoint{143.117236bp}{100.269121bp}}{\pgfpoint{140.586123bp}{102.90794bp}}{\pgfpoint{140.586123bp}{106.163086bp}}
\pgfpathcurveto{\pgfpoint{140.586123bp}{109.418233bp}}{\pgfpoint{143.117236bp}{112.057051bp}}{\pgfpoint{146.239521bp}{112.057051bp}}
\pgfpathcurveto{\pgfpoint{149.361805bp}{112.057051bp}}{\pgfpoint{151.892912bp}{109.418233bp}}{\pgfpoint{151.892912bp}{106.163086bp}}
\pgfclosepath
\color[rgb]{0.0,0.0,0.0}\pgfseteorule\pgfusepath{fill}
\pgfpathmoveto{\pgfpoint{151.892912bp}{106.163086bp}}
\pgfpathcurveto{\pgfpoint{151.892912bp}{102.90794bp}}{\pgfpoint{149.361805bp}{100.269121bp}}{\pgfpoint{146.239521bp}{100.269121bp}}
\pgfpathcurveto{\pgfpoint{143.117236bp}{100.269121bp}}{\pgfpoint{140.586123bp}{102.90794bp}}{\pgfpoint{140.586123bp}{106.163086bp}}
\pgfpathcurveto{\pgfpoint{140.586123bp}{109.418233bp}}{\pgfpoint{143.117236bp}{112.057051bp}}{\pgfpoint{146.239521bp}{112.057051bp}}
\pgfpathcurveto{\pgfpoint{149.361805bp}{112.057051bp}}{\pgfpoint{151.892912bp}{109.418233bp}}{\pgfpoint{151.892912bp}{106.163086bp}}
\pgfclosepath
\color[rgb]{0.0,0.0,0.0}
\pgfusepath{stroke}
\end{pgfscope}
\begin{pgfscope}
\pgfsetlinewidth{1.0bp}
\pgfsetrectcap 
\pgfsetmiterjoin \pgfsetmiterlimit{10.0}
\pgfpathmoveto{\pgfpoint{231.892912bp}{106.763071bp}}
\pgfpathcurveto{\pgfpoint{231.892912bp}{103.507924bp}}{\pgfpoint{229.361805bp}{100.869105bp}}{\pgfpoint{226.239521bp}{100.869105bp}}
\pgfpathcurveto{\pgfpoint{223.117236bp}{100.869105bp}}{\pgfpoint{220.586123bp}{103.507924bp}}{\pgfpoint{220.586123bp}{106.763071bp}}
\pgfpathcurveto{\pgfpoint{220.586123bp}{110.018218bp}}{\pgfpoint{223.117236bp}{112.657036bp}}{\pgfpoint{226.239521bp}{112.657036bp}}
\pgfpathcurveto{\pgfpoint{229.361805bp}{112.657036bp}}{\pgfpoint{231.892912bp}{110.018218bp}}{\pgfpoint{231.892912bp}{106.763071bp}}
\pgfclosepath
\color[rgb]{0.0,0.0,0.0}\pgfseteorule\pgfusepath{fill}
\pgfpathmoveto{\pgfpoint{231.892912bp}{106.763071bp}}
\pgfpathcurveto{\pgfpoint{231.892912bp}{103.507924bp}}{\pgfpoint{229.361805bp}{100.869105bp}}{\pgfpoint{226.239521bp}{100.869105bp}}
\pgfpathcurveto{\pgfpoint{223.117236bp}{100.869105bp}}{\pgfpoint{220.586123bp}{103.507924bp}}{\pgfpoint{220.586123bp}{106.763071bp}}
\pgfpathcurveto{\pgfpoint{220.586123bp}{110.018218bp}}{\pgfpoint{223.117236bp}{112.657036bp}}{\pgfpoint{226.239521bp}{112.657036bp}}
\pgfpathcurveto{\pgfpoint{229.361805bp}{112.657036bp}}{\pgfpoint{231.892912bp}{110.018218bp}}{\pgfpoint{231.892912bp}{106.763071bp}}
\pgfclosepath
\color[rgb]{0.0,0.0,0.0}
\pgfusepath{stroke}
\end{pgfscope}
\begin{pgfscope}
\pgfsetlinewidth{1.0bp}
\pgfsetrectcap 
\pgfsetmiterjoin \pgfsetmiterlimit{10.0}
\pgfpathmoveto{\pgfpoint{312.292912bp}{106.763092bp}}
\pgfpathcurveto{\pgfpoint{312.292912bp}{103.507946bp}}{\pgfpoint{309.761805bp}{100.869127bp}}{\pgfpoint{306.639521bp}{100.869127bp}}
\pgfpathcurveto{\pgfpoint{303.517236bp}{100.869127bp}}{\pgfpoint{300.986123bp}{103.507946bp}}{\pgfpoint{300.986123bp}{106.763092bp}}
\pgfpathcurveto{\pgfpoint{300.986123bp}{110.01824bp}}{\pgfpoint{303.517236bp}{112.657057bp}}{\pgfpoint{306.639521bp}{112.657057bp}}
\pgfpathcurveto{\pgfpoint{309.761805bp}{112.657057bp}}{\pgfpoint{312.292912bp}{110.01824bp}}{\pgfpoint{312.292912bp}{106.763092bp}}
\pgfclosepath
\color[rgb]{0.0,0.0,0.0}\pgfseteorule\pgfusepath{fill}
\pgfpathmoveto{\pgfpoint{312.292912bp}{106.763092bp}}
\pgfpathcurveto{\pgfpoint{312.292912bp}{103.507946bp}}{\pgfpoint{309.761805bp}{100.869127bp}}{\pgfpoint{306.639521bp}{100.869127bp}}
\pgfpathcurveto{\pgfpoint{303.517236bp}{100.869127bp}}{\pgfpoint{300.986123bp}{103.507946bp}}{\pgfpoint{300.986123bp}{106.763092bp}}
\pgfpathcurveto{\pgfpoint{300.986123bp}{110.01824bp}}{\pgfpoint{303.517236bp}{112.657057bp}}{\pgfpoint{306.639521bp}{112.657057bp}}
\pgfpathcurveto{\pgfpoint{309.761805bp}{112.657057bp}}{\pgfpoint{312.292912bp}{110.01824bp}}{\pgfpoint{312.292912bp}{106.763092bp}}
\pgfclosepath
\color[rgb]{0.0,0.0,0.0}
\pgfusepath{stroke}
\end{pgfscope}
\begin{pgfscope}
\pgfsetlinewidth{1.0bp}
\pgfsetrectcap 
\pgfsetmiterjoin \pgfsetmiterlimit{10.0}
\pgfpathmoveto{\pgfpoint{395.292912bp}{106.763083bp}}
\pgfpathcurveto{\pgfpoint{395.292912bp}{103.507937bp}}{\pgfpoint{392.761805bp}{100.869117bp}}{\pgfpoint{389.639521bp}{100.869117bp}}
\pgfpathcurveto{\pgfpoint{386.517236bp}{100.869117bp}}{\pgfpoint{383.986123bp}{103.507937bp}}{\pgfpoint{383.986123bp}{106.763083bp}}
\pgfpathcurveto{\pgfpoint{383.986123bp}{110.01823bp}}{\pgfpoint{386.517236bp}{112.657048bp}}{\pgfpoint{389.639521bp}{112.657048bp}}
\pgfpathcurveto{\pgfpoint{392.761805bp}{112.657048bp}}{\pgfpoint{395.292912bp}{110.01823bp}}{\pgfpoint{395.292912bp}{106.763083bp}}
\pgfclosepath
\color[rgb]{0.0,0.0,0.0}\pgfseteorule\pgfusepath{fill}
\pgfpathmoveto{\pgfpoint{395.292912bp}{106.763083bp}}
\pgfpathcurveto{\pgfpoint{395.292912bp}{103.507937bp}}{\pgfpoint{392.761805bp}{100.869117bp}}{\pgfpoint{389.639521bp}{100.869117bp}}
\pgfpathcurveto{\pgfpoint{386.517236bp}{100.869117bp}}{\pgfpoint{383.986123bp}{103.507937bp}}{\pgfpoint{383.986123bp}{106.763083bp}}
\pgfpathcurveto{\pgfpoint{383.986123bp}{110.01823bp}}{\pgfpoint{386.517236bp}{112.657048bp}}{\pgfpoint{389.639521bp}{112.657048bp}}
\pgfpathcurveto{\pgfpoint{392.761805bp}{112.657048bp}}{\pgfpoint{395.292912bp}{110.01823bp}}{\pgfpoint{395.292912bp}{106.763083bp}}
\pgfclosepath
\color[rgb]{0.0,0.0,0.0}
\pgfusepath{stroke}
\end{pgfscope}
\begin{pgfscope}
\pgfsetlinewidth{1.0bp}
\pgfsetrectcap 
\pgfsetmiterjoin \pgfsetmiterlimit{10.0}
\pgfpathmoveto{\pgfpoint{475.692912bp}{106.763074bp}}
\pgfpathcurveto{\pgfpoint{475.692912bp}{103.507927bp}}{\pgfpoint{473.161805bp}{100.869108bp}}{\pgfpoint{470.039521bp}{100.869108bp}}
\pgfpathcurveto{\pgfpoint{466.917236bp}{100.869108bp}}{\pgfpoint{464.386123bp}{103.507927bp}}{\pgfpoint{464.386123bp}{106.763074bp}}
\pgfpathcurveto{\pgfpoint{464.386123bp}{110.018221bp}}{\pgfpoint{466.917236bp}{112.657039bp}}{\pgfpoint{470.039521bp}{112.657039bp}}
\pgfpathcurveto{\pgfpoint{473.161805bp}{112.657039bp}}{\pgfpoint{475.692912bp}{110.018221bp}}{\pgfpoint{475.692912bp}{106.763074bp}}
\pgfclosepath
\color[rgb]{0.0,0.0,0.0}\pgfseteorule\pgfusepath{fill}
\pgfpathmoveto{\pgfpoint{475.692912bp}{106.763074bp}}
\pgfpathcurveto{\pgfpoint{475.692912bp}{103.507927bp}}{\pgfpoint{473.161805bp}{100.869108bp}}{\pgfpoint{470.039521bp}{100.869108bp}}
\pgfpathcurveto{\pgfpoint{466.917236bp}{100.869108bp}}{\pgfpoint{464.386123bp}{103.507927bp}}{\pgfpoint{464.386123bp}{106.763074bp}}
\pgfpathcurveto{\pgfpoint{464.386123bp}{110.018221bp}}{\pgfpoint{466.917236bp}{112.657039bp}}{\pgfpoint{470.039521bp}{112.657039bp}}
\pgfpathcurveto{\pgfpoint{473.161805bp}{112.657039bp}}{\pgfpoint{475.692912bp}{110.018221bp}}{\pgfpoint{475.692912bp}{106.763074bp}}
\pgfclosepath
\color[rgb]{0.0,0.0,0.0}
\pgfusepath{stroke}
\end{pgfscope}
\begin{pgfscope}
\pgfsetlinewidth{1.0bp}
\pgfsetrectcap 
\pgfsetmiterjoin \pgfsetmiterlimit{10.0}
\pgfpathmoveto{\pgfpoint{558.092912bp}{106.763077bp}}
\pgfpathcurveto{\pgfpoint{558.092912bp}{103.507931bp}}{\pgfpoint{555.561805bp}{100.869111bp}}{\pgfpoint{552.439521bp}{100.869111bp}}
\pgfpathcurveto{\pgfpoint{549.317236bp}{100.869111bp}}{\pgfpoint{546.786123bp}{103.507931bp}}{\pgfpoint{546.786123bp}{106.763077bp}}
\pgfpathcurveto{\pgfpoint{546.786123bp}{110.018224bp}}{\pgfpoint{549.317236bp}{112.657042bp}}{\pgfpoint{552.439521bp}{112.657042bp}}
\pgfpathcurveto{\pgfpoint{555.561805bp}{112.657042bp}}{\pgfpoint{558.092912bp}{110.018224bp}}{\pgfpoint{558.092912bp}{106.763077bp}}
\pgfclosepath
\color[rgb]{0.0,0.0,0.0}\pgfseteorule\pgfusepath{fill}
\pgfpathmoveto{\pgfpoint{558.092912bp}{106.763077bp}}
\pgfpathcurveto{\pgfpoint{558.092912bp}{103.507931bp}}{\pgfpoint{555.561805bp}{100.869111bp}}{\pgfpoint{552.439521bp}{100.869111bp}}
\pgfpathcurveto{\pgfpoint{549.317236bp}{100.869111bp}}{\pgfpoint{546.786123bp}{103.507931bp}}{\pgfpoint{546.786123bp}{106.763077bp}}
\pgfpathcurveto{\pgfpoint{546.786123bp}{110.018224bp}}{\pgfpoint{549.317236bp}{112.657042bp}}{\pgfpoint{552.439521bp}{112.657042bp}}
\pgfpathcurveto{\pgfpoint{555.561805bp}{112.657042bp}}{\pgfpoint{558.092912bp}{110.018224bp}}{\pgfpoint{558.092912bp}{106.763077bp}}
\pgfclosepath
\color[rgb]{0.0,0.0,0.0}
\pgfusepath{stroke}
\end{pgfscope}
\begin{pgfscope}
\pgfsetlinewidth{1.0bp}
\pgfsetrectcap 
\pgfsetmiterjoin \pgfsetmiterlimit{10.0}
\pgfpathmoveto{\pgfpoint{638.492912bp}{106.763068bp}}
\pgfpathcurveto{\pgfpoint{638.492912bp}{103.507921bp}}{\pgfpoint{635.961805bp}{100.869102bp}}{\pgfpoint{632.839521bp}{100.869102bp}}
\pgfpathcurveto{\pgfpoint{629.717236bp}{100.869102bp}}{\pgfpoint{627.186123bp}{103.507921bp}}{\pgfpoint{627.186123bp}{106.763068bp}}
\pgfpathcurveto{\pgfpoint{627.186123bp}{110.018215bp}}{\pgfpoint{629.717236bp}{112.657033bp}}{\pgfpoint{632.839521bp}{112.657033bp}}
\pgfpathcurveto{\pgfpoint{635.961805bp}{112.657033bp}}{\pgfpoint{638.492912bp}{110.018215bp}}{\pgfpoint{638.492912bp}{106.763068bp}}
\pgfclosepath
\color[rgb]{0.0,0.0,0.0}\pgfseteorule\pgfusepath{fill}
\pgfpathmoveto{\pgfpoint{638.492912bp}{106.763068bp}}
\pgfpathcurveto{\pgfpoint{638.492912bp}{103.507921bp}}{\pgfpoint{635.961805bp}{100.869102bp}}{\pgfpoint{632.839521bp}{100.869102bp}}
\pgfpathcurveto{\pgfpoint{629.717236bp}{100.869102bp}}{\pgfpoint{627.186123bp}{103.507921bp}}{\pgfpoint{627.186123bp}{106.763068bp}}
\pgfpathcurveto{\pgfpoint{627.186123bp}{110.018215bp}}{\pgfpoint{629.717236bp}{112.657033bp}}{\pgfpoint{632.839521bp}{112.657033bp}}
\pgfpathcurveto{\pgfpoint{635.961805bp}{112.657033bp}}{\pgfpoint{638.492912bp}{110.018215bp}}{\pgfpoint{638.492912bp}{106.763068bp}}
\pgfclosepath
\color[rgb]{0.0,0.0,0.0}
\pgfusepath{stroke}
\end{pgfscope}
\begin{pgfscope}
\pgftransformcm{1.0}{-0.0}{0.0}{1.0}{\pgfpoint{57.700006bp}{97.924993bp}}
\pgftext[left,top]{\sffamily\mdseries\upshape\huge
\color[rgb]{0.0,0.0,0.0}$x_1$}
\end{pgfscope}
\begin{pgfscope}
\pgftransformcm{1.0}{-0.0}{0.0}{1.0}{\pgfpoint{138.100006bp}{97.924993bp}}
\pgftext[left,top]{\sffamily\mdseries\upshape\huge
\color[rgb]{0.0,0.0,0.0}$x_{t_1}$}
\end{pgfscope}
\begin{pgfscope}
\pgftransformcm{1.0}{-0.0}{0.0}{1.0}{\pgfpoint{217.700006bp}{98.524993bp}}
\pgftext[left,top]{\sffamily\mdseries\upshape\huge
\color[rgb]{0.0,0.0,0.0}$x_{t_1+1}$}
\end{pgfscope}
\begin{pgfscope}
\pgftransformcm{1.0}{-0.0}{0.0}{1.0}{\pgfpoint{298.100006bp}{98.524993bp}}
\pgftext[left,top]{\sffamily\mdseries\upshape\huge
\color[rgb]{0.0,0.0,0.0}$x_{2t_1}$}
\end{pgfscope}
\begin{pgfscope}
\pgftransformcm{1.0}{-0.0}{0.0}{1.0}{\pgfpoint{381.100006bp}{98.524993bp}}
\pgftext[left,top]{\sffamily\mdseries\upshape\huge
\color[rgb]{0.0,0.0,0.0}$x_{2t_1+1}$}
\end{pgfscope}
\begin{pgfscope}
\pgftransformcm{1.0}{-0.0}{0.0}{1.0}{\pgfpoint{461.500006bp}{98.524993bp}}
\pgftext[left,top]{\sffamily\mdseries\upshape\huge
\color[rgb]{0.0,0.0,0.0}$x_{2t_1+t_2}$}
\end{pgfscope}
\begin{pgfscope}
\pgftransformcm{1.0}{-0.0}{0.0}{1.0}{\pgfpoint{543.900006bp}{98.524993bp}}
\pgftext[left,top]{\sffamily\mdseries\upshape\huge
\color[rgb]{0.0,0.0,0.0}$x_{2t_1+t_2+1}$}
\end{pgfscope}
\begin{pgfscope}
\pgftransformcm{1.0}{-0.0}{0.0}{1.0}{\pgfpoint{625.500006bp}{98.524993bp}}
\pgftext[left,top]{\sffamily\mdseries\upshape\huge
\color[rgb]{0.0,0.0,0.0}$x_{k}$}
\end{pgfscope}
\begin{pgfscope}
\pgftransformcm{1.40625}{-0.0}{0.0}{1.40625}{\pgfpoint{85.770807bp}{108.899994bp}}
\pgftext[left,top]{\sffamily\mdseries\upshape\fontsize{36.135}{36.135}\selectfont
\color[rgb]{0.0,0.0,0.0}$\hdots$}
\end{pgfscope}
\begin{pgfscope}
\pgftransformcm{1.40625}{-0.0}{0.0}{1.40625}{\pgfpoint{246.170807bp}{109.499994bp}}
\pgftext[left,top]{\sffamily\mdseries\upshape\fontsize{36.135}{36.135}\selectfont
\color[rgb]{0.0,0.0,0.0}$\hdots$}
\end{pgfscope}
\begin{pgfscope}
\pgftransformcm{1.40625}{-0.0}{0.0}{1.40625}{\pgfpoint{409.570807bp}{109.499994bp}}
\pgftext[left,top]{\sffamily\mdseries\upshape\fontsize{36.135}{36.135}\selectfont
\color[rgb]{0.0,0.0,0.0}$\hdots$}
\end{pgfscope}
\begin{pgfscope}
\pgftransformcm{1.40625}{-0.0}{0.0}{1.40625}{\pgfpoint{572.370807bp}{109.499994bp}}
\pgftext[left,top]{\sffamily\mdseries\upshape\fontsize{36.135}{36.135}\selectfont
\color[rgb]{0.0,0.0,0.0}$\hdots$}
\end{pgfscope}
\begin{pgfscope}
\pgftransformcm{1.0}{0.0}{0.0}{1.0}{\pgfpoint{32.500006bp}{184.499994bp}}
\pgftext[left,base]{\sffamily\mdseries\upshape\huge
\color[rgb]{0.0,0.0,0.0}$e$}
\end{pgfscope}
\begin{pgfscope}
\pgfsetlinewidth{1.0bp}
\pgfsetrectcap 
\pgfsetmiterjoin \pgfsetmiterlimit{10.0}
\pgfsetdash{{10.0bp}{5.0bp}}{0.0bp}
\pgfpathmoveto{\pgfpoint{48.500006bp}{74.899994bp}}
\pgfpathlineto{\pgfpoint{48.500006bp}{137.299994bp}}
\pgfpathlineto{\pgfpoint{169.100006bp}{137.299994bp}}
\pgfpathlineto{\pgfpoint{169.100006bp}{74.899994bp}}
\pgfpathlineto{\pgfpoint{48.500006bp}{74.899994bp}}
\pgfclosepath
\color[rgb]{0.0,0.0,0.0}
\pgfusepath{stroke}
\end{pgfscope}
\begin{pgfscope}
\pgfsetlinewidth{1.0bp}
\pgfsetrectcap 
\pgfsetmiterjoin \pgfsetmiterlimit{10.0}
\pgfsetdash{{10.0bp}{5.0bp}}{0.0bp}
\pgfpathmoveto{\pgfpoint{210.100006bp}{75.299994bp}}
\pgfpathlineto{\pgfpoint{210.100006bp}{137.699994bp}}
\pgfpathlineto{\pgfpoint{330.700006bp}{137.699994bp}}
\pgfpathlineto{\pgfpoint{330.700006bp}{75.299994bp}}
\pgfpathlineto{\pgfpoint{210.100006bp}{75.299994bp}}
\pgfclosepath
\color[rgb]{0.0,0.0,0.0}
\pgfusepath{stroke}
\end{pgfscope}
\begin{pgfscope}
\pgfsetlinewidth{1.0bp}
\pgfsetrectcap 
\pgfsetmiterjoin \pgfsetmiterlimit{10.0}
\pgfsetdash{{10.0bp}{5.0bp}}{0.0bp}
\pgfpathmoveto{\pgfpoint{372.100006bp}{75.299994bp}}
\pgfpathlineto{\pgfpoint{372.100006bp}{137.699994bp}}
\pgfpathlineto{\pgfpoint{492.700006bp}{137.699994bp}}
\pgfpathlineto{\pgfpoint{492.700006bp}{75.299994bp}}
\pgfpathlineto{\pgfpoint{372.100006bp}{75.299994bp}}
\pgfclosepath
\color[rgb]{0.0,0.0,0.0}
\pgfusepath{stroke}
\end{pgfscope}
\begin{pgfscope}
\pgfsetlinewidth{1.0bp}
\pgfsetrectcap 
\pgfsetmiterjoin \pgfsetmiterlimit{10.0}
\pgfsetdash{{10.0bp}{5.0bp}}{0.0bp}
\pgfpathmoveto{\pgfpoint{537.700006bp}{75.299994bp}}
\pgfpathlineto{\pgfpoint{537.700006bp}{137.699994bp}}
\pgfpathlineto{\pgfpoint{658.300006bp}{137.699994bp}}
\pgfpathlineto{\pgfpoint{658.300006bp}{75.299994bp}}
\pgfpathlineto{\pgfpoint{537.700006bp}{75.299994bp}}
\pgfclosepath
\color[rgb]{0.0,0.0,0.0}
\pgfusepath{stroke}
\end{pgfscope}
\begin{pgfscope}
\pgfsetlinewidth{1.0bp}
\pgfsetrectcap 
\pgfsetmiterjoin \pgfsetmiterlimit{10.0}
\pgfpathmoveto{\pgfpoint{48.500006bp}{152.899994bp}}
\pgfpathlineto{\pgfpoint{169.700006bp}{152.899994bp}}
\pgfpathlineto{\pgfpoint{169.700006bp}{152.899994bp}}
\color[rgb]{0.0,0.0,0.0}
\pgfusepath{stroke}
\end{pgfscope}
{\begin{pgfscope}
\color[rgb]{0.0,0.0,0.0}\pgfpathqmoveto{53.853561bp}{148.25354bp}
\pgfpathqlineto{48.853561bp}{153.25354bp}
\pgfpathqlineto{48.5bp}{152.899994bp}
\pgfpathqlineto{48.853561bp}{152.546448bp}
\pgfpathqlineto{53.853561bp}{157.546448bp}
\pgfpathqlineto{54.207108bp}{157.899994bp}
\pgfpathqlineto{53.5bp}{158.607101bp}
\pgfpathqlineto{53.146454bp}{158.25354bp}
\pgfpathqlineto{48.146454bp}{153.25354bp}
\pgfpathqlineto{47.792892bp}{152.899994bp}
\pgfpathqlineto{48.146454bp}{152.546448bp}
\pgfpathqlineto{53.146454bp}{147.546448bp}
\pgfpathqlineto{53.5bp}{147.192886bp}
\pgfpathqlineto{54.207108bp}{147.899994bp}
\pgfpathqlineto{53.853561bp}{148.25354bp}
\pgfclosepath
\pgfusepathqfill
\end{pgfscope}}
{\begin{pgfscope}
\color[rgb]{0.0,0.0,0.0}\pgfpathqmoveto{164.346466bp}{157.546448bp}
\pgfpathqlineto{169.346466bp}{152.546448bp}
\pgfpathqlineto{169.700012bp}{152.899994bp}
\pgfpathqlineto{169.346466bp}{153.25354bp}
\pgfpathqlineto{164.346466bp}{148.25354bp}
\pgfpathqlineto{163.992889bp}{147.899994bp}
\pgfpathqlineto{164.700012bp}{147.192886bp}
\pgfpathqlineto{165.053558bp}{147.546448bp}
\pgfpathqlineto{170.053558bp}{152.546448bp}
\pgfpathqlineto{170.407104bp}{152.899994bp}
\pgfpathqlineto{170.053558bp}{153.25354bp}
\pgfpathqlineto{165.053558bp}{158.25354bp}
\pgfpathqlineto{164.700012bp}{158.607101bp}
\pgfpathqlineto{163.992889bp}{157.899994bp}
\pgfpathqlineto{164.346466bp}{157.546448bp}
\pgfclosepath
\pgfusepathqfill
\end{pgfscope}}
\begin{pgfscope}
\pgfsetlinewidth{1.0bp}
\pgfsetrectcap 
\pgfsetmiterjoin \pgfsetmiterlimit{10.0}
\pgfpathmoveto{\pgfpoint{210.100006bp}{152.699994bp}}
\pgfpathlineto{\pgfpoint{331.300006bp}{152.699994bp}}
\pgfpathlineto{\pgfpoint{331.300006bp}{152.699994bp}}
\color[rgb]{0.0,0.0,0.0}
\pgfusepath{stroke}
\end{pgfscope}
{\begin{pgfscope}
\color[rgb]{0.0,0.0,0.0}\pgfpathqmoveto{215.453552bp}{148.053543bp}
\pgfpathqlineto{210.453552bp}{153.053543bp}
\pgfpathqlineto{210.100006bp}{152.699997bp}
\pgfpathqlineto{210.453552bp}{152.346436bp}
\pgfpathqlineto{215.453552bp}{157.346436bp}
\pgfpathqlineto{215.807098bp}{157.699997bp}
\pgfpathqlineto{215.100006bp}{158.407104bp}
\pgfpathqlineto{214.74646bp}{158.053543bp}
\pgfpathqlineto{209.74646bp}{153.053543bp}
\pgfpathqlineto{209.392914bp}{152.699997bp}
\pgfpathqlineto{209.74646bp}{152.346436bp}
\pgfpathqlineto{214.74646bp}{147.346436bp}
\pgfpathqlineto{215.100006bp}{146.992889bp}
\pgfpathqlineto{215.807098bp}{147.699997bp}
\pgfpathqlineto{215.453552bp}{148.053543bp}
\pgfclosepath
\pgfusepathqfill
\end{pgfscope}}
{\begin{pgfscope}
\color[rgb]{0.0,0.0,0.0}\pgfpathqmoveto{325.946442bp}{157.346436bp}
\pgfpathqlineto{330.946442bp}{152.346436bp}
\pgfpathqlineto{331.300018bp}{152.699997bp}
\pgfpathqlineto{330.946442bp}{153.053543bp}
\pgfpathqlineto{325.946442bp}{148.053543bp}
\pgfpathqlineto{325.592896bp}{147.699997bp}
\pgfpathqlineto{326.300018bp}{146.992889bp}
\pgfpathqlineto{326.653564bp}{147.346436bp}
\pgfpathqlineto{331.653564bp}{152.346436bp}
\pgfpathqlineto{332.007111bp}{152.699997bp}
\pgfpathqlineto{331.653564bp}{153.053543bp}
\pgfpathqlineto{326.653564bp}{158.053543bp}
\pgfpathqlineto{326.300018bp}{158.407104bp}
\pgfpathqlineto{325.592896bp}{157.699997bp}
\pgfpathqlineto{325.946442bp}{157.346436bp}
\pgfclosepath
\pgfusepathqfill
\end{pgfscope}}
\begin{pgfscope}
\pgfsetlinewidth{1.0bp}
\pgfsetrectcap 
\pgfsetmiterjoin \pgfsetmiterlimit{10.0}
\pgfpathmoveto{\pgfpoint{371.700006bp}{152.499994bp}}
\pgfpathlineto{\pgfpoint{492.900006bp}{152.499994bp}}
\pgfpathlineto{\pgfpoint{492.900006bp}{152.499994bp}}
\color[rgb]{0.0,0.0,0.0}
\pgfusepath{stroke}
\end{pgfscope}
{\begin{pgfscope}
\color[rgb]{0.0,0.0,0.0}\pgfpathqmoveto{377.053558bp}{147.853546bp}
\pgfpathqlineto{372.053558bp}{152.853546bp}
\pgfpathqlineto{371.699982bp}{152.5bp}
\pgfpathqlineto{372.053558bp}{152.146439bp}
\pgfpathqlineto{377.053558bp}{157.146439bp}
\pgfpathqlineto{377.407135bp}{157.5bp}
\pgfpathqlineto{376.699982bp}{158.207108bp}
\pgfpathqlineto{376.346466bp}{157.853546bp}
\pgfpathqlineto{371.346466bp}{152.853546bp}
\pgfpathqlineto{370.992889bp}{152.5bp}
\pgfpathqlineto{371.346466bp}{152.146439bp}
\pgfpathqlineto{376.346466bp}{147.146439bp}
\pgfpathqlineto{376.699982bp}{146.792892bp}
\pgfpathqlineto{377.407135bp}{147.5bp}
\pgfpathqlineto{377.053558bp}{147.853546bp}
\pgfclosepath
\pgfusepathqfill
\end{pgfscope}}
{\begin{pgfscope}
\color[rgb]{0.0,0.0,0.0}\pgfpathqmoveto{487.546478bp}{157.146439bp}
\pgfpathqlineto{492.546478bp}{152.146439bp}
\pgfpathqlineto{492.899994bp}{152.5bp}
\pgfpathqlineto{492.546478bp}{152.853546bp}
\pgfpathqlineto{487.546478bp}{147.853546bp}
\pgfpathqlineto{487.192902bp}{147.5bp}
\pgfpathqlineto{487.899994bp}{146.792892bp}
\pgfpathqlineto{488.253571bp}{147.146439bp}
\pgfpathqlineto{493.253571bp}{152.146439bp}
\pgfpathqlineto{493.607086bp}{152.5bp}
\pgfpathqlineto{493.253571bp}{152.853546bp}
\pgfpathqlineto{488.253571bp}{157.853546bp}
\pgfpathqlineto{487.899994bp}{158.207108bp}
\pgfpathqlineto{487.192902bp}{157.5bp}
\pgfpathqlineto{487.546478bp}{157.146439bp}
\pgfclosepath
\pgfusepathqfill
\end{pgfscope}}
\begin{pgfscope}
\pgfsetlinewidth{1.0bp}
\pgfsetrectcap 
\pgfsetmiterjoin \pgfsetmiterlimit{10.0}
\pgfpathmoveto{\pgfpoint{536.900006bp}{152.899994bp}}
\pgfpathlineto{\pgfpoint{658.100006bp}{152.899994bp}}
\pgfpathlineto{\pgfpoint{658.100006bp}{152.899994bp}}
\color[rgb]{0.0,0.0,0.0}
\pgfusepath{stroke}
\end{pgfscope}
{\begin{pgfscope}
\color[rgb]{0.0,0.0,0.0}\pgfpathqmoveto{542.253571bp}{148.25354bp}
\pgfpathqlineto{537.253571bp}{153.25354bp}
\pgfpathqlineto{536.899994bp}{152.899994bp}
\pgfpathqlineto{537.253571bp}{152.546448bp}
\pgfpathqlineto{542.253571bp}{157.546448bp}
\pgfpathqlineto{542.607086bp}{157.899994bp}
\pgfpathqlineto{541.899994bp}{158.607101bp}
\pgfpathqlineto{541.546478bp}{158.25354bp}
\pgfpathqlineto{536.546478bp}{153.25354bp}
\pgfpathqlineto{536.192902bp}{152.899994bp}
\pgfpathqlineto{536.546478bp}{152.546448bp}
\pgfpathqlineto{541.546478bp}{147.546448bp}
\pgfpathqlineto{541.899994bp}{147.192886bp}
\pgfpathqlineto{542.607086bp}{147.899994bp}
\pgfpathqlineto{542.253571bp}{148.25354bp}
\pgfclosepath
\pgfusepathqfill
\end{pgfscope}}
{\begin{pgfscope}
\color[rgb]{0.0,0.0,0.0}\pgfpathqmoveto{652.746429bp}{157.546448bp}
\pgfpathqlineto{657.746429bp}{152.546448bp}
\pgfpathqlineto{658.100006bp}{152.899994bp}
\pgfpathqlineto{657.746429bp}{153.25354bp}
\pgfpathqlineto{652.746429bp}{148.25354bp}
\pgfpathqlineto{652.392914bp}{147.899994bp}
\pgfpathqlineto{653.100006bp}{147.192886bp}
\pgfpathqlineto{653.453583bp}{147.546448bp}
\pgfpathqlineto{658.453583bp}{152.546448bp}
\pgfpathqlineto{658.807098bp}{152.899994bp}
\pgfpathqlineto{658.453583bp}{153.25354bp}
\pgfpathqlineto{653.453583bp}{158.25354bp}
\pgfpathqlineto{653.100006bp}{158.607101bp}
\pgfpathqlineto{652.392914bp}{157.899994bp}
\pgfpathqlineto{652.746429bp}{157.546448bp}
\pgfclosepath
\pgfusepathqfill
\end{pgfscope}}
\begin{pgfscope}
\pgftransformcm{1.0}{-0.0}{0.0}{1.0}{\pgfpoint{102.300006bp}{172.724993bp}}
\pgftext[left,top]{\sffamily\mdseries\upshape\huge
\color[rgb]{0.0,0.0,0.0}$t_1$}
\end{pgfscope}
\begin{pgfscope}
\pgftransformcm{1.0}{-0.0}{0.0}{1.0}{\pgfpoint{263.300006bp}{172.524993bp}}
\pgftext[left,top]{\sffamily\mdseries\upshape\huge
\color[rgb]{0.0,0.0,0.0}$t_1$}
\end{pgfscope}
\begin{pgfscope}
\pgftransformcm{1.0}{-0.0}{0.0}{1.0}{\pgfpoint{424.900006bp}{172.324993bp}}
\pgftext[left,top]{\sffamily\mdseries\upshape\huge
\color[rgb]{0.0,0.0,0.0}$t_2$}
\end{pgfscope}
\begin{pgfscope}
\pgftransformcm{1.0}{-0.0}{0.0}{1.0}{\pgfpoint{591.300006bp}{172.724993bp}}
\pgftext[left,top]{\sffamily\mdseries\upshape\huge
\color[rgb]{0.0,0.0,0.0}$t_3$}
\end{pgfscope}
\begin{pgfscope}
\pgftransformcm{1.0}{-0.0}{0.0}{1.0}{\pgfpoint{103.700006bp}{61.324993bp}}
\pgftext[left,bottom]{\sffamily\mdseries\upshape\huge
\color[rgb]{0.0,0.0,0.0}$=$}
\end{pgfscope}
\begin{pgfscope}
\pgftransformcm{1.0}{-0.0}{0.0}{1.0}{\pgfpoint{266.500006bp}{61.324993bp}}
\pgftext[left,bottom]{\sffamily\mdseries\upshape\huge
\color[rgb]{0.0,0.0,0.0}$=$}
\end{pgfscope}
\begin{pgfscope}
\pgftransformcm{1.0}{-0.0}{0.0}{1.0}{\pgfpoint{429.700006bp}{61.324993bp}}
\pgftext[left,bottom]{\sffamily\mdseries\upshape\huge
\color[rgb]{0.0,0.0,0.0}$=$}
\end{pgfscope}
\begin{pgfscope}
\pgftransformcm{1.0}{-0.0}{0.0}{1.0}{\pgfpoint{594.700006bp}{61.324993bp}}
\pgftext[left,bottom]{\sffamily\mdseries\upshape\huge
\color[rgb]{0.0,0.0,0.0}$=$}
\end{pgfscope}
\begin{pgfscope}
\pgfsetlinewidth{1.0bp}
\pgfsetrectcap 
\pgfsetmiterjoin \pgfsetmiterlimit{10.0}
\pgfpathmoveto{\pgfpoint{707.300006bp}{112.099992bp}}
\pgfpathcurveto{\pgfpoint{707.300006bp}{50.465021bp}}{\pgfpoint{549.077411bp}{0.499994bp}}{\pgfpoint{353.899982bp}{0.499994bp}}
\pgfpathcurveto{\pgfpoint{158.722583bp}{0.499994bp}}{\pgfpoint{0.500003bp}{50.465021bp}}{\pgfpoint{0.500003bp}{112.099992bp}}
\pgfpathcurveto{\pgfpoint{0.500003bp}{173.734972bp}}{\pgfpoint{158.722583bp}{223.699994bp}}{\pgfpoint{353.899982bp}{223.699994bp}}
\pgfpathcurveto{\pgfpoint{549.077411bp}{223.699994bp}}{\pgfpoint{707.300006bp}{173.734972bp}}{\pgfpoint{707.300006bp}{112.099992bp}}
\pgfclosepath
\color[rgb]{0.0,0.0,0.0}
\pgfusepath{stroke}
\end{pgfscope}
\end{pgfpicture}